\documentclass[11pt]{article}

\newif\ifiwoca
\iwocafalse

\usepackage[english]{babel}
\usepackage{Macros}

\title{Exploiting Low-Rank Objective Structure \\ in Discrete Quadratic Optimization}

\author{
    Ria Stevens \\ Rice University \\ ria.stevens@rice.edu
    \and
    Fangshuo Liao \\ Rice University \\ fangshuo.liao@rice.edu
    \and
    Barbara Su \\ Rice University \\ barbara.su@rice.edu
    \and
    Thanasis Hadjidimoulas \\ Rice University \\ th81@rice.edu
    \and
    Jianqiang Li \\ Rice University \\ jl567@rice.edu
    \and
    Anastasios Kyrillidis \\ Rice University \\ anastasios@rice.edu
}
\date{}

\begin{document}
\maketitle

\begin{abstract}
    
We study the problem of maximizing a complex-valued quadratic form over the $K^{\text{th}}$ roots of unity.
We show that when the objective matrix $\bfQ^\star \in \bbC^{n \times n}$ of the quadratic has rank~$r$, the global maximizer belongs to a candidate set of size $\OO(rn^{2r-1})$. This set can be constructed deterministically in $\OO(rn^{2r+1})$ time by enumerating the vertices of a hyperplane arrangement in $\bbR^{2r}.$ 
The algorithm is embarrassingly parallel; with~$P$ processors, the time complexity drops to $\OO(r n^{2r+1}/P)$.
For approximately low-rank settings, where the objective matrix is a noise-perturbed variant of a rank-$r$ matrix, we prove that applying our framework to a spectral truncation yields a multiplicative $(1 - \OO(\norm{\bfH}_2 / \delta^{\star}))$-approximation guarantee, where $\delta^{\star}$ denotes the eigengap of the underlying rank-$r$ matrix and $\bfH$ represents the perturbation. To scale to high-dimensional problems, we establish a randomized sampling variant. We prove that uniformly sampling $S \geq \OO(1/\varepsilon^{r-1})$ candidates achieves a $(1-\varepsilon)\cos^2(\pi / K)$-approximation of the optimal rank-$r$ solution with high probability. Crucially, this sample size is entirely independent of $n$, reducing the overall runtime to $\OO(S \cdot n^2)$. Computational experiments on synthetic benchmarks and large-scale graphs for \textsc{Max-3-Cut} confirm that our algorithms match or exceed semi-definite programming solution quality on structured instances while enabling massive parallelization across heterogeneous hardware and scaling seamlessly to problems where $n \geq 10^6$.

\end{abstract}

\clearpage

\section{Introduction}

Discrete quadratic optimization captures a wide class of fundamental problems spanning combinatorial optimization, signal processing and statistical physics. The goal is to maximize a quadratic objective over a set of discrete variables, such as the Boolean hypercube or the discrete roots of unity. Because this framework includes NP-hard problems such as \textsc{Max-Cut}, a substantial body of work has focused on polynomial-time approximation algorithms. In this work, we show that low-rank structure in the objective matrix enables exact polynomial-time algorithms.

Given a Hermitian positive semi-definite (PSD) matrix $\bfQ^\star \in \bbC^{n \times n}$, we consider the discrete complex quadratic maximization problem:
\begin{equation}
\label{eq:max_quad}
 \optstar := \max_{\bfz \in \AK^n} \bfz^\dagger \bfQ^\star \bfz,
\end{equation}
where for an integer $K \geq 2$, $\AK$ represents the discrete alphabet of $K$-th roots of unity:
\begin{equation}
\label{eq:AK_def}
    \AK := \bigl\{e^{2\pi\imag k/K} : k = 0, 1, \ldots, K-1\bigr\} \subset \bbC.
\end{equation}

When $K=2$ and $\bfQ^\star$ is a real matrix, this problem is known as quadratic unconstrained binary optimization (QUBO) and has applications in graph partitioning \cite{goemans1995improved}, computational biology \cite{forrester2008quadratic}, and machine learning \cite{date2021qubo}, among others \cite{punnen2022quadratic}. Instances of \eqref{eq:max_quad} of particular interest in theoretical computer science are \textsc{Max-Cut} and \textsc{Max-3-Cut}. Specifically, when $K=2$ or $3$ and $\bfQ^\star$ is a graph Laplacian, solving \eqref{eq:max_quad} is equivalent to solving \textsc{Max-Cut} and \textsc{Max-3-Cut}, respectively. The formulation also captures phase estimation problems in signal processing. In $K$-ary Phase Shift Keying Multiple-Input Multiple-Output (MIMO) detection, decoding a vector of transmitted phases $\bfz \in \mathcal{A}_K^n$ under maximum-likelihood criteria reduces to a complex discrete quadratic maximization. Similarly, in angular synchronization, one aims to estimate $n$ unknown phases given noisy pairwise relative phase measurements, which amounts to maximizing \eqref{eq:max_quad} when $\bfQ^\star$ is the complex Hermitian measurement matrix \cite{singer2011angular}.

Because finding an exact solution is NP-hard, standard approaches to \eqref{eq:max_quad} often rely on semidefinite programming (SDP) relaxations. In a seminal work, Goemans and Williamson~\cite{goemans1995improved} proposed an SDP relaxation of \textsc{Max-Cut} which, together with a random-hyperplane rounding technique, guarantees a $0.878$ approximation ratio. Under the Unique Games Conjecture, this ratio is optimal \cite{khot2007optimal}. Frieze and Jerrum~\cite{frieze1997improved} generalized the SDP-based approach to Max-$K$-Cut, proving a $0.800217$ ratio for $K=3$. A few years later, Goemans and Williamson~\cite{goemans2001approximation} introduced a complex SDP relaxation achieving $0.836008$ for $K=3$.
So, Zhang, and Ye~\cite{so2007approximating} extended this approach to general formulations of \eqref{eq:max_quad} involving any Hermitian, PSD $\bfQ^\star$, achieving an approximation ratio of $K^2 \sin^2\left(\pi / K\right) / {4\pi}$. While these SDP formulations offer strong worst-case guarantees, their $\OO(n^{3.5})$ computational complexity limits their scalability in high-dimensional settings. 
Other methods include heuristic algorithms, which often scale better than SDPs but are sensitive to initialization and may get stuck in local optima \cite{panxing2016genetic,ma2017moh,gui2018bqp}, and learning-based algorithms, which scale well but fail to provide rigorous approximation guarantees \cite{abe2019solving,tonshoff2022one,schuetz2022combinatorial,barrett2022learning,qiu2024ros}. Recent work has developed exact solvers for \eqref{eq:max_quad} via sum-of-squares hierarchies, though these approaches remain exponentially expensive for large networks~\cite{al2025exact}.

In this work, we take a different approach and show that \emph{low-rank structure in the objective matrix} can be exploited to construct exact polynomial-time algorithms. To motivate the use of low-rank structure, consider binary ($K=2$) quadratic optimization with a rank-$1$ objective matrix $\bfQ^{\star} = \lambda \bfq\bfq^{\top}$. In this case, the objective in \eqref{eq:max_quad} can be rewritten as $\bfz^{\top}\bfQ^{\star}\bfz = \lambda(\bfz^{\top}\bfq)^2$ and the maximizer is simply $\bfz^\star = \mathrm{sign}(\bfq)$. Given $\bfq$, finding this solution requires only linear time. 

Unfortunately, this technique does not readily extend to complex $\bfQ^{\star}$ or $K > 2$. For $K = 3$, for example, each $z_i$ has three choices instead of two, and the sign function has no natural ternary analogue. Our algorithms overcome this barrier for general $K$ and arbitrary rank~$r$, at the cost of increasing the set of maximizers to search over from $1$ to $\OO( rn^{2r-1})$. Our key observation is geometric: for a rank-$r$ objective, the assignment space can be partitioned by an arrangement of $\OO(nK)$ hyperplanes in $\RR^{2r}$, with the global maximizer corresponding to a vertex of this arrangement.
Enumerating these vertices yields a candidate set of polynomial size that is guaranteed to contain the optimal solution.

This geometric approach allows us to handle objective matrices exhibiting either \emph{exact} or \emph{approximate} low-rank structure. Specifically, we assume that either $\bfQ^\star$ has exact rank $r \ll n$, or that the objective matrix is a perturbation $\bfQ = \bfQ^{\star} + \bfH$, where $\bfQ^{\star}$ is strictly rank-$r$ and $\bfH$ represents an unstructured additive noise or residual matrix. These exact and approximate low-rank objective matrices appear in problems such as limited-feedback MIMO beamforming problems \cite{zheng2007mimo, leung2010optimal} and noncoherent sequence detection \cite{papailiopoulos2010maximum}. 

\subsection{Contributions}

Our contributions are as follows.
\begin{enumerate}
    \item \emph{Exact algorithms for rank-$r$ matrices.} We provide a deterministic algorithm (\Cref{alg:rankr}) relying on candidate enumeration over a hyperplane arrangement that \textit{exactly} solves \eqref{eq:max_quad} for any $K \geq 2$ in $\OO\left( rn^{2r+1
    }\right)$ time when $\bfQ^\star$ has rank $r \geq 2$ (\Cref{thm:rankr}). When $\bfQ^\star$ has rank $1$, this algorithm runs in $\OO\left(n^2 \right)$ time (\Cref{thm:rank1}). These algorithms are embarrassingly parallel, with an $\OO\left( rn^{2r+1
    } / P \right)$ runtime across $P$ processors.
    \item \emph{Approximation guarantees under perturbation.} Real-world matrices are rarely exactly low-rank. To address this, we prove in \Cref{thm:mult} under a mild incoherence assumption that, when the observed matrix $\bfQ = \bfQ^\star + \bfH$ is the perturbation of a rank-$r$ signal, applying our algorithm to a rank-$r$ truncation of $\bfQ$ achieves a multiplicative approximation ratio of $1-\OO\left(\norm{\bfH}_2 / \delta^\star \right)$, where $\delta^\star$ is the eigengap of $\bfQ^\star$.
    \item \emph{Randomized candidate sampling.}
        Since exhaustive enumeration costs $n^{\Theta(r)}$ time, we analyze a randomized
        alternative: draw $S$ uniform directions on the complex unit sphere and keep the
        best rounded assignment. In \Cref{thm:rr_tail}, we prove that this strategy yields a $(1-\varepsilon)\cos^2(\pi/K)$-approximation with high probability under \Cref{asm:general_position}, where the required sample size $S$ is entirely independent of $n$.
    \item \emph{Empirical validation.} We evaluate the practical performance of our algorithms in solving \textsc{Max-3-Cut}. Through extensive experiments on structured and unstructured graphs, we show that our low-rank algorithms achieve performance competitive with state-of-the-art heuristic algorithms. Further, we demonstrate that the randomized alternatives to our deterministic algorithms can find near-optimal solutions with little to no dependence on $n$ in the size of the sample set.
\end{enumerate}

\subsection{Related Works}

In addition to the aforementioned literature, a rich body of work has explored semidefinite programming (SDP) relaxations for discrete quadratic optimization problems~\cite{andersson2001new,goemans2001approximation,burer2002rank,burer2003nonlinear,zhang2006complex,newman2018complex}. 
The concept of low-rank structure in discrete optimization has traditionally been leveraged via the Burer-Monteiro matrix factorization technique~\cite{alizadeh1997complementarity,barvinok1995problems,pataki1998rank,burer2002rank,burer2003nonlinear,park2017non}. This method enforces low-rankness by factorizing the full-dimensional SDP matrix variable, optimizing over a significantly smaller non-convex search space.
Another parallel line of research analyzes \textsc{Max-Cut} algorithms specifically on low-threshold rank graphs, which are characterized by adjacency matrices possessing only a small number of large eigenvalues~\cite{barak2011rounding, arora2015subexponential,oveis2015regularity}.
Our framework differs fundamentally from these paradigms: instead of modifying the optimization variable or constraining the adjacency spectrum, we directly exploit instances with low-rank \emph{Laplacians}.

More closely aligned with our efforts, discrete quadratic maximization involving structured low-rank objective matrices has been investigated in~\cite{ferrez2005solving,leung2010optimal,kyrillidis2011rank,kyrillidis2014fixed}. Among these, the combinatorial framework proposed by Kyrillidis and Karystinos~\cite{kyrillidis2014fixed} is the most conceptually related to our own. However, their approach lacks generality as it is restricted to specific settings and does not scale to arbitrary alphabet sizes $K$. Furthermore, they do not characterize the theoretical approximation guarantees of their algorithm when subjected to approximately low-rank or noise-perturbed matrices, nor do they consider accelerating search space exploration via randomized candidate sampling. Consequently, our work provides both a significant algorithmic generalization and the first robust theoretical framework for handling approximate low-rank systems.

\section{Preliminaries}

\subsection{Notation}
Vectors are bold lowercase ($\bfz$), matrices bold uppercase ($\bfQ$).
The conjugate transpose is $(\cdot)^{\dagger}$; for real vectors, $(\cdot)^{\dagger} = (\cdot)^{\top}$.
The $i$-th row and $j$-th column of a matrix $\mathbf{A}$ are denoted $\mathbf{A}_{i,:}$ and $\mathbf{A}_{:,j}$, respectively; the $i$-th entry of a vector $\bfz$ is~$z_i$.
For a positive integer $n$, we write $[n] = \{1, \ldots, n\}$.
$\Repart(\cdot)$ and $\Impart(\cdot)$ denote real and imaginary parts; $\arg(\cdot)$ denotes the complex argument (computed via the two-argument arctangent). We write $\imag \coloneqq \sqrt{-1}$. We denote by $\langle \cdot, \cdot \rangle$ the standard Hermitian
inner product on $\bbC^r$:
for $\mathbf{u}, \mathbf{v} \in \bbC^r$, $\langle \mathbf{u}, \mathbf{v} \rangle = \mathbf{u}^\dagger \mathbf{v}$.
$\norm{\cdot}_2$ is the spectral norm for matrices and the Euclidean norm for vectors.
$\norm{\cdot}_{\infty}$ is the element-wise maximum absolute value.


\subsection{Reformulation as a double maximization}
\label{sec:reformDoubleMax}

For a rank-$r$, Hermitian PSD objective matrix $\bfQ_r$, we draw from \cite{kyrillidis2014fixed} and reformulate \eqref{eq:max_quad} as a double maximization problem over auxiliary angular variables and coordinate assignments. 
Using the low-rank factorization $\bfQ_r = \bfV \bfV^\dagger$ for $\bfV \in \bbC^{n \times r}$, the quadratic form can be rewritten as the maximization of the squared norm of a complex vector--matrix product:
\begin{equation}
\label{eqn:quad_max-norm}
    \max_{\bfz \in \AK^n} \bfz^{\dagger}\bfQ_r\bfz = \max_{\bfz \in \AK^n} \norm{\bfV^{\dagger}\bfz}_2^2 \propto \max_{\bfz \in \AK^n} \norm{\bfV^{\dagger}\bfz}_2.
\end{equation}
The norm of $\bfV^{\dagger}\bfz$ can be recovered by selecting the direction that aligns it with the real axis. Specifically, by Cauchy--Schwarz,
for any unit vector $\bfc \in \bbC^r$, $\Repart(\bfz^{\dagger}\bfV\bfc) \leq \norm{\bfV^{\dagger}\bfz}_2$,
with equality when $\bfc = \bfV^{\dagger}\bfz / \norm{\bfV^{\dagger}\bfz}_2$. Thus, we can rewrite \eqref{eqn:quad_max-norm} as:
 \begin{equation}
    \max_{\bfz \in \AK^n} \norm{\bfV^{\dagger}\bfz}_2 = \max_{\bfz \in \AK^n}\max_{\bfc:\norm{\bfc}_2 = 1} \Repart(\bfz^{\dagger}\bfV\bfc).
\end{equation}
To search over all possible $\bfc$, we parametrize the complex unit sphere with {hyperpolar coordinates}. In particular, we define the hypercube 
\begin{equation}
\label{eq:hypercube}
    \HH_r := \left(-\frac{\pi}{2}, \frac{\pi}{2}\right]^{2(r-1)} \times \left(-\frac{\pi}{K}, \frac{\pi}{K}\right].
\end{equation}
Then, every unit vector $\bfc \in \bbC^r$ may be uniquely expressed as a function of some $\bfphi \in \HH_r$ as: 
\begin{equation}
\label{eq:hyperpolar_def}
    c_j(\bfphi) = P_{2(j-1)}\bigl(\cos\phi_{2j-1}\sin\phi_{2j} + \imag\sin\phi_{2j-1}\bigr), \quad j = 1, \ldots, r{-}1, \quad {c_r(\bfphi) = P_{2(r-1)}\,e^{\imag\phi_{2r-1}},}
\end{equation}
where $P_k = \prod_{i=1}^{k}\cos\phi_i$ and $P_0 = 1$.
For $r = 1$, this reduces to a single phase $\varphi \in (-\pi/K, \pi/K]$ with $c(\varphi) = e^{\imag\varphi}$. Given this parametrization, we can solve \eqref{eq:max_quad} by solving:
\begin{equation}
\label{eq:rankr_double}
    \max_{\bfphi \in \HH_r}  \sum_{i=1}^{n}  \max_{z_i \in \AK}  \Repart \bigl( \bar{z}_i \cdot \bfV_{i,:}\,\bfc(\bfphi) \bigr).
\end{equation}
We also write $\tilde{\bfc}(\bfphi) \in \bbR^{2r}$ for the real representation of $\bfc$, $(\Repart\{\bfc\}^{\top}, \Impart\{\bfc\}^{\top})^{\top}$. By construction, $\norm{\tilde{\bfc}(\bfphi)}_2 = 1$.


\subsection{Decision boundaries}
Our algorithms work by introducing auxiliary angular variables and analyzing when the optimal assignment of each coordinate $z_i$ changes as these angles vary.
The elements of $\AK$ partition the unit circle into $K$ sectors; the \emph{decision boundaries} are the $K$ rays that bisect adjacent sectors, located at angles $\vartheta_m = \frac{\pi(2m+1)}{K}$ for $m = 0, 1, \ldots, K{-}1$.
The assignment $z_i$ changes precisely when the phase of the corresponding spectral projection crosses one of these rays.
\textcolor{black}{However, not all $K$ boundary rays are geometrically distinct in the auxiliary parameter space. The following definition names the reduced count; the justification follows immediately after and is formalized in \Cref{lem:antipodal}.}
\begin{definition}[Effective boundary count]
\label{def:BK}
For alphabet size $K \geq 2$, the number of geometrically distinct decision boundaries per coordinate is:
\begin{equation}
\label{eq:BK}
B_K = \begin{cases} K/2 & \text{if } K \text{ is even,} \\ K & \text{if } K \text{ is odd.} \end{cases}
\end{equation}
\end{definition}

\noindent
The distinction arises from antipodal symmetry: for even~$K$, the element $-1 \in \AK$, causing boundary rays at angles $\vartheta_m$ and $\vartheta_m + \pi$ to define the same hyperplane in the auxiliary parameter space (\Cref{lem:antipodal}).
For odd $K$ (including the primary case $K = 3$), no two boundary angles differ by exactly $\pi$, so all $K$ boundaries are distinct.


\medskip
\noindent \textbf{Augmented matrix.}
The decision boundary conditions for all $n$ coordinates and all $B_K$ boundaries can be expressed as a single real linear system.
Define rotation factors $\rho_m = e^{-\imag \pi(2m+1)/K}$ for $m = 0, \ldots, B_K{-}1$ and construct the augmented real matrix $\tilde{\bfV} \in \bbR^{nB_K \times 2r}$ by stacking $B_K$ rotated copies of the factor matrix:
\begin{equation}
\label{eq:tildeV}
\tilde{\bfV} = \begin{bmatrix}
\Repart\{\rho_0 \bfV\} & \Impart\{\rho_0 \bfV\} \\
\Repart\{\rho_1 \bfV\} & \Impart\{\rho_1 \bfV\} \\
\vdots & \vdots \\
\Repart\{\rho_{B_K-1} \bfV\} & \Impart\{\rho_{B_K-1} \bfV\}
\end{bmatrix}.
\end{equation}
Each row $\tilde{\bfV}_{j,:}$ defines a hyperplane $\{\bfphi \in \HH_r : \tilde{\bfV}_{j,:}\,\tilde{\bfc}(\bfphi) = 0\}$ in the auxiliary parameter space.
The $j$-th row corresponds to the $m$-th decision boundary of coordinate $i$, where $j = (m-1)n + i$.
The full set of $nB_K$ hyperplanes forms a \emph{hyperplane arrangement} that partitions $\HH_r$ into cells; within each cell, the optimal assignment vector $\bfz$ is constant.
We organize the hyperplanes into $n$ \emph{coordinate groups} $\GG^{(1)}, \ldots, \GG^{(n)}$, where $\GG^{(i)}$ contains the $B_K$ hyperplanes arising from vertex~$i$.

\subsection{Assumptions}


\begin{assumption}[General Position of $\bfV$ and $\Vtilde$]
\label{asm:general_position}
The factor matrix $\bfV \in \bbC^{n \times r}$ and its augmented real embedding $\Vtilde \in \bbR^{nB_K \times 2r}$ (defined in \eqref{eq:tildeV}) satisfy the following conditions:
\begin{enumerate}
    \item The rows of $\bfV$ are in general linear position over $\bbC$. That is, every $r$-subset of rows is linearly independent over $\bbC$, and no row is the zero vector.
    \item Every $(2r-1)$-subset of rows of $\Vtilde$ that contains \emph{at most two} rows from any single group $i \in [n]$ is linearly independent over $\bbR$.
\end{enumerate}
\end{assumption}



\begin{assumption}[$\mu$-incoherence]
\label{asm:incoherence}
For all $i \in [n]$,
\begin{equation}
\label{eq:incoherence}
 \norm{\bfV_{i,:}}_2^2 \leq \frac{\mu^2}{n}\sum_{j=1}^{r}\lambda_j,
\end{equation}
where $\mu \in [1, \sqrt{n}]$ is known as the {incoherence parameter}.
\end{assumption}

\subsection{Warm-Up: Exact Rank-1 Algorithm}
\label{sec:rank1}

In this section, we focus on the case where $\bfQ^\star$ is rank-$1$ and detail the algorithm for solving~\eqref{eq:max_quad} as a warm-up for the general case. We begin by reformulating the problem. As derived in \Cref{sec:reformDoubleMax}, for a rank-$1$ PSD matrix $\bfQ^\star = \lambda \bfq \bfq^\dagger$, our objective~\eqref{eq:max_quad} can be expressed as the double maximization problem~\eqref{eq:rank1_double}. This new formulation requires only a single auxiliary angular variable $\varphi \in (-\pi/K, \pi/K]$ and relies on the unit vector $c(\varphi) = e^{\imag \varphi}$:

\begin{equation}
    \label{eq:rank1_double}
    \max_{\bfz \in \AK^n}  \abs{\bfz^{\dagger}\bfq}^2
    = \max_{\varphi \in \left(-\pi/K, \pi/K \right]}  \sum_{i=1}^{n} \max_{z_i \in \AK}  \Repart \left( \bar{z}_i \cdot q_i \cdot c(\varphi) \right).
\end{equation}


The above expression can be optimized by choosing each $z_i$ independently for a fixed $\varphi$: the optimal $z_i$ is the element of $\AK$ closest in phase to $q_i \cdot c(\varphi)$.
Then, as $\varphi$ varies continuously, the optimal $z_i$ remains piecewise constant. It changes only when the phase of $q_i \cdot c(\varphi)$ crosses a decision boundary between two adjacent sectors of $\AK$.
Over the interval $\varphi \in \left(-\pi/K, \pi/K \right]$, each coordinate $q_i$ crosses exactly one boundary point $\varphi_i$.

\medskip
\noindent \textbf{Computing the boundary points.}
Writing $\theta_i = \arg(q_i)$ for the phase of each eigenvector component, the optimal assignment of coordinate~$i$ at angle $\varphi$ is $z_i = \exp(2\pi\imag k_i / K)$ where $k_i = \floor{K(\theta_i + \varphi)/(2\pi) \rev{{}+ 1/2}}$.
This assignment changes when $\theta_i + \varphi$ crosses a sector boundary, which occurs at
$\varphi_i = \frac{2\pi}{K}\bigl(\frac{1}{2} + \floor{K\theta_i/(2\pi)} - K\theta_i/(2\pi)\bigr)$.
Sorting the $n$ boundary points $\varphi_{i_1} \leq \cdots \leq \varphi_{i_n}$ partitions the sweep into $n+1$ cells, each with a constant assignment vector.
\Cref{alg:rank1} constructs all $n+1$ candidates by starting with the assignment at $\varphi = -\pi/K$ and flipping one coordinate at each boundary point.

\begin{algorithm}[t]
\caption{Exact solver of~\eqref{eq:max_quad} for a rank-1 matrix.}
\label{alg:rank1}
\begin{algorithmic}[1]
\STATE \textbf{Input:} PSD matrix $\bfQ^{\star} \in \bbC^{n \times n}$ with eigenvector $\bfq \in \bbC^n$; alphabet size~$K$.
\STATE \textbf{Output:} Optimal $\bfz^{\star} \in \AK^n$.
\STATE $\boldsymbol{\theta} \leftarrow \arg(\bfq) \in \bbR^n$ \hfill\COMMENT{phase of each $q_i$}
\STATE $\boldsymbol{\varphi} \leftarrow \frac{2\pi}{K} \paren{\frac{1}{2} + \floor{\frac{K\boldsymbol{\theta}}{2\pi}} - \frac{K\boldsymbol{\theta}}{2\pi}} \in \bbR^n$ \hfill\COMMENT{boundary points}
\STATE $i_1, \ldots, i_n \leftarrow \text{Argsort}(\boldsymbol{\varphi})$
\STATE $\bfk \leftarrow \floor{K\boldsymbol{\theta}/(2\pi)}$
\STATE $\mathcal{Z} \leftarrow \bigl\{ \exp(2\pi\imag\,\bfk / K) \bigr\}$
\FOR{$\ell = 1, \ldots, n$}
  \STATE $\bfk \leftarrow \bfk + \bfe_{i_\ell} \bmod K$ \hfill\COMMENT{flip node $i_\ell$ across its boundary}
  \STATE $\mathcal{Z} \leftarrow \mathcal{Z} \cup \bigl\{ \exp(2\pi\imag\,\bfk / K) \bigr\}$
\ENDFOR
\RETURN $\argmax_{\bfz \in \mathcal{Z}}  \bfz^{\dagger}\bfQ^{\star}\bfz$
\end{algorithmic}
\end{algorithm}

\begin{theorem}[Rank-1 exactness]
\label{thm:rank1}
Let $\bfQ^{\star} = \lambda \bfq\bfq^{\dagger} \in \bbC^{n \times n}$ be rank-1 PSD.
\Cref{alg:rank1} solves~\eqref{eq:max_quad} exactly in $\OO(n^2)$ time for any $K \geq 2$.
\end{theorem}

\begin{proof}[Proof sketch]
As $\varphi$ sweeps $(-\pi/K, \pi/K]$, each coordinate's optimal assignment is piecewise constant with exactly one transition.
The $n$ boundary points $\varphi_{i_1} \leq \cdots \leq \varphi_{i_n}$ partition the domain into $n+1$ cells such that within each cell the assignment vector~$\bfz(\varphi)$ is constant.
Since the maximizer lies in some cell, it belongs to the candidate set~$\mathcal{Z}$.
Computing the eigenvector costs $\OO(n^2)$, sorting costs $\OO(n \log n)$, and evaluating all $n+1$ candidates costs $\OO(n^2)$.
The full proof is provided in {\Cref{app:rank1_proof}}. 
\end{proof}


\section{Exact Rank-$r$ Algorithm}
\label{sec:rankr}
We now generalize the rank-1 approach to an arbitrary rank~$r$.
The key idea is the same---we consider the double maximization over auxiliary angles, then enumerate cells of the resulting partition---but the geometry is richer: the search space is an $(2r{-}1)$-dimensional hypercube $\HH_r$, and the cells are defined by an arrangement of $nB_K$ hyperplanes in $\bbR^{2r}$.
Within each cell, the decision function $\bfd(\bfV;\bfphi) \in \AK^n$ is constant.
The candidate set $\mathcal{S}(\bfV)$ is defined as the set of all distinct assignment vectors arising from these cells. Recall the double maximization problem~\eqref{eq:rankr_double}:
\begin{equation}
    \max_{\bfz \in \AK^n} \norm{\bfV^{\dagger}\bfz}_2
    = \max_{\bfphi \in \HH_r}  \sum_{i=1}^{n}  \max_{z_i \in \AK}  \Repart \bigl( \bar{z}_i \cdot \bfV_ {i,:}\,\bfc(\bfphi) \bigr).
\end{equation}

\medskip
\noindent \textbf{Decision boundaries and hyperplane arrangement.}
The optimal assignment $z_i$ changes when the complex number $\alpha_i(\bfphi) := \bfV_{i,:}\,\bfc(\bfphi)$ crosses a decision boundary.
For each coordinate $i \in [n]$, there are $B_K$ geometrically distinct boundaries (see \Cref{def:BK}), one for each bisector between adjacent roots of unity.
In total, $nB_K$ hyperplanes
\begin{equation}
\label{eq:hyperplanes}
\bigl\{\bfphi \in \HH_r : \tilde{\bfV}_{j,:}\,\tilde{\bfc}(\bfphi) = 0\bigr\}, \quad j = 1, \ldots, nB_K,
\end{equation}
partition $\HH_r$ into cells.

\begin{figure}[t]
\centering
\includegraphics[width=0.42\textwidth]{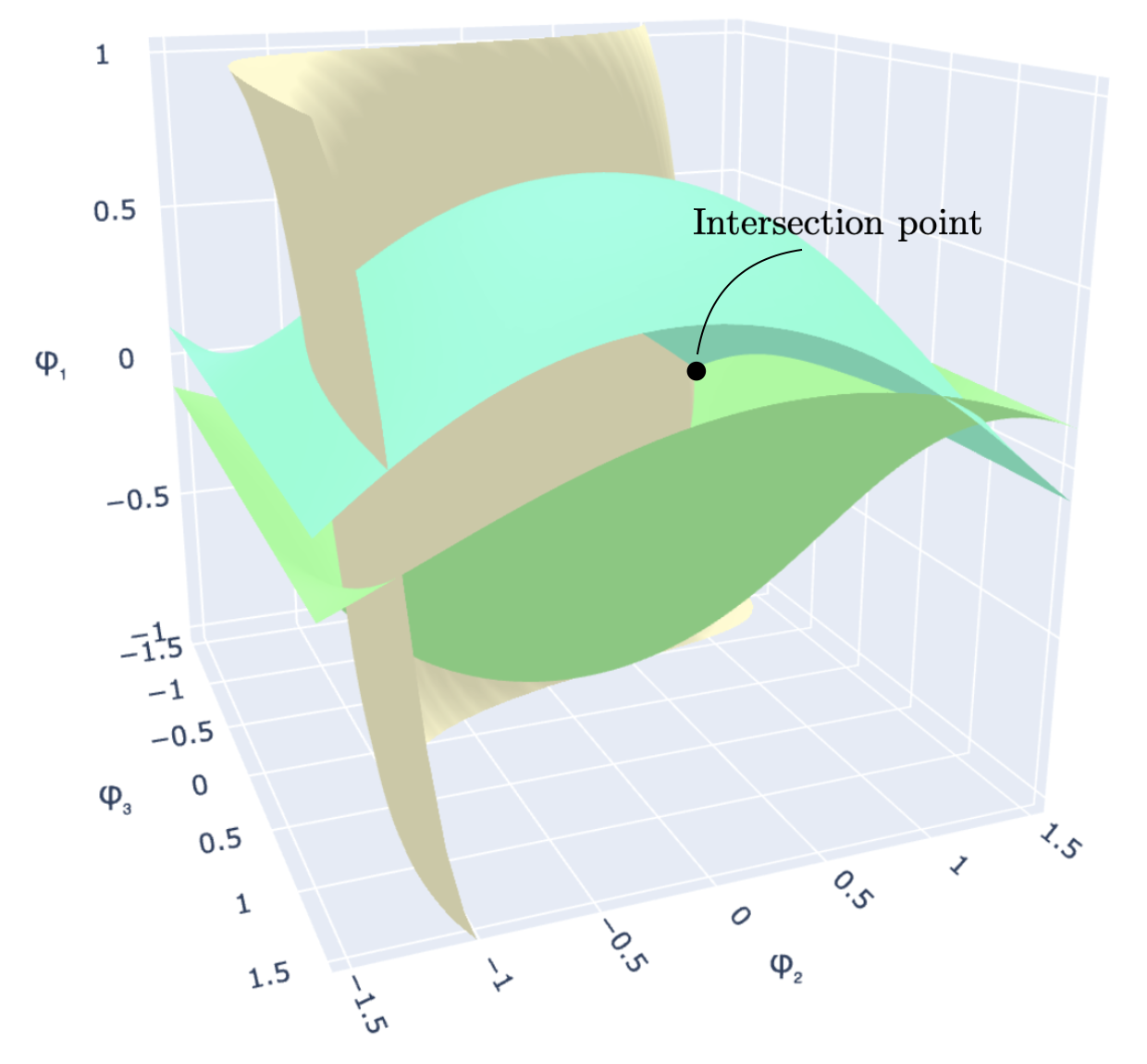}
\hfill
\includegraphics[width=0.42\textwidth]{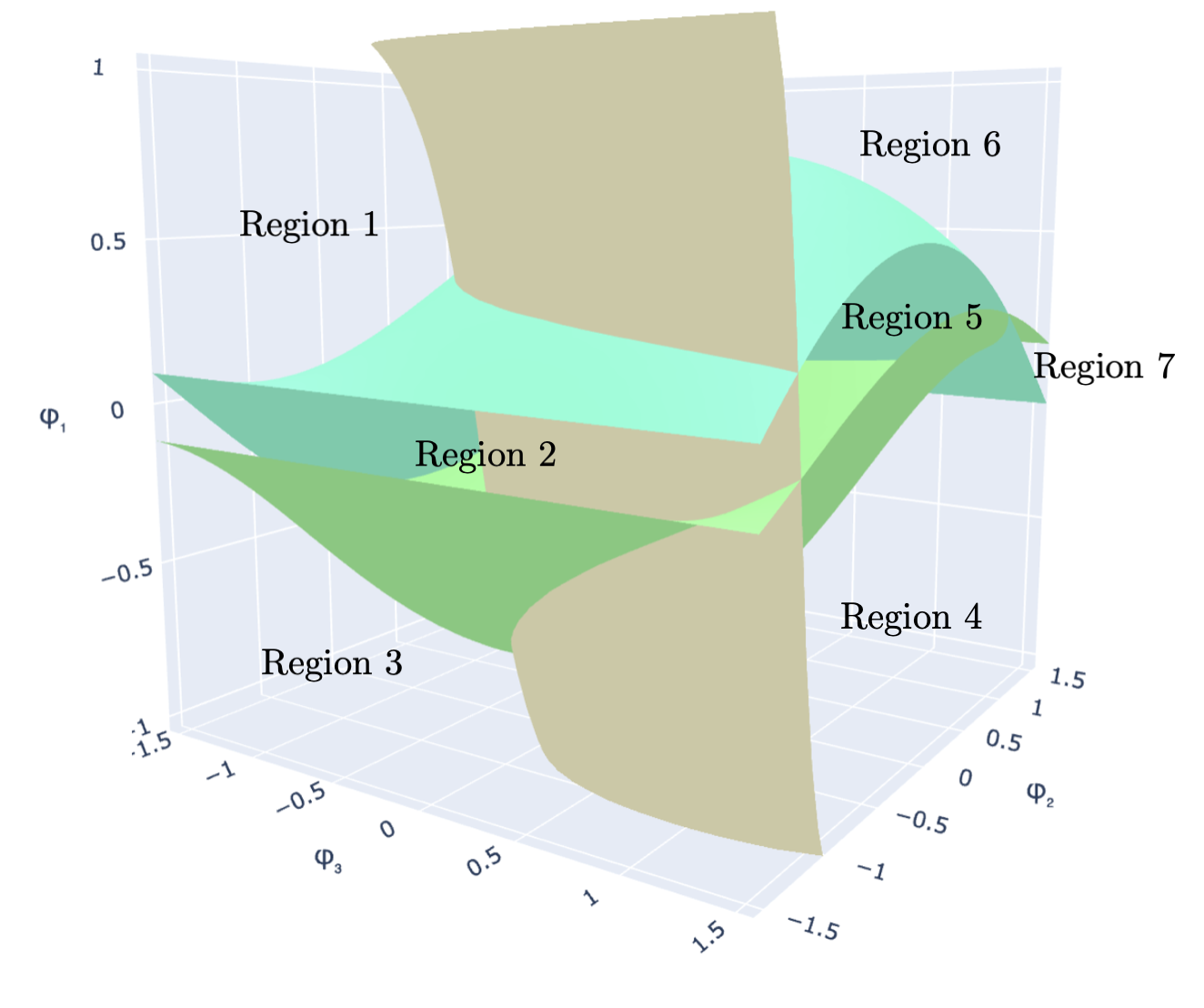}
\caption{Hyperplane arrangement for rank $r=2$ with $K=3$.
\emph{Left:} Three decision-boundary surfaces in the auxiliary space $(\varphi_1, \varphi_2, \varphi_3)$; their intersection defines a vertex.
\emph{Right:} The arrangement partitions the space into seven cells, each corresponding to a distinct candidate assignment vector.}
\label{fig:hypersurfaces}
\end{figure}

\medskip
\noindent \textbf{Structural properties.}
Two key properties of the hyperplane arrangement simplify the enumeration and govern which hyperplane combinations yield valid vertices:
\begin{enumerate}
\item \emph{Common axis.}
All $B_K$ hyperplanes from the same coordinate group $\GG^{(i)}$ (i.e., originating from the $i$-th row of~$\bfV$) intersect along a common $(2r{-}3)$-dimensional manifold, namely the locus where $\bfV_{i,:}\,\bfc(\bfphi) = 0$.
Geometrically, this means that choosing any two of the $B_K$ hyperplanes from group $\GG^{(i)}$ forces all $B_K$ to pass through the same point (\Cref{cor:two_implies_all} in the appendix).
\item \emph{Rank deficiency.}
Since each coordinate group spans a 2-dimensional subspace of $\bbR^{2r}$ (spanned by $\Repart\{\bfV_{i,:}\}$ and $\Impart\{\bfV_{i,:}\}$), selecting three or more hyperplanes from the same group makes the defining system rank-deficient (\Cref{lem:rank_deficiency}).
Hence, a \emph{valid} index set $I \subset [nB_K]$ of size $2r{-}1$ contains at most two indices from each coordinate group.
This constraint dramatically reduces the number of candidates: from $\binom{nB_K}{2r-1}$ (all subsets) to $\OO( rn^{2r-1})$ (valid subsets only).
\end{enumerate}

\medskip
\noindent \textbf{Candidate set size.} The algorithm relies on enumerating solutions from the candidate size. Thus, the computation complexity depends on the candidate size, which is given by Theorem~\ref{thm:candidate_size_simp}.

\begin{restatable}[Candidate set cardinality]{theorem}{card}
\label{thm:candidate_size_simp}
    The candidate set satisfies
    \begin{equation}
    \label{eq:card}
    \abs{\mathcal{S}(\bfV)} = \sum_{d=1}^{r}  \sum_{i=0}^{d-1} \binom{n}{i}\binom{n-i}{2(d-i)-1}\,B_K^{2(d-i)-2}\,(B_K - 1)^i = \OO( rn^{2r-1}).
    \end{equation}
    The sum over~$d$ arises from the recursive structure: the rank-$r$ candidate set is the union of interior vertices at rank~$r$ and the full candidate set at rank~$r{-}1$.
\end{restatable}

\begin{proof}[Proof sketch]
At rank level~$d$, interior vertices are formed by choosing $2d{-}1$ hyperplanes with at most two from each coordinate group.
Let $i$ groups contribute two hyperplanes each; the remaining $2d{-}1{-}2i$ come from distinct groups.
Each such configuration yields $\binom{n}{i}\binom{n-i}{2d-1-2i}$ choices of groups, $B_K^{2(d-i)-2}$ choices of boundaries from single groups (accounting for the $(-\pi/K, \pi/K]$ restriction on the last angle), and $(B_K - 1)^i$ adjacent cells per vertex (\Cref{lem:adjacent_cells}).
Summing over $i$ and recursing over $d$ gives~\eqref{eq:card}.
The full proof is provided in \Cref{app:candidate_set_proof}.
\end{proof}

\medskip
\noindent \textbf{Vertex enumeration and ambiguity resolution.}
To systematically identify all candidates, \Cref{alg:rankr} examines the vertices of the cell partition.
In the $(2r{-}1)$-dimensional hypercube $\HH_r$, each vertex is defined by the intersection of $2r{-}1$ hypersurfaces.
The algorithm iterates over all valid index sets $I \subset [nB_K]$ of size $2r{-}1$ (i.e., at most two indices from any single coordinate group) and, for each, solves $\tilde{\bfV}_{I,:}\,\tilde{\bfc} = \mathbf{0}$ for the null vector $\tilde{\bfc} \in \bbR^{2r}$.
If the system has rank $2r{-}1$, the null space is one-dimensional, yielding a unique direction that is converted to the complex vector $\bfc = \tilde{\bfc}_{1:r} + \imag\,\tilde{\bfc}_{r+1:2r}$ and evaluated via the decision rule~\eqref{eq:decision_rule}.
This decision rule maps a direction $\bfc$ to an assignment $\bfz$ as:
\begin{equation}
\label{eq:decision_rule}
z_i = \exp \Bigl(\frac{2\pi\imag}{K} k_i\Bigr), \quad k_i = \argmin_{k \in \{0,\ldots,K{-}1\}} \Bigl|\frac{2\pi k}{K} - \arg\bigl(\bfV_{i,:}\,\bfc\bigr)\Bigr|_{\bmod\,2\pi}.
\end{equation}

A subtlety arises when one or more coordinate groups contribute two hyperplanes to~$I$: at the vertex, all $B_K$ hyperplanes from such a group pass through (by the common-axis property), and the decision rule cannot determine $z_i$ uniquely because $\bfV_{i,:}\,\bfc(\bfphi) = 0$.
We resolve this by a \emph{fixed-angle intersection}, wherein we remove the ambiguous hyperplanes from~$I$, fix $\phi_{2r-1} = \pi/K$, and solve the reduced system.
The resulting intersection determines $z_i$ on one side of each boundary. The $B_K - 1$ distinct values of $z_i$ across adjacent cells are recovered by repeating with consecutive boundary pairs. See Algorithm~\ref{alg:rankr} for the complete procedure.

In addition to interior vertices, the algorithm must handle candidates on the boundary of $\HH_r$.
These boundary cases occur when one angular component is fixed at an extreme value, effectively reducing the search to a lower-dimensional subproblem.
They are handled by recursive calls to the rank-$(r{-}1)$ algorithm~\cite{kyrillidis2011rank}. In these calls, we fix one angular coordinate, thereby eliminating one degree of freedom and reducing the effective rank by one.

\begin{algorithm}[t]
\caption{Exact solver of~\eqref{eq:max_quad} for a rank-$r$ matrix.}
\label{alg:rankr}
\begin{algorithmic}[1]
\STATE \textbf{Input:} PSD matrix $\bfQ \in \bbC^{n \times n}$, factor $\bfV \in \bbC^{n \times r}$, rank~$r$, alphabet size~$K$.
\STATE \textbf{Output:} Optimal $\bfz^{\star} \in \AK^n$ (exact when $\mathrm{rank}(\bfQ) = r$).
\STATE Construct $\tilde{\bfV} \in \bbR^{nB_K \times 2r}$ via~\eqref{eq:tildeV}.
\STATE $\mathrm{opt} \leftarrow -\infty$;  $\bfz^{\star} \leftarrow \mathbf{0}$
\STATE $\II \leftarrow \{I \subset [nB_K] : |I| = 2r{-}1,  \text{at most 2 per coordinate group}\}$
\FOR{$I \in \II$}
  \IF{$\mathrm{rank}(\tilde{\bfV}_{I,:}) = 2r{-}1$}
    \STATE Solve $\tilde{\bfV}_{I,:}\,\tilde{\bfc} = \mathbf{0}$ for unit $\tilde{\bfc} \in \bbR^{2r}$
    \STATE $\bfc \leftarrow \tilde{\bfc}_{1:r} + \imag\,\tilde{\bfc}_{r+1:2r}$
    \STATE $\bfz \leftarrow \text{DecisionRule}(\bfV, \bfc)$ via~\eqref{eq:decision_rule}
    \STATE $\bfz \leftarrow \text{ResolveAmbiguity}(\tilde{\bfV}, I, \bfc, \bfV, \bfz, K)$ \hfill\COMMENT{fix used-vertex assignments}
    \IF{$\Repart(\bfz^{\dagger}\bfQ\bfz) > \mathrm{opt}$}
      \STATE $\mathrm{opt} \leftarrow \Repart(\bfz^{\dagger}\bfQ\bfz)$;  $\bfz^{\star} \leftarrow \bfz$
    \ENDIF
  \ENDIF
\ENDFOR
\IF{$r > 1$}
  \STATE $\bfz_{r-1} \leftarrow \text{Exact-Low-Rank}(\bfQ, \bfV_{:,1:r-1}, r{-}1, K)$ \hfill\COMMENT{recurse}
  \IF{$\Repart(\bfz_{r-1}^{\dagger}\bfQ\bfz_{r-1}) > \mathrm{opt}$}
    \STATE $\bfz^{\star} \leftarrow \bfz_{r-1}$
  \ENDIF
\ENDIF
\RETURN $\bfz^{\star}$
\end{algorithmic}
\end{algorithm}

\begin{restatable}[Rank-$r$ exactness]{theorem}{rankr}
\label{thm:rankr}
Let $\bfQ_r \in \bbC^{n \times n}$ be rank-$r$ PSD.
\Cref{alg:rankr} solves~\eqref{eq:max_quad} with objective~$\bfQ_r$ exactly in $\OO( rn^{2r+1})$ time.
With $P$ processors, the parallel time is $\OO( rn^{2r+1}/P)$.
\end{restatable}

\begin{proof}[Proof sketch]
The $nB_K$ hyperplanes partition $\HH_r$ into cells within each of which $\bfd(\bfV;\bfphi)$ is constant.
The global maximizer of~\eqref{eq:rankr_double} must lie in one of these cells, so and must therefore be attained either at an interior vertex or on the boundary of~$\HH_r$.
The main loop enumerates these interior vertices, and the recursive rank-$(r-1)$ calls looks at the boundary faces. 
By \Cref{thm:candidate_size_simp}, the total candidate count is $\OO( rn^{2r-1})$ and the evaluation of each candidate costs $\OO(n^2)$. The algorithm can be parallelized by processing each index set $I \in \II$ independently. 
The complete proof appears in~\Cref{app:rankr_proof}.
\end{proof}



\section{Approximation Guarantees}

In practice, the objective matrix $\bfQ$ in~\eqref{eq:max_quad} can rarely be strictly low-rank. Without a low-rank constraint, Algorithm~\ref{alg:rankr}'s computation complexity blows up exponentially with the rank $r$. Fortunately, many matrices encountered in real-world are \textit{approximately} low-rank, allowing us to perform a low-rank truncation on $\bfQ$ first before applying Algorithm~\ref{alg:rankr}. In Section~\ref{sec:approx}, we model the setup by considering $\bfQ$ as a perturbed version of a low-rank matrix $\bfQ^\star$ with $\bfQ = \bfQ^\star + \bfH$, where $\bfH$ is a noise matrix, and study the theoretical guarantee of solving~\eqref{eq:max_quad} with the low-rank truncation. Section~\ref{sec:sampling} takes a further step into reducing the computational complexity by adopting a randomized approach to sample and check elements from the candidate set. In Section~\ref{sec:sampling}, we provide high-probability bounds on the optimization objective value to justify its effectiveness.

\subsection{Perturbed Low-Rank Objectives}\label{sec:approx}
Here we consider the objective in~\eqref{eq:max_quad} that involves the matrix $\bfQ = \bfQ^{\star} + \bfH$ as a perturbation of a rank-$r^{\star}$ signal $\bfQ^{\star}$. In particular, our goal is to characterize the quality of the resulting solution in terms of the noiseless quadratic objective $\bfz^\dagger\bfQ^\star\bfz$.

Assume the observed matrix decomposes as $\bfQ = \bfQ^{\star} + \bfH$,
where $\bfQ^{\star} \succeq 0$ is Hermitian with rank $r^{\star} \ll n$ and eigendecomposition $\bfQ^\star = \sum_{j=1}^r\lambda_j^\star\bfu_j^\star\bfu_j^{\star\dagger}$, and $\bfH \in \bbC^{n \times n}$ represents a perturbation (which need not be Hermitian or PSD).
Note that $\bfQ$ itself may not be Hermitian or PSD.
\Cref{alg:approx} computes the top-$r$ spectral truncation $\bfQ_r$ of~$\bfQ$ and runs the exact low-rank solver on~$\bfQ_r$.
Define the maximizers:
\begin{equation}
\label{eq:def_maximizers}
\begin{aligned}
\bfz^{\star}   \text{ attains } ~   \optQ &:= \max_{\bfz \in \AK^n}  \bfz^{\dagger}\bfQ^{\star}\bfz, \quad
\bfz_r   \text{ attains } ~  \texttt{OPT}\text{-}\bfQ_r &:= \max_{\bfz \in \AK^n}  \bfz^{\dagger}\bfQ_r\bfz.
\end{aligned}
\end{equation}
Our goal is to bound the quality of $\bfz_r$ evaluated on the signal matrix $\bfQ^{\star}$, i.e., to bound $|\optQ - \bfz_r^{\dagger}\bfQ^{\star}\bfz_r|$ and $\frac{\bfz_r^{\dagger}\bfQ^{\star}\bfz_r}{\optQ}$.
Our theory applies for general $K \geq 2$.

\begin{algorithm}[t]
\caption{Approximate solver of~\eqref{eq:max_quad} via low-rank truncation.}
\label{alg:approx}
\begin{algorithmic}[1]
\STATE \textbf{Input:} Matrix $\bfQ \in \bbC^{n \times n}$, target rank~$r$, alphabet size~$K$.
\STATE Compute top-$r$ singular values $\sigma_1 \geq \cdots \geq \sigma_r$ and left singular vectors $\bfu_1, \ldots, \bfu_r$ of $\bfQ$.
\STATE $\bfQ_r \leftarrow \sum_{j=1}^r \sigma_j \bfu_j \bfu_j^{\dagger}$;  $\bfV_r \leftarrow [\sqrt{\sigma_1}\bfu_1 \,|\, \cdots \,|\, \sqrt{\sigma_r}\bfu_r]$.
\RETURN Exact-Low-Rank$(\bfQ_r, \bfV_r, r, K)$.
\end{algorithmic}
\end{algorithm}

Recall that $\lambda_j^\star$s are the eigenvalues of $\bfQ^\star$, sorted in descending order.
Define the eigengap $\delta^{\star} = \min \bigl\{\min_{j \in [r-1]} |\lambda_j^{\star} - \lambda_{j+1}^{\star}|,  \lambda_r^{\star}\bigr\}$ and let $\bfz_r$ denote the output of \Cref{alg:approx}. We begin by establishing an bound on the additive approximation error. This bound has two components: $n\lambda_{r+1}^{\star}$ is the price of rank truncation (zero when $r = r^{\star}$), and $n\lambda_1^{\star}\norm{\bfH}_2/\delta^{\star}$ captures the effect of perturbation.

\begin{restatable}[Additive bound]{theorem}{additiveBound}
\label{thm:additive}
    If $r \leq r^{\star}$ and $\norm{\bfH}_2 \leq \delta^{\star}/2$, then
    \begin{equation}
    \label{eq:additive}
        \bigl|\optQ - \bfz_r^{\dagger}\bfQ^{\star}\bfz_r\bigr| \leq \OO \Bigl( n\Bigl(\lambda_{r+1}^{\star} + \frac{\lambda_1^{\star}}{\delta^{\star}}\norm{\bfH}_2\Bigr)\Bigr).
    \end{equation}
\end{restatable}
The full proof of Theorem~\ref{thm:additive} is provided in~\Cref{app:additiveRatioProof}. Theorem~\ref{thm:additive} upper bounds the gap between the quadratic form $\bfz_r^\dagger\bfQ^{\star}\bfz_r$ and $\texttt{OPT-}\bfQ^{\star}$. The scaling factor of $p$ in the right-hand side of (\ref{eq:additive}) is due to the fact that both $\texttt{OPT-}\bfQ^{\star}$ and $\bfz_r^{\dagger}\bfQ^{\star}\bfz_r$ are quadratic in $\bfz^{\star}$ and $\bfz_r$, which has $\ell_2$-norm $\sqrt{p}$ by construction. In particular, the right-hand side of the upper bound in (\ref{eq:additive}) consists of two component. The first component $p\lambda_{r+1}^{\star}$ scales with the eigenvalue $\lambda_{r+1}^{\star}$, and can be regarded as the price one pay when using a smaller rank less than the rank of $\bfQ^{\star}$, i.e. $r\leq r^{\star}$. The second term depends on the noise scale $\norm{\bfH}_2$, showing the influence of the perturbation on the quality of the solution. One should notice that both the requirement $\norm{\bfH}_2\leq \frac{\delta^{\star}}{2}$ and the upper bound in (\ref{eq:additive}) depends on the eigengap $\delta^{\star}$. Such dependency appears because in the algorithm we need to explicitly construct the Hermitian matrix $\bfQ_r$ from the non-Hermitian $\bfQ$, which naturally introduced the eigenvector perturbation $\bfu_i$ from $\bfu_i^{\star}$ due to $\bfH$. It is common in prior work that such eigenvector perturnbation bound involves a dependency on the eigengap $\delta^{\star}$ \cite{chen2018asymmetry,wedin1972perturbation}. The corollary below shows that we can drop the dependency on the eigengap if the perturbation matrix $\bfH$ is Hermitian:
\begin{restatable}{corollary}{approxRatioHerm}
\label{cor:approx_ratio_perturbation_herm}
     Let $\bfz_r,\texttt{OPT-}\bfQ^{\star}$ be defined in (\ref{eq:def_maximizers}). Let $\lambda_j^{\star}$ be the $j$th eigenvalue of $\bfQ^{\star}$ with $\lambda_j^{\star} = 0$ if $j > r$. If $r\leq r^{\star}$ and $\norm{\bfH}_2$ is Hermitian, then we have that:
     \begin{equation}
        \label{eq:ub_e_by_h_cor1}
        \left|\texttt{OPT-}\bfQ^{\star} - \bfz_r^{\dagger}\bfQ^{\star}\bfz_r\right| \leq O\paren{p\paren{\lambda_{r+1}^\star + \norm{\bfH}_2}}.
    \end{equation}
\end{restatable}
In short, \Cref{cor:approx_ratio_perturbation_herm} saves a multiplicative factor of $\frac{\lambda_1^{\star}}{\delta^{\star}}$ compared with Theorem \ref{thm:additive} by assuming that $\bfH$ is Hermitian. The proof of \Cref{cor:approx_ratio_perturbation_herm} is given in \Cref{app:proof_cor_approx}.


Theorem~\ref{thm:additive} made no assumptions on the noise matrix $\bfH$. 
With different assumptions on $\bfH$, we can turn (\ref{eq:additive}) into different approximation error bounds. In this section, we give a simple result under the assumption that each entry of $\bfH$ is independently drawn from a complex Gaussian, $\NN_{\CC}(0, \alpha^2)$. We specifically consider the circularly-symmetric complex Gaussian distribution, defined as:
\begin{definition}
    A complex random variable $Z = X + \imag Y$ follows the complex Gaussian distribution $\NN_{\CC}(0, \alpha^2)$ if $X$ and $Y$ are independently distributed according to $\NN(0, \alpha^2/2)$.
\end{definition}

We use the following fact for our analysis:

\begin{restatable}[Proposition 3.1 of \cite{benhamou2019operatornormupperbound}]{fact}{vershyninGaussianMatrix}
\label{thm:vershynin_gaussian_matrix}
    Let $\bfM\in\bbC^{n\times n}$ be a square matrix with independently entries $\bfM_{ij}\in\texttt{subG}\paren{1}$. Then there exists constant $C, c > 0$ such that for all $A\geq C$, with probability at least $1 - C\exp\paren{-cAn}$, we have that $\norm{\bfM}_2\leq A\sqrt{n}$.
\end{restatable}

Using \Cref{thm:additive} and \Cref{thm:vershynin_gaussian_matrix}, we establish more specific bounds, as in the following corollary.

\begin{corollary}
    \label{cor:noise_plug_in}
    Let $\bfz_r,\texttt{OPT-}\bfQ^{\star}$ be defined in (\ref{eq:def_maximizers}). Let $\lambda_j^{\star}$ be the $j$th eigenvalue of $\bfQ^{\star}$ with $\lambda_j^{\star} = 0$ if $j > r$, and define $\delta^{\star} = \min\left\{\min_{j\in[r-1]}\left|\lambda_j^{\star} - \lambda_{j+1}^{\star}\right|,\lambda_r^{\star}\right\}$. Let $\bfH \in \mathbb{C}^{p \times p}$ such that $H_{i, k} \sim \NN_\CC \left(0, \alpha^2 \right)$ for some $\alpha \in \mathbb{R}$. If $\alpha\leq \frac{c\delta^{\star}}{\sqrt{p}}$ for some small enough constant $c$, then, there exist constants $C, c' > 0$ such that with probability at least $1 - C\exp\paren{-c'n}$, we have that:
    \[
        \left|\texttt{OPT-}\bfQ^{\star} - \bfz_r\bfQ^{\star}\bfz_r\right| \leq O\paren{p\lambda_{r+1}^{\star} + \tfrac{p^{\frac{3}{2}}\lambda_1^{\star}\varepsilon}{\delta^{\star}}}.
    \]
\end{corollary}

\begin{proof}
    By definition, since $H_{i,k}\sim\mathcal{N}_{\mathcal{C}}\paren{0,\alpha^2}$, we have that $\text{Re}\paren{H_{i,k}},\text{Im}\paren{H_{i,k}}\sim\mathcal{N}\paren{0,\frac{\alpha^2}{2}}$ independently. Thus, we have that $H_{i,k}$ is sub-Gaussian in the complex space, which allows us to use Theorem~\ref{thm:vershynin_gaussian_matrix} to obtain that, with probability at least $1 - C\exp\paren{-c'p}$, we have $ \norm{\alpha^{-1}\bfH}_2\leq O\paren{\sqrt{p}} \Rightarrow \norm{\bfH}_2\leq O\paren{\alpha\sqrt{p}}$. 
    Notice the choice of $\alpha$; plugging that upper bound into (\ref{eq:additive}) gives the desired result.
\end{proof}
One could notice that the upper bound on the operator norm of $\bfH$ matches the standard result in the real case \cite{vershynin2018high}, both with a scaling of $\sqrt{p}$. In addition, Corollary \ref{cor:noise_plug_in} implies that the noise term dominates only when $\varepsilon\geq \Omega\paren{\frac{\delta^{\star}\lambda_{r+1}^{\star}}{\lambda_1^{\star}\sqrt{p}}}$, showcasing a relationship between the approximation error, the noise scale, and the threshold rank $r$.

Using this additive bound, we can bound the approximation ratio incurred by the perturbation and rank-$r$ approximation. The detailed proof of Theorem~\ref{thm:mult} is available in \Cref{app:multRatioProof}.


\begin{restatable}[Multiplicative bound]{theorem}{multBound}
\label{thm:mult}
Let $\bfQ^{\star}$ be rank-$r$ Hermitian with $\norm{\bfu_1^{\star}}_{\infty} \leq \mu/\sqrt{n}$.
If $\norm{\bfH}_2 \leq \delta^{\star}/2$, then
\begin{equation}
\label{eq:mult}
\frac{\bfz_r^{\dagger}\bfQ^{\star}\bfz_r}{\optQ}
\geq 1 - \OO \left(\frac{\norm{\bfH}_2}{\delta^{\star}}\right).
\end{equation}
\end{restatable}


\subsection{Sampling Candidates} \label{sec:sampling}

In this section, we consider a randomized alternative to \Cref{alg:rankr} based on sampling random candidates. Let $\SC{r} := \{\bfc \in \bbC^r : \norm{\bfc}_2 = 1\}$ denote the
complex unit sphere and write $\bfc \sim \mathrm{Unif}(\SC{r})$ for the uniform
distribution on it. For $\bfc \in \SC{r}$, let $\bfz(\bfc) \in \AK^n$ be the
assignment produced by the decision rule~\eqref{eq:decision_rule}, and define
the sampled objective value
\begin{equation}
\label{eq:def_f}
  f(\bfc) := \norm{\bfV^{\dagger}\bfz(\bfc)}_2^2 = \bfz(\bfc)^{\dagger}\bfQ_r\,\bfz(\bfc).
\end{equation}
At rank~$1$ the complex sphere is the unit circle: $\bfc = e^{\imag\varphi}$,
and $\bfc \sim \mathrm{Unif}(\SC{1})$ is the same as
$\varphi \sim \mathrm{Unif}[0, 2\pi)$; we abbreviate
$f(\varphi) := f(e^{\imag\varphi})$.

Instead of iterating over all of the cells, we uniformly sample $S$ candidate
directions $\bfc_1, \ldots, \bfc_S \sim \mathrm{Unif}(\SC{r})$ and return the
best of the rounded assignments $\bfz(\bfc_1), \ldots, \bfz(\bfc_S)$. The
difficulty in analyzing this scheme is that no off-the-shelf guarantee
describes how much probability mass the optimal or near-optimal directions carry. The
results of this section establish exactly that: we prove upper bounds on the
sample complexity, showing that a near-optimal objective value is reached
with a number of samples independent of $n$.

\subsubsection{Rank-1}
\label{sec:thm1a}

We give bounds on the sample complexity required to obtain a near-optimum by randomly sampling $S$ directions and choosing the maximum objective.
We introduce the rounding margin, which will be used in our sample complexity bound for rank-$1$ matrices. 

\begin{definition}[Rounding margin of the rank-1 optimum]
\label{def:rank1_margin}

Assume $r=1$. Consider a matrix $\bfQ = \lambda \bfq \bfq^\dagger$. Let $\varphi^\star$ maximize $f$, i.e., $f(\varphi^\star)=\mathrm{OPT}_1$. Let $\vartheta_m$ be a decision boundary and let $\theta_i = \mathrm{arg}(\bfq_i)$.
The {rounding margin}
\[
  \eta \;:=\; \frac{K}{\pi}\,\min_{i \in [n]}\
  \operatorname{dist}\!\bigl(\theta_i+\varphi^\star,\ \{\vartheta_m\}_{m=0}^{K-1}\bigr)
  \;\in\; (0,1]
\]
is the normalized distance from the closest coordinate to a decision boundary
at the optimum. 

\end{definition}

We remark that the rounding margin is positive under \Cref{asm:general_position}, which guarantees no angle $\theta_i+\varphi^\star$ lies exactly on a boundary $\vartheta_m=\pi(2m{+}1)/K$. Our first bound, making use of the rounding margin, holds only for rank-$1$ objectives $\bfQ_1$, and establishes conditions under which we can, with high probability, obtain the \textit{exact optimum}.

\begin{restatable}{theorem}{tailRankOne}
\label{thm:tail_rank1}
Let $\bfQ_1=\bfV\bfV^\dagger\in\CC^{n\times n}$ be rank-1 PSD with
$V_{i,1}=r_ie^{\imag\theta_i}$, $r_i\ge0$, and $\bfc=e^{\imag\varphi}$ with
$\varphi\sim\mathrm{Unif}[0,2\pi)$. Let $K\ge2$ and suppose
\Cref{asm:general_position} holds. Let $\eta\in(0,1]$ be the rounding margin
of \Cref{def:rank1_margin}. Then,
\begin{equation}
\label{eq:tail_rank1_incoh}
   \Pr[\varphi]{f(\varphi)=\mathrm{OPT}_1}\ \ge\ \eta .
\end{equation}
Consequently, drawing $S\ge \log(1/\delta)/\eta$ i.i.d. directions and returning the
best candidate among them recovers the exact rank-1 optimum with probability at least $1-\delta$.
\end{restatable}

\begin{proof}
Under \Cref{asm:general_position}, $f$ is piecewise constant on $\mathbb{S}^1$ with
$nK$ open cells separated by the rounding discontinuities, and
$f\equiv\mathrm{OPT}_1$ on the cell $C^\star\ni\varphi^\star$. The assignment is
$\tfrac{2\pi}{K}$-periodic up to a global relabel: $k_i(\varphi+\tfrac{2\pi}{K})
=k_i(\varphi)+1$, so $\bfz(\varphi+\tfrac{2\pi}{K})=\omega\,\bfz(\varphi)$ with
$\omega=e^{2\pi\imag/K}$, and $f(\varphi+\tfrac{2\pi}{K})=|\omega|^2f(\varphi)
=f(\varphi)$. 
Hence all $K$ rotated copies $C^\star+\tfrac{2\pi j}{K}$, $j=0,\dots,K{-}1$,
of the optimal cell satisfy $f\equiv\mathrm{OPT}_1$. These copies are pairwise
disjoint: rotating by $\tfrac{2\pi}{K}$ changes the assignment (a global
relabel by $\omega\neq1$), so no cell contains both $\varphi$ and
$\varphi+\tfrac{2\pi}{K}$, i.e., every cell has width less than $\tfrac{2\pi}{K}$. By
\Cref{def:rank1_margin} the nearest coordinate lies at distance $\ge\eta\pi/K$ from a
boundary on each side of $\varphi^\star$, so $|C^\star|\ge 2\pi\eta/K$ and, as
$\varphi\sim\mathrm{Unif}[0,2\pi)$,
\[
   \Pr[\varphi]{f(\varphi)=\mathrm{OPT}_1}\ \ge\ \frac{K\,|C^\star|}{2\pi}\ \ge\ \eta.
\]
For the sample complexity, $\Pr{\text{all }S\text{ draws fail}}\le(1-\eta)^S\le
e^{-\eta S}\le\delta$ whenever $S\ge\log(1/\delta)/\eta$.
\end{proof}

\subsubsection{Rank-$r$}

Our next bound is sharpest for small $r$: the per-sample success probability
decays as $\varepsilon^{r-1}$, so at fixed accuracy $\varepsilon$ the sample
complexity grows exponentially in $r$. The proof relies on showing that the
squared inner product $|\bfu^\dagger\bfc|^2$ between a uniformly sampled direction $\bfc$
and a fixed extremal direction $\bfu$ follows a $\mathrm{Beta}(1,r{-}1)$
distribution.

\begin{restatable}{theorem}{betaTail}
\label{thm:rr_tail}
Let $K \geq 2$ and $r \geq 2$ and assume $\bfV \in \CC^{n \times r}$ is of full column rank
(so $\bfQ_r = \bfV \bfV^\dagger$ has rank $r$).
For any $\varepsilon \in (0, 1)$,
\begin{equation}
\label{eq:rr_tail_main}
 \Pr[\bfc \sim \mathrm{Unif}(\SC{r})]{
 f(\bfc) \geq \cos^2(\pi/K)(1 - \varepsilon) \mathrm{OPT}_r}
 \geq \varepsilon^{r-1}.
\end{equation}

Consequently, given a set of $S \geq \log(1/\delta)/\varepsilon^{r-1}$ i.i.d. samples $\{\bfc_i\}_{i=1}^S$ such that $\bfc_i \sim \mathrm{Unif}(\SC{r})$, it holds that with probability at least $1 - \delta$,
\begin{equation}
\label{}
 \max_{1 \leq i \leq S} f({\bfc_i}) \geq \cos^2(\pi/K)(1 - \varepsilon) \mathrm{OPT}_r.
\end{equation}
\end{restatable}

\begin{proof}[Proof sketch]
The proof rests on two deterministic facts: nearest-root rounding loses at
most the worst residual cosine per coordinate, giving
$f(\bfc)\ge\cos^2(\pi/K)\,N(\bfc)^2$ with $N(\bfc):=\norm{\bfV\bfc}_1$; and
the norm $N$ admits a linear minorant $|\bfa^\dagger\bfc|\le N(\bfc)$ (a
complex Hahn--Banach support functional at the maximizer of $N$ on the
sphere) with $\norm{\bfa}_2\ge M:=\max_{\norm{\bfc}_2=1}N(\bfc)$, while
duality gives $\OPTr\le M^2$. With $\bfu:=\bfa/\norm{\bfa}_2$, the alignment
$X=|\bfu^\dagger\bfc|^2$ has an explicit law on the complex sphere,
$\mathrm{Beta}(1,r-1)$, so the event $\{X\ge1-\varepsilon\}$ has probability
exactly $\varepsilon^{r-1}$. On this event, combining the two deterministic
facts gives $f(\bfc)\ge\cos^2(\pi/K)(1-\varepsilon)M^2
\ge\cos^2(\pi/K)(1-\varepsilon)\OPTr$. The exponent $r-1$ reflects the measure
of spherical caps of fixed relative size, which decays exponentially with the
dimension; this is why the bound is sharpest for small $r$. The sample
complexity follows from independence across draws:
$(1-\varepsilon^{r-1})^S\le e^{-S\varepsilon^{r-1}}$.

The detailed proof of Theorem~\ref{thm:rr_tail} is available in \Cref{app:proof-beta-tail}.
\end{proof}

As a consequence of \Cref{thm:rr_tail}, we remark that for a low-rank, Hermitian, PSD matrix $\bfQ$, we can approximate a discrete complex quadratic maximization with a sample size $S$ that is entirely independent of $n$. This reduces the computational complexity of our technique from $\OO(rn^{2r-1})$ to $\OO(S \cdot n^2) = \OO(n^2 / \varepsilon^{r-1})$, at the cost of an approximation factor.

Our final theoretical result improves upon the dependence on $r$ in the sample complexity of \Cref{thm:rr_tail}, at the expense of obtaining a weaker approximation factor and introducing a dependence on the incoherence parameter $\mu$, which may be dependent on $n$.

\begin{restatable}{theorem}{pzBasedTail}
\label{thm:tail_n_dep}

Let $K \geq 2$ and $r \geq 2$ and assume $\bfV \in \CC^{n \times r}$ is of full column rank.
Assume that \Cref{asm:general_position} and~\Cref{asm:incoherence}
with incoherence parameter $\mu$ hold. For any $\varepsilon \in (0, 1)$,
\begin{equation}
\label{eq:thm_tail_n_dep}
\Pr[\bfc \sim \mathrm{Unif}(\SC{r})] {
 f(\bfc) \geq \frac{\cos^2(\pi/K)}{r} (1-\varepsilon) \OPTr}
 \geq 
 \frac{\varepsilon^2}{1 + V_0\bigl(r, K, \mu\bigr)},
\end{equation}
where for some absolute constant $c$,
\begin{equation}
\label{eq:thm_tail_n_dep_V0}
 V_0\bigl(r, K, \mu \bigr)
 \leq 
 \frac{c \cdot r^6 \cdot \mu^4}{\cos^8(\pi/K)}.
\end{equation}
Consequently, given a set of $S \geq \OO \left(\frac{\log(1/\delta) \mu^4 r^6}{\varepsilon^2 \cos^8(\pi/K)}\right)$ i.i.d. samples $\{\bfc_i\}_{i=1}^S$ such that $\bfc_i \sim \mathrm{Unif}(\SC{r})$, it holds that with probability at least $1 - \delta$,
\begin{equation}
\label{eq:pz_tail_success_error}
    \max_{1 \leq i \leq S} f({\bfc_i}) \geq \frac{\cos^2(\pi/K)}{r} (1 - \varepsilon)\mathrm{OPT}_r.
\end{equation}
\end{restatable}

\begin{proof}[Proof sketch]
The proof applies the Paley--Zygmund inequality: a non-negative random
variable exceeds a $(1-\varepsilon)$-fraction of its mean with probability at
least $\varepsilon^2/(1+V_0)$, where $V_0=\Var{f}/(\E{f})^2$. The main point
is that $V_0$ does not grow with $n$: both the mean and the standard
deviation of $f$ scale as $\Lambda n$, where $\Lambda=\sum_j\lambda_j$. 
This term therefore cancels in the ratio. For the numerator, the covariance of
any two unit-modulus edge terms is at most one, which already gives
$\mathrm{Var}[f]\le n^2\Lambda^2$.
For the denominator, two bounds are chained: the in-expectation guarantee
$\E{f}\ge(\cos^2(\pi/K)/r)\,\OPTr$, proved in \Cref{app:expectationBound}, and a lower bound on
$\OPTr$ itself, obtained by rounding the top eigenvector. To prove this lower bound, we show that incoherence forces
the mass of $\bfu_1$ to spread over many coordinates. As a result, the candidate $\bfz$ obtained by rounding $\bfu$
already certifies $\OPTr \gtrsim \bfz^\dagger \bfQ \bfz \gtrsim\cos^2(\pi/K)\Lambda n/(r^2\mu^2)$. 
The detailed proof of Theorem~\ref{thm:tail_n_dep} is available in \Cref{app:proof-pz-tail}.
\end{proof}

\section{Experiments}
\label{sec:experiments}
We evaluate our low-rank algorithms (\Cref{alg:rankr} and \Cref{alg:approx}) on \textsc{Max-3-Cut}, the $K = 3$ instantiation of \eqref{eq:max_quad} with a graph Laplacian objective matrix, using synthetic graphs and established benchmarks. Additional details and results appear in \Cref{app:experiments}.

\medskip
\noindent \textbf{Setup.}
We test on four synthetic graph families and the GSet benchmark~\cite{gset_dataset}.
\emph{Erd\H{o}s--R\'enyi} $G(n,p)$ graphs include each edge independently with probability~$p$.
\emph{Random $d$-regular} graphs are sampled uniformly from the set of $d$-regular graphs on $n$ vertices; we use $d = 5$ in our experiments.
\emph{Toroidal lattices} arrange $n = k \times k$ vertices on a torus with nearest-neighbor edges; these graphs have highly structured Laplacian spectra.
\emph{Stochastic block model (SBM)} graphs have 3 balanced communities with intra-community edge probability $p_{\mathrm{in}} = 0.5$ and inter-community probability $p_{\mathrm{out}} = 0.1$.
For each random graph family, we average over 20 instances with fixed seeds.
All experiments were run on a single Intel Xeon machine with 128 GB RAM.

\medskip
\noindent \textbf{Baselines.}
We compare against: (i)~\textsc{Frieze--Jerrum}~\cite{frieze1997improved}: a complex SDP relaxation solved via CVXPY/SCS, followed by hyperplane rounding;
(ii)~\textsc{Greedy}~\cite{gui2018bqp}: a steepest-ascent local search that iteratively assigns each node to its best partition;
(iii)~\textsc{Genetic}~\cite{panxing2016genetic}: a population-based metaheuristic that evolves candidate cuts via crossover and mutation;
and (iv)~\textsc{Random}: the best of $n+1$ uniformly random assignments, matching the candidate count of our rank-1 algorithm for a fair comparison.
On GSet instances, we also compare to \textsc{MOH}~\cite{ma2017moh}, a breakout search heuristic, using cut values from the original publication.
\textcolor{black}{We also compare to \textsc{Simulated Annealing} (SA) \cite{johnson1991optimization} with geometric cooling schedule ($T_0 = 1$, cooling factor $0.995$, $10^4$ iterations per temperature) in the rank-1 scaling experiments.}

\medskip
\noindent \textbf{Parallel implementation.}
For rank-2 and rank-3 experiments, we parallelize \Cref{alg:rankr} by mapping disjoint batches of index sets to parallel workers.
Candidates are generated as a stream, so the implementation never materializes the full set of $\binom{nB_K}{2r-1}$ combinations in memory.
Each worker receives a slice of the candidate space, computes cut values via batched matrix operations, and returns only the best candidate found.
No inter-worker communication occurs during execution.
The full 15-GPU rank-2 campaign appears in \Cref{app:gpu}.

\medskip
\noindent \textbf{Small-scale results.}
\Cref{fig:ratios} shows empirical approximation ratios (score / best score) on 5-regular graphs at $n \in \{20, 50, 100\}$ and Erd\H{o}s--R\'enyi graphs at $n = 100$.
Our rank-3 algorithm achieves the highest average ratio on graphs with $n \leq 50$; rank-2 matches greedy and consistently outperforms SDP rounding across all densities.
On 5-regular graphs with 20 nodes, rank-3 achieves an average ratio of $0.997$, compared to $0.989$ for greedy and $0.976$ for SDP.
As $n$ increases to 100, rank-2 remains competitive with greedy (average ratio $0.985$ vs.\ $0.990$) while rank-3 becomes intractable due to its $\OO(n^7)$ cost.
On Erd\H{o}s--R\'enyi graphs, the advantage of higher rank is more pronounced at higher densities: at $p = 0.9$, rank-2 exceeds SDP by 3\%, reflecting the more concentrated spectrum of dense graphs.

\begin{figure}[t]
\centering
\includegraphics[width=\textwidth]{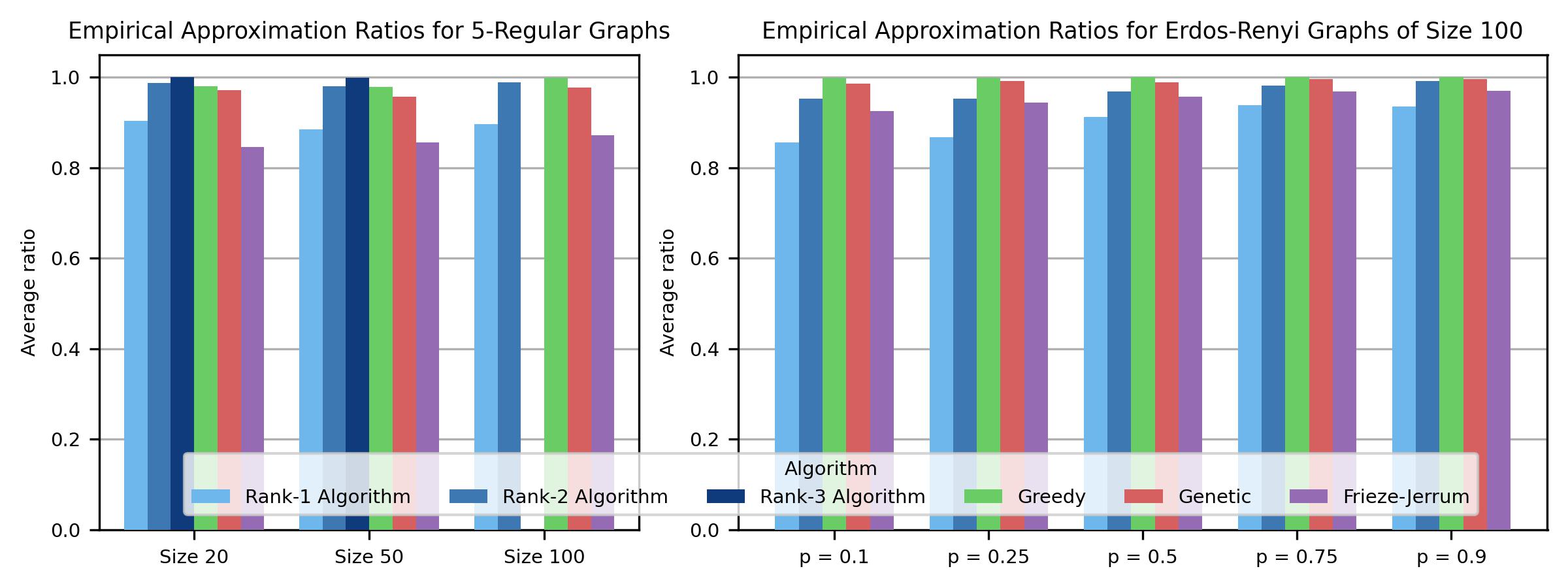}
\caption{Empirical approximation ratios on 5-regular graphs (left) and Erd\H{o}s--R\'enyi graphs with $n = 100$ (right).
Higher is better.
Our rank-2/3 algorithms match or exceed greedy and outperform Frieze--Jerrum SDP at all densities.}
\label{fig:ratios}
\end{figure}

\medskip
\noindent \textbf{GSet benchmark.}
\Cref{tab:gset} reports selected results on GSet instances.
On structured toroidal graphs (G48, G49), rank-1 finds the provably optimal cut value of 6000 in 5 seconds---57--74$\times$ faster than greedy, which requires 5--6~minutes.
On the largest instances ($n \geq 10,000$), greedy times out at 30~minutes while rank-1 completes in 76--352~seconds.
On less structured instances (Erd\H{o}s--R\'enyi, skew), rank-1 trades solution quality for speed: about 90\% of greedy's score at $15$--$23\times$ lower runtime.
These results align with our theory: on toroidal graphs, 
\textcolor{black}{the Laplacian spectrum is highly concentrated (the spectral energy ratio satisfies $\rho_1 < 20\%$ in our experiments), making rank-1 an excellent approximation; see \Cref{app:gpu} for the measured values.}
On random or skew graphs, the spectrum is flat, and a rank-1 approximation captures only a small fraction of the structure.
Extended results on all GSet instances appear in~\Cref{app:gset_full}.


\medskip
\noindent \textbf{Very large graphs.}
To test scalability limits, we generated random 3-regular graphs with $n = 50,000$ and $n = 100,000$ nodes.
On the $100,000$-node graph, rank-1 finishes in 17.1 hours and finds a cut value of $137,796$, while Random ($n+1$ random assignments) requires 2.3 hours and finds $100,781$.
Most of the runtime (14.5 hours) is spent computing the principal eigenvector via ARPACK; the sweep itself takes under 3 hours.
On the $50,000$-node graph, rank-1 finishes in 1.8 hours with cut $68,951$ vs.\ Random's $50,517$ in 22 minutes.
The incremental scoring and two-eigenvector rounding that reduce the million-node sweep to ${\sim}70$ minutes are detailed in \Cref{app:rank1scale}.

\begin{table}[t]
\centering
\caption{Selected GSet results for rank-1 vs.\ greedy and MOH~\cite{ma2017moh}.
Runtimes in seconds.
Bold: optimal cut value.
MOH runtimes omitted (different hardware).}
\label{tab:gset}
\setlength{\tabcolsep}{5pt}
\begin{tabular}{@{}lrcrrr@{}}
\toprule
Instance & $n$ & Type & Greedy & MOH & Rank-1 \\
\midrule
G1  & 800   & Erd\H{o}s--R. & 14859 (16s) & 15165 & 13331 (1.1s) \\
G14 & 800   & Skew    & 3914 (12s)  & 4012  & 3217 (0.6s) \\
G22 & 2000  & Erd\H{o}s--R. & 16669 (111s) & 17167 & 14145 (5.5s) \\
G43 & 1000  & Erd\H{o}s--R. & 8273 (48s) & 8578 & 6839 (2.1s) \\
G48 & 3000  & Torus   & 5998 (300s) & 6000  & \textbf{6000} (5.2s) \\
G49 & 3000  & Torus   & 5996 (394s) & 6000  & \textbf{6000} (5.3s) \\
G55 & 5000  & Erd\H{o}s--R. & timeout & 12429 & 9956 (65s) \\
G67 & 10000 & Torus   & timeout     & 8086  & 3117 (76s) \\
G81 & 20000 & Torus   & timeout     & 16321 & 4122 (352s) \\
\bottomrule
\end{tabular}
\end{table}


\medskip
\noindent \textbf{When does low-rank work?}
The \emph{spectral energy ratio} $\rho_r = 1 - \sum_{j=1}^r \lambda_j^{\star} / \sum_{j=1}^{r^{\star}} \lambda_j^{\star}$ measures the fraction of spectral energy not captured by the top $r$ eigenvalues.
When $\rho_r$ is small (e.g., $<5\%$), the rank-$r$ truncation is a good approximation and our bounds predict strong performance.
On torus graphs ($\rho_2 < 2\%$), rank-2 is exact and rank-1 is near-optimal.
On 5-regular graphs ($\rho_2 \approx 1.4\%$), rank-2 exceeds SDP at $n \geq 1000$.
On SBM and Erd\H{o}s--R\'enyi graphs ($\rho_2 > 98\%$), the low-rank approximation captures almost no structure, and greedy dominates.
This pattern is consistent with our perturbation bounds (\Cref{thm:additive,thm:mult}), which predict strong performance precisely when the spectral energy is concentrated.
In practice, computing $\rho_r$ (which requires only the top $r+1$ eigenvalues) provides a cheap diagnostic for whether the low-rank algorithm is likely to be effective on a given instance.

\subsection{Randomized Rounding}
\Cref{alg:rank1,alg:rankr}
deterministically enumerate all $\OO(n^{2r-1})$ cells of the arrangement and search over all the corresponding candidates to find the exact optimum.
In \Cref{sec:sampling}, we theoretically analyze a randomized alternative, wherein we find the optimum among a small subset of $S$ candidates by sampling uniformly from the unit sphere. 
In this section, we empirically validate this algorithm on \textsc{Max-3-Cut} over the random $5$-regular, Erd\H{o}s--R\'enyi, and toroidal families, sweeping $n$ from $10^3$ to $10^6$
and $S$ up to $10^6$. Values are normalized by the best cut found at the largest budget ($S = 10^6$; the exact optimum $\OPTr$ is out of reach at these sizes, where enumeration costs $\OO(n^{2r-1})$), reported as the median across the $3$ to $5$ seeds.

\medskip
\noindent \emph{The knee does not move with $n$.}
\Cref{fig:sampling} plots the best-of-$S$ value against $S$ across graphs of varying sizes ($n$). We observe that across three orders of
magnitude in $n$, the plateau is reached at the same $S$. Quantifying the knee, the budget to
reach $99.9\%$ of the plateau does not grow with $n$: it is flat at $S = 10$ on toroidal
graphs (where rank-$2$ is exact) and \emph{decreasing} in $n$ on
Erd\H{o}s--R\'enyi (median falling from ${\sim}10^5$ at $n = 10^3$ to $10^4$ at $n = 10^5$). 
This is the behavior predicted by \Cref{thm:rr_tail}, whose
per-sample success probability carries no $n$-dependence, and reinforces the idea that a fixed budget often suffices.

\medskip
\noindent \emph{One draw is already good; best-of-$S$ reaches $\OPTr$.}
On instances small enough to enumerate the candidate set exhaustively ($n \le 100$) we compute
the true rank-$r$ optimum $\OPTr$ and report best-of-$S$ as a fraction of it
(\Cref{fig:sampling_quality}). \Cref{thm:expected_rr} guarantees only
$\E{f(\ctilde)} \ge (\cos^2(\pi/K)/r) \OPTr = 1/8 \OPTr$ at $K = 3, r = 2$. Yet, a
\emph{single} draw attains, in expectation, $0.87$--$0.98$ of $\OPTr$ (empirical mean over $2{\times}10^5$ i.i.d.\ draws),
about $7\times$ the floor---the bound is loose because it worst-cases the rounding over all
coordinates at once, while the sphere's mass concentrates on near-optimal cells.
Best-of-$S$
then reaches $0.98$, $0.99$, and $0.996$ of $\OPTr$ at $S = 10, 100, 1000$. We remark that these ratios
do not worsen with $n$. Rather, they improve slightly with $n$: from $0.96$ at $n = 20$ (vs.\ the true $\OPTr$) to $0.99$ by $n = 10^4$ (vs.\ the best-of-$S$ reference, since $\OPTr$ is not enumerable there) on the sparse families.
Sampling and exhaustive enumeration agree on $\OPTr$ in $24$ of $27$
instances; in three small Erd\H{o}s--R\'enyi cases sampling finds a $2$--$5\%$ \emph{better}
optimum on a lower-rank great-sphere that fixed-rank enumeration misses---a case recovered only
by also enumerating the lower-rank ($r' < r$) cells.

\medskip
\noindent \emph{The analyzed sampler matches the one used in practice.}
Our theorems analyze the uniform-sphere sampler $\ctilde \sim \sigma_{2r-1}$; the natural
alternative---the randomized form of \Cref{alg:rankr}---instead draws the null vector of a
random $(2r-1)$-subset of rows of $\Vtilde$. On the sweep the two are indistinguishable in
cut quality (uniform-sphere is $\ge$ null-vector on every family$ \times $size cell up to
$n = 10^6$, by at most $0.15\%$; \Cref{tab:samplers}), but the null-vector draw rejects ${\approx}40\%$ of
subsets as degenerate---the $2r-1$ sampled rows of $\Vtilde$ are linearly dependent, admitting no unique null direction to round---whereas uniform-sphere sampling is rejection-free.
Analyzing the
uniform-sphere sampler therefore matches practice at no empirical cost while removing the
rejection step.

\begin{figure}[t]
\centering
\includegraphics[width=\textwidth]{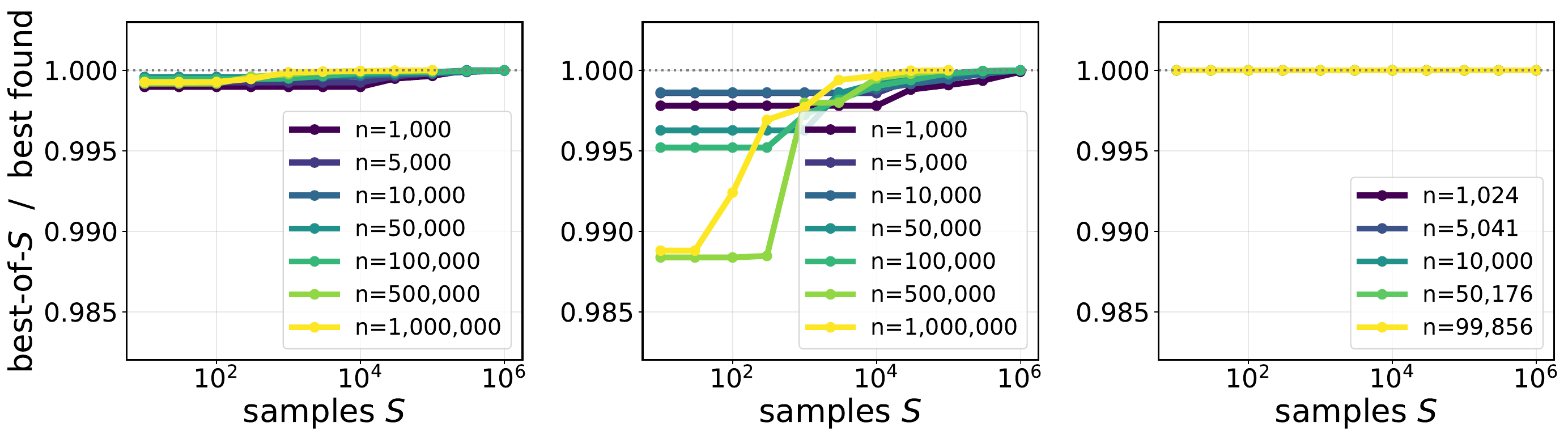}
\caption{Randomized uniform-sphere rounding on \textsc{Max-3-Cut} ($K = 3$, $r = 2$); panels,
left to right, are random $5$-regular, Erd\H{o}s--R\'enyi, and toroidal graphs.
Best-of-$S$ cut value (normalized by the best value found by either sampler at $10^6$ draws)
versus the number of samples $S$, one curve per size $n$ (up to $n = 10^6$). Curves for sizes
spanning three orders of magnitude overlap and plateau at the same $S$: the sample budget for a near-optimal cut
does not grow with $n$, as predicted by \Cref{thm:rr_tail}.}
\label{fig:sampling}
\end{figure}

\begin{figure}[t]
\centering
\includegraphics[width=0.62\textwidth]{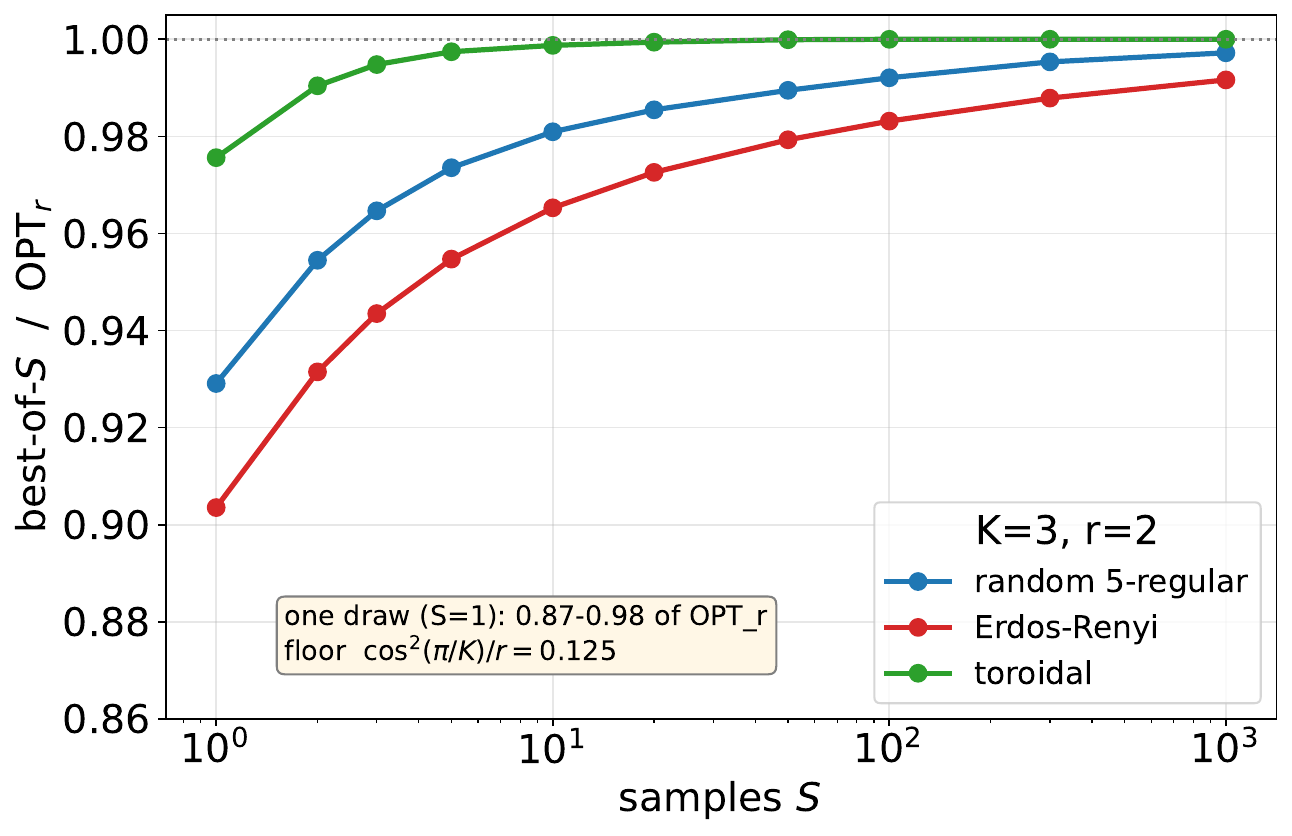}
\caption{Best-of-$S$ as a fraction of the \emph{true} rank-$r$ optimum $\OPTr$ (computed by
exhaustive cell enumeration, $n \le 100$; $K = 3$, $r = 2$), one curve per family, averaged
over sizes and seeds. A single draw ($S = 1$) attains, in expectation, $0.87$--$0.98$ of $\OPTr$---about
$7\times$ the guaranteed floor $\cos^2(\pi/K)/r = 0.125$ of \Cref{thm:expected_rr}---and
best-of-$S$ reaches $\OPTr$ within a fraction of a percent by $S \approx 100$.}
\label{fig:sampling_quality}
\end{figure}

\begin{table}[t]
\centering
\caption{Uniform-sphere (analyzed) versus null-vector (deployed) sampler at a matched
sample budget, largest size run per family, averaged over available seeds. Cut quality is
statistically identical, while uniform-sphere sampling has no infeasible draws.}
\label{tab:samplers}
\setlength{\tabcolsep}{6pt}
\begin{tabular}{@{}lrcc@{}}
\toprule
Family & $n$ & Uniform / null-vector & Null-vector feasible \\
\midrule
random $5$-regular & $1,000,000$ & $1.0000$ & $60.3\%$ \\
Erd\H{o}s--R\'enyi  & $1,000,000$ & $1.0008$ & $59.5\%$ \\
toroidal            & $99,856$     & $1.0000$ & $60.7\%$ \\
\bottomrule
\end{tabular}
\end{table}


\section{Acknowledgements}

\medskip 
\noindent \textbf{Disclosure of Generative AI Use.}
In accordance with SIAM's Policy on the Use of AI, we disclose the following.
The authors used a large language model assistant---Anthropic's Claude, via the
Claude Code CLI---in these assistive capacities during the preparation of this work:
\begin{itemize}[leftmargin=*, nosep]
  \item \textbf{Language editing:} improving the grammar, style, and clarity of
        author-written text.
  \item \textbf{Figure code:} generating and refining the Python/\textsf{matplotlib}
        scripts that render the experimental figures from the authors' own data.
  \item \textbf{Software testing:} writing unit tests and validation harnesses for
        the authors' algorithm implementations.
  \item \textbf{Proof checking:} independently verifying and stress-testing the
        mathematical proofs to surface potential errors and gaps for the authors
        to resolve.
\end{itemize}
No theorems, proofs, algorithms, experimental data, or numerical results were
generated by the AI tool; all scientific content was conceived, written, and
validated by the authors, and all AI-assisted output was reviewed and verified by
the authors. All listed authors are human and take full responsibility for the
integrity, accuracy, originality, and copyright of this submission.

\clearpage

\appendix

\crefalias{section}{appendix} 
\crefalias{subsection}{appendix} 
\crefalias{subsubsection}{appendix} 

\section{Auxiliary Lemmas}
\label{app:auxiliary_lemmas}




\begin{lemma}[Antipodal boundary equivalence]
\label{lem:antipodal}
The decision boundaries at angles $\vartheta_m$ and $\vartheta_m + \pi$ define the same hyperplane in $\HH_r$ if and only if $K$ is even.
\end{lemma}

\begin{proof}
The boundary condition at angle $\vartheta_m + \pi$ is
$\Impart\{e^{-\imag(\vartheta_m + \pi)}\alpha_i\} = -\Impart\{e^{-\imag\vartheta_m}\alpha_i\} = 0$,
which is identical to the condition at $\vartheta_m$.
This yields a reduction only if $\vartheta_m + \pi = \vartheta_{m'}$ for some $m' \in \{0, \ldots, K{-}1\}$, which requires $K/2 \in \NN$, i.e., $K$ even.
\end{proof}



\begin{lemma}[Phase equivalence]
\label{lem:phase_equiv}
For any $\bfz \in \AK^n$ and $\theta \in \AK$, $\norm{\bfV^{\dagger}\bfz}_2 = \norm{\bfV^{\dagger}(\theta\bfz)}_2$.
\end{lemma}

\begin{proof}
$\norm{\bfV^{\dagger}(\theta\bfz)}_2 = |\theta|\cdot\norm{\bfV^{\dagger}\bfz}_2 = \norm{\bfV^{\dagger}\bfz}_2$ since $|\theta| = 1$.
\end{proof}

\begin{lemma}[Auxiliary angle reduction]
\label{lem:angle_reduction}
Restricting $\phi_{2r-1} \in (-\pi/K, \pi/K]$ does not exclude any essentially distinct candidate.
Candidates from $\phi_{2r-1} \notin (-\pi/K, \pi/K]$ are phase rotations of candidates from within this range.
\end{lemma}

\begin{proof}
Any $\bfphi$ with $\phi_{2r-1} = \phi_{2r-1}^{(0)} + 2\pi k/K$ satisfies $\bfc(\bfphi) = e^{2\pi\imag k/K}\bfc(\bfphi^{(0)})$, so the decision function gives $\bfd(\bfV;\bfphi) = e^{-2\pi\imag k/K}\bfd(\bfV;\bfphi^{(0)})$.
By \cref{lem:phase_equiv}, both achieve the same objective value.
\end{proof}



\begin{corollary}[Two implies all]
\label{cor:two_implies_all}
If two distinct hyperplanes from $\GG^{(i)}$ pass through $\bfphi^{\star}$, then all $B_K$ hyperplanes from $\GG^{(i)}$ pass through $\bfphi^{\star}$.
\end{corollary}

\begin{proof}
Two conditions $\Impart\{e^{-\imag\vartheta_{m_1}}\alpha_i\} = 0$ and $\Impart\{e^{-\imag\vartheta_{m_2}}\alpha_i\} = 0$ with $m_1 \neq m_2$ require $\alpha_i = \bfV_{i,:}\bfc(\bfphi^{\star}) = 0$ (since no nonzero complex number can lie on two distinct lines through the origin).
Then $\Impart\{e^{-\imag\vartheta_m}\alpha_i\} = 0$ for all $m$.
\end{proof}

\begin{lemma}[Rank deficiency from triples]
\label{lem:rank_deficiency}
If an index set $I$ contains three or more indices from the same coordinate group, then $\mathrm{rank}(\tilde{\bfV}_{I,:}) < 2r{-}1$.
\end{lemma}

\begin{proof}
Each coordinate group $\GG^{(i)}$ spans a 2-dimensional subspace of $\RR^{2r}$ (spanned by the real and imaginary parts of $\bfV_{i,:}$).
Three rows from the same group contribute rank at most~2, not~3, to the system.
\end{proof}



\begin{lemma}[Cells adjacent to a vertex]
\label{lem:adjacent_cells}
A vertex defined by $I$ with $p$ coordinate groups contributing two hyperplanes each is adjacent to $(B_K - 1)^p$ cells.
\end{lemma}

\begin{proof}
For each ``double'' group $\GG^{(i)}$, all $B_K$ hyperplanes pass through the vertex (\cref{cor:two_implies_all}).
In the 2-dimensional normal space to the common axis of $\GG^{(i)}$, these $B_K$ hyperplanes appear as $B_K$ lines through the origin, creating $B_K$ angular sectors.
Moving away from the vertex into a cell means choosing one of the $B_K - 1$ wedges that does not cross either of the two defining hyperplanes from $I$.
The choices across distinct double groups are independent, giving $(B_K - 1)^p$ cells.
For single groups (one hyperplane in $I$), the assignment on the ``correct'' side is determined by the fixed-angle resolution.
\end{proof}

\section{Omitted Material from Section~\ref{sec:rank1}}
\subsection{Proof of Theorem~\ref{thm:rank1}}
\label{app:rank1_proof}

\begin{proof}
Since $\bfQ^{\star} = \lambda\bfq\bfq^{\dagger}$ with $\lambda > 0$, we have
$\bfz^{\dagger}\bfQ^{\star}\bfz = \lambda|\bfz^{\dagger}\bfq|^2$.
Maximizing~\eqref{eq:max_quad} is equivalent to maximizing $|\bfz^{\dagger}\bfq|^2$ over $\AK^n$.

For any complex number $d \in \CC$ and unit-modulus $c(\varphi) = e^{\imag\varphi}$, we have $\Repart(\bar{d}\cdot c(\varphi)) \leq |d|$, with equality when $c(\varphi) = d/|d|$.
Therefore:
\[
|\bfz^{\dagger}\bfq| = \max_{\varphi \in (-\pi,\pi]}  \Repart\bigl(\bfz^{\dagger}\bfq \cdot c(\varphi)\bigr).
\]
By symmetry of $\AK$ under multiplication by $e^{2\pi\imag/K}$, we can restrict $\varphi \in (-\pi/K, \pi/K]$ without loss (\Cref{lem:angle_reduction}).
Interchanging the maximizations:
\[
\max_{\bfz \in \AK^n} |\bfz^{\dagger}\bfq|
= \max_{\varphi \in (-\pi/K,\pi/K]}  \sum_{i=1}^n \max_{z_i \in \AK}  \Repart(\bar{z}_i \cdot q_i \cdot c(\varphi)).
\]
For fixed $\varphi$, the optimal $z_i$ is the element of $\AK$ whose phase is closest to $\arg(q_i \cdot c(\varphi)) = \theta_i + \varphi$, where $\theta_i = \arg(q_i)$.
Writing $z_i = e^{2\pi\imag k_i/K}$, the optimal $k_i$ satisfies
$k_i^{\star} = \argmin_{k}  |2\pi k/K - (\theta_i + \varphi)|_{\bmod\,2\pi}$.

This assignment is piecewise constant in $\varphi$, changing only when $\theta_i + \varphi$ crosses a boundary between adjacent sectors of $\AK$.
Over the interval $(-\pi/K, \pi/K]$, each $\theta_i$ produces exactly one boundary point
\[
\varphi_i = \frac{2\pi}{K}\Bigl(\frac{1}{2} + \floor{\frac{K\theta_i}{2\pi}} - \frac{K\theta_i}{2\pi}\Bigr).
\]
At $\varphi < \varphi_i$, the optimal assignment is $k_i = \floor{K\theta_i/(2\pi)}$; at $\varphi > \varphi_i$, it is $k_i = \ceil{K\theta_i/(2\pi)} \bmod K$.

Sorting $\varphi_{i_1} \leq \varphi_{i_2} \leq \cdots \leq \varphi_{i_n}$ partitions $(-\pi/K, \pi/K]$ into $n+1$ intervals.
On each interval, $\bfz(\varphi)$ is constant.
\Cref{alg:rank1} constructs all $n+1$ candidates by starting with the assignment at $\varphi = -\pi/K$ and flipping one coordinate at each boundary point.

\emph{Time complexity.}
Computing $\bfq$ from $\bfQ^{\star}$ costs $\OO(n^2)$ (a single eigenvector computation).
Computing phases $\boldsymbol{\theta}$: $\OO(n)$.
Computing boundary points $\boldsymbol{\varphi}$: $\OO(n)$.
Sorting: $\OO(n\log n)$.
Evaluating $n+1$ candidates at $\OO(n)$ each: $\OO(n^2)$.
Total: $\OO(n^2)$.
\end{proof}

\section{Omitted Material from Section~\ref{sec:rankr}}
\subsection{Proof of Theorem~\ref{thm:candidate_size_simp}} \label{app:candidate_set_proof}

\card*

\begin{proof}
We count candidates at each rank level $d = 1, \ldots, r$ and sum via the recursive decomposition $\mathcal{S}(\bfV_{:,1:d}) = \mathcal{J}(\bfV_{:,1:d}) \cup \mathcal{S}(\bfV_{:,1:d-1})$.

\emph{Counting distinct vertices at rank $d$.}
A vertex in the interior of $\HH_d$ is defined by a valid index set $I$ of size $2d{-}1$ with at most 2 indices per coordinate group.
Let $i$ groups contribute 2 hyperplanes each.
\begin{enumerate}
\item Choose which $i$ groups are ``double'': $\binom{n}{i}$ ways.
\item Choose which $2d{-}1{-}2i$ other groups are ``single'': $\binom{n-i}{2d-1-2i}$ ways.
\item Choose which boundary from each single group: naively $B_K^{2d-1-2i}$ ways, but the restriction $\phi_{2r-1} \in (-\pi/K, \pi/K]$ removes one degree of freedom, giving an effective $B_K^{2(d-i)-2}$ distinct vertices (the last angular coordinate's boundary is absorbed into the hypercube restriction; see~\cite{kyrillidis2014fixed}).
\item For double groups, no boundary choice is needed: by \Cref{cor:two_implies_all}, all $B_K$ boundaries from that group pass through the vertex regardless.
\item Each vertex with $i$ double groups is adjacent to $(B_K - 1)^i$ cells (\Cref{lem:adjacent_cells}).
\end{enumerate}
Combining:
\[
|\mathcal{J}(\bfV_{:,1:d})| = \sum_{i=0}^{d-1}\binom{n}{i}\binom{n-i}{2(d-i)-1}\,B_K^{2(d-i)-2}\,(B_K-1)^i.
\]
Summing over all $d\leq r$:
\[
|\mathcal{S}(\bfV)| = \sum_{d=1}^{r} |\mathcal{J}(\bfV_{:,1:d})| = \sum_{d=1}^{r}\sum_{i=0}^{d-1}\binom{n}{i}\binom{n-i}{2(d-i)-1}\,B_K^{2(d-i)-2}\,(B_K-1)^i.
\]

For constant $K$ (and hence constant $B_K$) and fixed $d$, the inner term is dominated by the $i=0$ term, yielding $\OO\left(n^{2d-1}\right)$. Summing over all $d\leq r$, the total size of the candidate set is $\OO(rn^{2r-1})$.
\end{proof}

\subsection{Proof of Theorem~\ref{thm:rankr}} \label{app:rankr_proof}

\rankr*

\begin{proof}
\emph{Step 1: Reformulation.}
Recall the reformulation of \eqref{eq:max_quad} as a double maximization problem:
$\max_{\bfz \in \AK^n}\norm{\bfV^{\dagger}\bfz}_2 = \max_{\bfphi \in \HH_r}\max_{\bfz \in \AK^n}\Repart\{\bfz^{\dagger}\bfV\bfc(\bfphi)\} = \max_{\bfphi \in \HH_r}\sum_{i=1}^n\max_{z_i \in \AK}\Repart\{\bar{z}_i\bfV_{i,:}\bfc(\bfphi)\}$.

\emph{Step 2: Piecewise constant decision.}
For fixed $\bfphi$, the optimal $z_i$ is the element of $\AK$ closest in phase to $\alpha_i(\bfphi) = \bfV_{i,:}\bfc(\bfphi)$.
The decision function $d_i(\bfV;\bfphi)$ changes only when $\arg(\alpha_i)$ crosses a boundary angle~$\vartheta_m$.

\emph{Step 3: Cell partition.}
The $nB_K$ hyperplanes~\eqref{eq:hyperplanes} partition $\HH_r$ into cells.
Within each cell, $\bfd(\bfV;\bfphi)$ is constant.
The objective is continuous in $\bfphi$, so the global maximum is attained either (a) at an interior vertex (intersection of $2r{-}1$ hyperplanes) or (b) on the boundary of $\HH_r$.

\emph{Step 4: Interior vertex enumeration.}
\Cref{alg:rankr} enumerates all valid index sets $I$ with $|I| = 2r{-}1$ and $\mathrm{rank}(\tilde{\bfV}_{I,:}) = 2r{-}1$.
For each, it computes the vertex $\bfphi(I)$, applies the decision rule, resolves ambiguities for used coordinates (\Cref{alg:resolve}), and evaluates the objective.
This covers all cells whose leading vertex lies in the interior of $\HH_r$.

\emph{Step 5: Boundary cases.}
Cells touching $\phi_{2r-1} = -\pi/K$ correspond to phase-rotated candidates from elsewhere (\Cref{lem:angle_reduction})---no loss.
Cells where $\phi_{2r-2} = \pm\pi/2$ correspond to the rank-$(r{-}1)$ subproblem on $\bfV_{:,1:r-1}$: fixing one angular coordinate eliminates one degree of freedom.
The recursive call handles these.
This boundary-to-recursion reduction is purely structural and independent of $K$ (the mechanism---fixing one coordinate---is the same for even and odd $K$, only the boundary value $\pm\pi/K$ differs).
See~\cite{kyrillidis2014fixed} for the original even-$K$ proof; the extension to odd $K$ is immediate since the fixed-angle procedure depends only on the hyperpolar parameterization, not on the specific value of $B_K$.

\emph{Step 6: Completeness.}
Every cell is associated with at least one candidate in $\mathcal{S}(\bfV)$ (interior vertices from Step~4, boundary from Step~5).
Since the maximizer lies in some cell, it is in $\mathcal{S}(\bfV)$.

\emph{Step 7: Complexity.}
By \Cref{thm:candidate_size_simp}, $|\mathcal{S}(\bfV)| = \OO(r\cdot n^{2r-1})$.
Each evaluation costs $\OO(n^2)$.
Total: $\OO(r\cdot n^{2r+1})$.
\end{proof}

\subsection{Ambiguity Resolution Subroutine}
\label{app:resolve}

\begin{algorithm}[H]
\caption{Ambiguity resolution for used vertices.}
\label{alg:resolve}
\begin{algorithmic}[1]
\STATE \textbf{Input:} $\tilde{\bfV}$, index set $I$, null vector $\bfc$, factor $\bfV$, candidate $\bfz$, alphabet size $K$.
\STATE $\mathcal{N}_I \leftarrow \{\floor{j/B_K} : j \in I\}$ \hfill\COMMENT{coordinate indices touched by $I$}
\FOR{each $i \in \mathcal{N}_I$ with $|\bfV_{i,:}\bfc| < \epsilon$} 
  \STATE $I' \leftarrow I \setminus \{j \in I : \floor{j/B_K} = i\}$ 
  \STATE Construct $\tilde{\bfV}' = \tilde{\bfV}_{I',:}\begin{bmatrix}\mathbf{I}_{2r-2} & \mathbf{0}\\ 0 & \sin(\pi/K)\\ 0 & \cos(\pi/K)\end{bmatrix}$ 
  \STATE Solve $\tilde{\bfV}'\tilde{\bfc}' = \mathbf{0}$ for $\tilde{\bfc}' \in \RR^{2r-2}$
  \STATE $\bfc' \leftarrow$ complex unit vector from $\tilde{\bfc}'$
  \STATE $z_i \leftarrow \argmax_{a \in \AK} \Repart(\bar{a}\cdot\bfV_{i,:}\bfc')$ 
\ENDFOR
\RETURN $\bfz$
\end{algorithmic}
\end{algorithm}

When a vertex has $p$ coordinate groups each contributing two hyperplanes, \cref{alg:resolve} is called $p$ times, once per ambiguous coordinate.
For each such coordinate, the procedure generates $B_K - 1$ distinct assignments by repeating with consecutive boundary pairs from the group (see~\cite{kyrillidis2014fixed}, case~(iii)), yielding $(B_K - 1)^p$ total candidate vectors per vertex.


\section{Proofs on Perturbed Low-Rank Matrices} \label{app:proofsOfspikedMatrix}

\subsection{Useful Facts}
In these proofs, we are going to use the following theorems.
\begin{theorem}[Wedin's Theorem \cite{wedin1972perturbation}]
\label{thm:wedin}
Let $\bfM,\widetilde{\bfM} \in \bbC^{m\times n}$ be two matrices with rank-\( r \) SVDs:
\begin{equation*}
    \begin{gathered}
        \bfM = 
        \begin{bmatrix} 
            \bfU_1 & \bfU_2 
        \end{bmatrix} 
        \begin{bmatrix} 
            \bm{\Sigma}_1 & 0 \\ 
            0 & \bm{\Sigma}_2 
        \end{bmatrix}
        \begin{bmatrix} 
            \bfV_1^\top \\ 
            \bfV_2^\top 
        \end{bmatrix};\\ \quad 
        \widetilde{\bfM} = \bfM + \bm{\Delta} = 
        \begin{bmatrix} 
            \widetilde{\bfU}_1 & \widetilde{\bfU}_2 
        \end{bmatrix} 
        \begin{bmatrix} 
            \widetilde{\bm{\Sigma}}_1 & 0 \\ 
            0 & \widetilde{\bm{\Sigma}}_2 
        \end{bmatrix}
        \begin{bmatrix} 
            \widetilde{\bfV}_1^\top \\ 
            \widetilde{\bfV}_2^\top 
        \end{bmatrix}.
    \end{gathered}
\end{equation*}
Let $\sigma_j,\widetilde{\sigma}_j$ denote the $j$th singular value of $\bfM,\widetilde{\bfM}$, respectively, and let $\bm{\Delta} = \bfM - \widetilde{\bfM}$. If $\delta = \min \left\{ \min_{1 \leq i \leq r, r+1 \leq j \leq n} |\sigma_i - \widetilde{\sigma}_j|, \min_{1 \leq i \leq r} \sigma_i \right\} > 0$, then:
\begin{align*}
    & \max\left\{\norm{\paren{\bfI - \bfU_1\bfU_1^\top}\widetilde{\bfU}_1}_2, \norm{\paren{\bfI - \bfV_1\bfV_1^\top}\widetilde{\bfV}_1}_2\right\}\\
    & \quad\quad\quad \leq \frac{\norm{\bm{\Delta}}_2}{\delta}
\end{align*}
\end{theorem}

\vershyninGaussianMatrix*


\subsection{Proof of Theorem~\ref{thm:additive}}
\label{app:additiveRatioProof}

\additiveBound*

Recall the following definitions:
\begin{itemize}
    \item $\bfQ^{\star}\in\bbC^{\dim\times \dim}$ is Hermitian PSD with eigendecomposition $\bfQ^{\star} = \sum_{i=1}^{r^\star}\lambda_i^{\star}\bfu_i^{\star}\bfu_i^{\star\dagger}$
    \item $\bfQ = \bfQ^{\star} + \bfH$ with SVD $\bfQ = \sum_{i=1}^\dim\lambda_i\bfu_i\bfv_i^\dagger$
    \item $\bfQ_r = \sum_{i=1}^r\lambda_i\bfu_i\bfu_i^{\dagger}$
    \item $\bfz^\star$ attains $\texttt{OPT-}\bfQ^{\star} = \max_{\bfz\in\mathcal{A}_K^\dim}\bfz^\dagger\bfQ^\star\bfz$
    \item $\bfz_r:=\argmax_{\bfz\in\mathcal{A}_K^\dim}\bfz^\dagger\bfQ_r\bfz$
\end{itemize}
We require the following general lemma to prove the theorems of interest.
\begin{lemma}
    \label{lem:decomp_approx_err}
    Let $\bfQ^{\star},\bfQ_r,\bfz^{\star}$, and $\bfz_r$ be defined above. Then we have that
    \[
        \left|\bfz^{\star\dagger}\bfQ^{\star}\bfz^{\star} - \bfz_r\bfQ^{\star}\bfz_r\right| \leq \dim \norm{\bfQ^{\star} - \bfQ_r}_2.
    \]
\end{lemma}
\begin{proof}[Proof of \Cref{lem:decomp_approx_err}]
    Since for all $z\in\mathcal{A}_K$ we have that $|z| = 1$, we must have that $\norm{\bfz}_2^2 = \dim$ for all $\bfz\in\mathcal{A}_K^\dim$. We decompose the quantity of interest as
    \begin{align*}
        & \left|\bfz^{\star\dagger}\bfQ^{\star}\bfz^{\star} - \bfz_r\bfQ^{\star}\bfz_r\right|\\
        &\quad\quad = \left|\bfz^{\star\dagger}\bfQ^{\star}\bfz^{\star} - \bfz_r\widehat{\bfQ}\bfz_r + \bfz_r\widehat{\bfQ}\bfz_r - \bfz_r\bfQ^{\star}\bfz_r\right|\\
        &\quad\quad \leq \left|\bfz^{\star\dagger}\bfQ^{\star}\bfz^{\star} - \bfz_r\widehat{\bfQ}\bfz_r\right| + \left|\bfz_r\widehat{\bfQ}\bfz_r - \bfz_r\bfQ^{\star}\bfz_r\right|\\
        &\quad\quad = \left|\max_{\bfz\in\mathcal{A}_K^\dim}\bfz^\dagger\bfQ^{\star}\bfz - \max_{\bfz\in\mathcal{A}_K^\dim}\bfz^\dagger\widehat{\bfQ}\bfz\right| \\
        &\quad\quad\quad + \left|\bfz_r\widehat{\bfQ}\bfz_r - \bfz_r\bfQ^{\star}\bfz_r\right|\\
        &\quad\quad \leq \max_{\bfz\in\mathcal{A}_K^\dim}\left|\bfz^\dagger\bfQ^{\star}\bfz - \bfz^\dagger\widehat{\bfQ}\bfz\right| + \left|\bfz_r\widehat{\bfQ}\bfz_r - \bfz_r\bfQ^{\star}\bfz_r\right|\\
        &\quad\quad \leq \norm{\bfQ^\star - \widehat{\bfQ}}_2\cdot \max_{\bfz\in\mathcal{A}_K^\dim}\norm{\bfz}_2^2 + \norm{\bfQ^\star - \widehat{\bfQ}}_2\norm{\bfz_r}_2^2\\
        &\quad\quad \leq 2\dim\norm{\bfQ^\star - \widehat{\bfQ}}_2
    \end{align*}
    where the last inequality follows from the fact that $\norm{\bfz}_2^2 = \dim$ for all $\bfz\in\mathcal{A}_K^\dim$.
\end{proof}

To handle the non-Hermitian noise in Theorem \ref{thm:additive}, we first prove an auxiliary lemma.
\begin{lemma}
    \label{lem:mat_est_err}
    Let $\bfQ^{\star},\bfQ,\bfQ_r,\bfH\in\bbC^{\dim\times \dim}$ be defined above. Let $\lambda_j$ be the $j$th eigenvalue of $\bfQ^{\star}$, and define $\delta^{\star} = \min\left\{\min_{j\in[r-1]}\left|\lambda_j^{\star} - \lambda_{j+1}^{\star}\right|,\lambda_r^{\star}\right\}$. If $\norm{\bfH}_2\leq \frac{\delta^{\star}}{2}$, then we have that
    \[
        \norm{\bfQ^\star - \widehat{\bfQ}}_2 \leq O\paren{\lambda_{r+1}^{\star} + \frac{\lambda_1^{\star}}{\delta^{\star}}\norm{\bfH}_2}
    \]
\end{lemma}
\begin{proof}[Proof of \Cref{lem:mat_est_err}]
    Since $\bfQ^{\star}$ is Hermitian PSD, we write its top-$r$ eigendecomposition as $\bfQ^{\star} = \bfU^{\star}\bm{\Lambda}^{\star}\bfU^{\star\dagger}$ with $\bfU^{\star}\in\bbC^{\dim\times r}$ and $\bm{\Lambda}^{\star}\in\R^{r\times r}$ being a diagonal matrix. Let $\bfu^{\star}_j$ and $\bfu_j$ denote the $j$th column of $\bfU^{\star}$ and $\bfU$, respectively. Thus, we can decompose the quantity of interest as
    \begin{equation}
        \label{eq:lem_mat_et_err_1}
        \begin{aligned}
            \norm{\bfQ^\star - \widehat{\bfQ}}_2 & \leq \norm{\bfQ^{\star} - \bfU^{\star}\bm{\Lambda}^{\star}\bfU^{\star\dagger}}_2\\
            &\quad\quad\quad + \norm{\bfU^{\star}\bm{\Lambda}^{\star}\bfU^{\star\dagger} - \bfU\bm{\Lambda}^{\star}\bfU^{\dagger}}_2\\
            &\quad\quad\quad + \norm{\bfU\bm{\Lambda}^{\star}\bfU^{\dagger} - \widehat{\bfQ}}_2
        \end{aligned}
    \end{equation}
    For the first term, since we know that $\bfU^{\star}\bm{\Lambda}^{\star}\bfU^{\star\dagger}$ is the top-$r$ truncated SVD of $\bfQ^{\star}$, we cam simply bound
    \begin{equation}
        \label{eq:lem_mat_et_err_15}
        \norm{\bfQ^{\star} - \bfU^{\star}\bm{\Lambda}^{\star}\bfU^{\star\dagger}}_2 \leq \lambda_{r+1}^{\star}
    \end{equation}
    For the third term, we have that
    \begin{equation}
        \label{eq:lem_mat_et_err_2}
        \begin{aligned}
            \norm{\bfU\bm{\Lambda}^{\star}\bfU^{\dagger} - \widehat{\bfQ}}_2 & = \norm{\bfU\paren{\bm{\Lambda}^{\star} - \bm{\Sigma}}\bfU^{\dagger}}_2\\
            & = \norm{\bm{\Lambda}^{\star} - \bm{\Sigma}}_2\\
            & = \max_{j\in[r]}\left|\bm{\Lambda}_{jj}^{\star} - \bm{\Sigma}_{jj}\right|\\
            & \leq \norm{\bfH}_2
        \end{aligned}
    \end{equation}
    where the last inequality follows from Weyl's inequality due to the fact that $\bm{\Lambda}_{jj}$ and $\bm{\Sigma}_{jj}$ are the singular values of $\bfQ^{\star}$ and $\bfQ$, respectively. For the second term, we can write
    \begin{equation}
        \label{eq:lem_mat_et_err_3}
        \begin{aligned}
            & \norm{\bfU^{\star}\bm{\Lambda}^{\star}\bfU^{\star\dagger} - \widehat{\bfQ}}_2\\
            &\quad\quad = \norm{\bfU^{\star}\bm{\Lambda}^{\star}\bfU^{\star\dagger} - \bfU\bm{\Lambda}^{\star}\bfU^{\dagger}}_2\\
            &\quad\quad = \norm{\bfU^{\star}\bm{\Lambda}^{\star}\paren{\bfU^{\star} - \bfU}^{\dagger} + \paren{\bfU^{\star} - \bfU}\bm{\Lambda}^{\star}\bfU^{\dagger}}_2\\
            &\quad\quad \leq \norm{\bfU^{\star}\bm{\Lambda}^{\star}\paren{\bfU^\star-\bfU}^{\dagger}}_2 + \norm{\paren{\bfU^{\star}-\bfU}\bm{\Lambda}^{\star}\bfU^{\dagger}}_2\\
            &\quad\quad \leq 2\norm{\bm{\Lambda}^{\star}\paren{\bfU^\star-\bfU}^\dagger}_2\\
            &\quad\quad \leq 2\lambda_1^{\star}\norm{\bfU^{\star} - \bfU}_2
        \end{aligned}
    \end{equation}
    where the second-to-last inequality is because $\max\{\norm{\bfU^{\star}}_2,\norm{\bfU}_2\} \leq 1$ due to the fact that $\bfU^{\star},\bfU$ contains orthogonal columns. Next, we will study the term $\norm{\bfU^{\star} - \bfU}_2$ using Wedin's Theorem. To start, we notice that 
    \begin{align*}
        \norm{\bfU^{\star} - \bfU}_2 & = \norm{\bfU^{\star} - \paren{\bfI - \bfU^{\star}\bfU^{\star\dagger}}\bfU + \bfU^{\star}\bfU^{\star\dagger}\bfU}_2\\
        & \leq \norm{\paren{\bfI-\bfU^{\star}\bfU^{\star\dagger}}\bfU}_2 + \norm{\bfU^{\star}\paren{\bfI - \bfU^{\star\dagger}\bfU}}_2
    \end{align*}
    Since $\bfU^{\star}$ and $\bfU$ are the top-$r$ singular vector of $\bfQ^{\star}$ and $\bfQ$, respectively, the first term can be simply bounded by Theorem~\ref{thm:wedin} using
    \[
        \norm{\paren{\bfI-\bfU^{\star}\bfU^{\star\dagger}}\bfU}_2 \leq \frac{\norm{\bfH}_2}{\delta}
    \]
    For the second term, we have that
    \begin{align*}
        \norm{\bfU^{\star}\paren{\bfI - \bfU^{\star\dagger}\bfU}}_2 & \leq \norm{\bfI - \bfU^{\star\dagger}\bfU}_2\\
        & = \max_{i\in[r]}\left|1 - \sigma_r\paren{\bfU^{\star\dagger}\bfU}\right|\\
        & = 1 - \sigma_{\min}\paren{\bfU^{\star\dagger}\bfU}
    \end{align*}
    where the first inequality is due to the fact that $\norm{\bfU}^{\star}$ is a unitary matrix, and the last inequality is because $\bfU^{\star\dagger}\bfU$ has singular values in the range $[0,1]$. Let $\bfx\in\bbC^r$ with $\norm{\bfx}_2 = 1$ be given such that $\norm{\bfU^{\star\dagger}\bfU\bfx}_2 = \sigma_{\min}\paren{\bfU^{\star\dagger}\bfU}$. Then we must have that $\norm{\bfU\bfx}_2 = 1$ since $\bfU$ contains orthonormal columns. Then we have that
    \begin{align*}
        \sigma_{\min}\paren{\bfU^{\star\dagger}\bfU} & = \norm{\bfU^{\star\dagger}\bfU\bfx}_2\\
        & = \norm{\bfU^\star\bfU^{\star\dagger}\bfU\bfx}_2\\
        & = 1 - \norm{\paren{\bfI -\bfU^{\star}\bfU^{\star\dagger}}\bfU\bfx}_2\\
        & \geq 1 - \norm{\paren{\bfI -\bfU^{\star}\bfU^{\star\dagger}}\bfU}_2
    \end{align*}
    Thus, the second term can be upper bounded by
    \begin{align*}
        \norm{\bfU^{\star}\paren{\bfI - \bfU^{\star\top}\bfU}}_2 & = 1 - \sigma_{\min}\paren{\bfU^{\star\top}\bfU}\\
        & \leq \norm{\paren{\bfI -\bfU^{\star}\bfU^{\star\top}}\bfU}_2\\
        & \leq \frac{\norm{\bfH}_2}{\delta}
    \end{align*}
    Therefore, we can conclude that $\norm{\bfU^{\star} - \bfU}_2 \leq \frac{2}{\delta}\norm{\bfH}_2$, which leads to
    \begin{equation}
        \label{eq:lem_mat_et_err_4}
        \norm{\bfU^{\star}\bm{\Lambda}^{\star}\bfU^{\star\dagger} - \widehat{\bfQ}}_2 \leq \frac{4\lambda_1^{\star}}{\delta}\norm{\bfH}_2
    \end{equation}
    by (\ref{eq:lem_mat_et_err_3}). Plugging (\ref{eq:lem_mat_et_err_15}), (\ref{eq:lem_mat_et_err_4}) and (\ref{eq:lem_mat_et_err_2}) back into (\ref{eq:lem_mat_et_err_1}) gives
    \[
        \norm{\bfQ^{\star} - \widehat{\bfQ}}_2\leq \lambda_{r+1}^{\star} + \paren{1 + \frac{4\lambda_1^{\star}}{\delta}}\norm{\bfH}_2
    \]
    Define $\delta^{\star} = \min\left\{\min_{j\in[r-1]}\left|\lambda_j^{\star} - \lambda_{j+1}^{\star}\right|,\lambda_r^{\star}\right\}$. By Weyl's inequality, we have that $\left|\lambda_i^{\star} - \sigma_i\right| \leq \norm{\bfH}_2$. Thus, if $\norm{\bfH}_2\leq \frac{\delta^{\star}}{2}$, we can guarantee that $\sigma_i \in\left[\sfrac{\lambda_{i-1}^{\star} + \lambda_i^{\star}}{2},\sfrac{\lambda_i^{\star} + \lambda_{i+1}^{\star}}{2}\right]$, which implies that
    \begin{align*}
        & \min_{j\in[r-1]}\left|\lambda_j^{\star}-\sigma_i\right|\\
        &\quad\quad  \geq \min\left\{\left|\lambda_{i-1}^{\star} - \sigma_i\right|,\left|\lambda_{i+1}^{\star} - \sigma_i\right|\right\}\\
        & \quad\quad\geq \min\left\{\left|\lambda_{i-1}^{\star} - \lambda_{i}^{\star}\right|,\left|\lambda_{i+1}^{\star} - \lambda_{i}^{\star}\right|\right\} - \norm{\bfH}_2\\
        & \quad\quad\geq \delta^{\star} - \norm{\bfH}_2\\
        &\quad\quad \geq \frac{\delta^{\star}}{2}
    \end{align*}
    for all $i\in[r]$. This gives that $\delta \geq \frac{\delta^{\star}}{2}$, which implies that
    \[
        \norm{\bfQ^{\star} - \widehat{\bfQ}}_2\leq \lambda_{r+1}^{\star} + \paren{1 + \frac{8\lambda_1^{\star}}{\delta^{\star}}}\norm{\bfH}_2
    \]
    Notice that $\delta^{\star}\leq \lambda_1$. Therefore, we must have that $\frac{\lambda_1^{\star}}{\delta^{\star}}\geq 1$, which implies that $1 + \frac{8\lambda_1^{\star}}{\delta^\star} \leq O\paren{\frac{\lambda_1^{\star}}{\delta^\star}}$. This concludes the proof.
\end{proof}
Now, we are ready to prove Theorem~\ref{thm:additive}.
\begin{proof}[Proof of Theorem~\ref{thm:additive}]
    Under the provided condition, we can apply Lemma~\ref{lem:mat_est_err} to obtain that
    \[
        \norm{\bfQ^\star - \widehat{\bfQ}}_2 \leq O\paren{\lambda_{r+1}^{\star} + \frac{\lambda_1^{\star}}{\delta^{\star}}\norm{\bfH}_2}
    \]
    Plugging the bound into Lemma~\ref{lem:decomp_approx_err} gives
    \[
        \left|\bfz^{\star\dagger}\bfQ^{\star}\bfz^{\star} - \bfz_r\bfQ^{\star}\bfz_r\right| \leq O\paren{\dim\paren{\lambda_{r+1}^{\star} + \frac{\lambda_1^{\star}}{\delta^{\star}}\norm{\bfH}_2}}
    \]
\end{proof}

\subsection{Proof of Corollary~\ref{cor:approx_ratio_perturbation_herm}}\label{app:proof_cor_approx}
The following lemma provides an upper bound of the approximation error when the perturbation $\bfH$ (and consequently $\bfQ$) is Hermitian.
\begin{lemma}
    \label{lem:mat_est_err_herm}
    Let $\bfQ^{\star},\bfQ,\bfQ_r,\bfH\in\bbC^{\dim\times \dim}$ be defined above. Let $\lambda_j$ be the $j$th eigenvalue of $\bfQ^{\star}$. If $\bfH$ is Hermitian, then we have that
    \[
        \norm{\bfQ^\star - \widehat{\bfQ}}_2 \leq O\paren{\lambda_{r+1}^{\star} + \frac{\lambda_1^{\star}}{\delta^{\star}}\norm{\bfH}_2}.
    \]
\end{lemma}
\begin{proof}[Proof of \Cref{lem:mat_est_err_herm}]
    Since $\bfH$ is Hermitian, we must have that $\bfQ$ is Hermitian. Recall that we defined $\bfQ_r^{\star} = \bfU^{\star}\bm{\Lambda}\bfU^{\star\dagger}$ as the top-$r$ eigendecomposition of $\bfQ^{\star}$, where $\bfU^{\star}\in\bbC^{\dim\times r}$ and $\bm{\Lambda}^{\star}\in\R^{r\times r}$. Then we have that
    \begin{align*}
        \norm{\bfQ^{\star} - \widehat{\bfQ}}_2 & = \norm{\bfQ^{\star} - \bfQ + \bfQ - \widehat{\bfQ}}_2\\
        & \leq \norm{\bfQ^{\star} - \bfQ}_2 + \norm{\bfQ - \widehat{\bfQ}}_2\\
        & = \norm{\bfH}_2 + \norm{\bfQ - \widehat{\bfQ}}_2.
    \end{align*}
    We can easily notice that $\widehat{\bfQ}$ is the top-$r$ eigendecomposition of $\bfQ$, since $\bfU = \bfV$ due to the fact that $\bfQ$ is Hermitian. Therefore, we have that $\norm{\bfQ - \widehat{\bfQ}}_2 = \lambda_{r+1}\paren{\bfQ}$. By Weyl's inequality, we have that
    \[
        \left|\lambda_{r+1}\paren{\bfQ} - \lambda_{r+1}^{\star}\right| \leq \norm{\bfH}_2.
    \]
    Therefore, we can conclude that $\lambda_{r+1}\paren{\bfQ}\leq \lambda_{r+1}^\star + \norm{\bfH}_2$. This gives that
    \[
        \norm{\bfQ - \widehat{\bfQ}}_2 \leq \lambda_{r+1}^\star + \norm{\bfH}_2,
    \]
    which implies that
    \[
        \norm{\bfQ^{\star} - \widehat{\bfQ}}_2 \leq \lambda_{r+1}^\star + 2\norm{\bfH}_2 = O\paren{\lambda_{r+1}^\star + \norm{\bfH}_2}.
    \]
\end{proof}
Now, we are ready to prove Corollary \ref{cor:approx_ratio_perturbation_herm}, which we restate below.

\approxRatioHerm* 

\begin{proof}[Proof of Corollary \ref{cor:approx_ratio_perturbation_herm}]
    Under the provided condition, we can apply Lemma~\ref{lem:mat_est_err_herm} to obtain that
    \[
        \norm{\bfQ^\star - \widehat{\bfQ}}_2 \leq O\paren{\lambda_{r+1}^\star + \norm{\bfH}_2}.
    \]
    Plugging the bound into Lemma~\ref{lem:decomp_approx_err} gives
    \[
        \left|\bfz^{\star\dagger}\bfQ^{\star}\bfz^{\star} - \bfz_r\bfQ^{\star}\bfz_r\right| \leq O\paren{\dim\paren{\lambda_{r+1}^\star + \norm{\bfH}_2}}.
    \]
\end{proof}

\subsection{Proof of Theorem~\ref{thm:mult}}
\label{app:multRatioProof}

\multBound* 

\begin{proof}[Proof of Theorem~\ref{thm:mult}]
By Theorem~\ref{thm:additive}, we know 
\begin{equation*}
    \zst^\dagger \Qst \zst - \zr^\dagger \Qst \zr \leq O \paren{\dim \paren{\lambda_{r+1}^\star + \frac{\lambda_1^\star}{\delta^\star} \norm{\bfH}_2}}.
\end{equation*} 

We can rearrange this bound to get a bound on the approximation ratio in terms of $\optst$ as follows:
\begin{equation}
\label{eqn:approxRatWithOptStar}
    \frac{\zr^\dagger \Qst \zr}{\optst} \geq 1 - \frac{O \paren{\dim \paren{\lambda_{r+1}^\star + \frac{\lambda_1^\star}{\delta^\star} \norm{\bfH}_2}}}{\optst}.
\end{equation}

\def\QstZst{\ensuremath{\mathbf{z}^{\star \dagger} \mathbf{Q}^\star \mathbf{z}^\star}}
\def\zhat{\ensuremath{\mathbf{\hat{z}}}}

We now focus on lower bounding $\optst = \zst^\dagger \Qst \zst$.
To do so, we introduce $\zhat \in \AA_K^\dim$. For now, we can treat $\zhat$ as an arbitrary vector in $\AA_K^\dim$. In Equation~\eqref{eqn:defkell}, we show that we can choose $\zhat$ to satisfy certain nice properties that allow us to lower bound $\zhat^\dagger \Qst \zhat$.

We can rewrite $\zhat^\dagger \Qst \zhat$ in terms of the decomposition of $\Qst$ as follows.
\begin{align*}
    \zhat^\dagger &\Qst \zhat  = \zhat^\dagger \left( \sum_{i=1}^r \lambda^\star_i \mathbf{u}_i^\star \mathbf{u}_i^{\star \dagger} \right) \zhat \\
    &= \sum_{i=1}^r \lambda^\star_i  \left(\zhat^\dagger \mathbf{u}_i^\star \right) \left(\mathbf{u}_i^{\star \dagger} \zhat \right)  \\
    &= \sum_{i=1}^r \lambda^\star_i  \abs{\zhat^\dagger \mathbf{u}_i^\star}^2  \\
    &\geq \lambda^\star_1 \cdot \abs{\zhat^\dagger \mathbf{u}_1^\star}^2.
\end{align*}
We now focus on bounding $\abs{\zhat^\dagger \mathbf{u}_1^\star}^2$:
\begin{align*}
    \abs{\zhat^\dagger \mathbf{u}_1^\star}^2
    &= \text{Re}\left( \zhat^\dagger \mathbf{u}_1^\star \right)^2 + \text{Im}\left( \zhat^\dagger \mathbf{u}_1^\star \right)^2 \\
    &= \left[ \text{Re}\left( \sum_{\ell=1}^\dim \bar{\hat{z}}_\ell u_{1,\ell}^\star \right) \right]^2 + \left[ \text{Im}\left( \sum_{\ell=1}^\dim \bar{\hat{z}}_\ell u_{1,\ell}^\star \right) \right]^2 \\
    &= \left[ \sum_{\ell=1}^\dim \text{Re}\left( \bar{\hat{z}}_\ell u_{1,\ell}^\star \right) \right]^2 + \left[ \sum_{\ell=1}^\dim \text{Im}\left( \bar{\hat{z}}_\ell u_{1,\ell}^\star \right) \right]^2. \\
\end{align*}
Each entry $u_{1,\ell}^\star$ of $u_{1}^\star$ may be written in polar form as \begin{equation*}
    u_{1,\ell}^\star~=~r_\ell \exp\left(j \theta_\ell \right),
\end{equation*}
where $r_\ell = | u_{1,\ell}^\star | \in \mathbb{R}$ and $\theta_\ell \in [0, 2\pi)$. Additionally, because $\zhat \in \AA_K^\dim$, $z^\star_\ell = \exp\left(j k_\ell \pi / K\right)$ for some $k_\ell \in \{0, 1, \ldots, K-1\}$.
We can re-express the quantity under consideration as:
\begin{align*}
    \abs{\zhat^\dagger \mathbf{u}_1^\star}^2
    &= \left[ \sum_{\ell=1}^\dim r_\ell \text{ Re}\left( \exp \left(j \left( \theta_\ell - \frac{2k_\ell \pi}{K} \right) \right) \right) \right]^2  \\
    &\quad\quad + \left[ \sum_{\ell=1}^\dim r_\ell \text{ Im}\left( \exp \left(j \left( \theta_\ell - \frac{2k_\ell \pi}{K} \right) \right) \right) \right]^2 \\
    &= \left[ \sum_{\ell=1}^\dim r_\ell \cos \left(  \theta_\ell - \frac{2k_\ell \pi}{K} \right) \right]^2 \\
    &\quad\quad + \left[ \sum_{\ell=1}^\dim r_\ell \sin \left(  \theta_\ell - \frac{2k_\ell \pi}{K} \right) \right]^2.
\end{align*}

\noindent Now, define $k_{\ell} \in \{0, \dots, K-1\}$ to be the unique integer such that $\theta_{\ell}$ lies in the $k_{\ell}$-th sector:\begin{equation}\label{eqn:defkell}\theta_{\ell} \in \left[ \frac{(2k_{\ell}-1)\pi}{K}, \frac{(2k_{\ell}+1)\pi}{K} \right) \pmod{2\pi}.\end{equation}
Let $\theta_{\ell}' = \paren{\theta_{\ell} - \frac{2k_{\ell}\pi}{K}} \text{ mod } 2\pi$. By the definition of $k_{\ell}$, we ensure that the rotated angle is close to zero, specifically $\theta_{\ell}'\in \left[0, \frac{\pi}{K}\right) \cup \left[\frac{(2K-1)\pi}{K}, 2\pi\right)$ for all $\ell\in[\dim]$. Consequently, the cosine of the offset angle is bounded below by the cosine of the maximum deviation $\pi/K$:$$    \cos\paren{\theta_{\ell}-\frac{2k_{\ell}\pi}{K}} \in \left[\cos\left(\frac{\pi}{K}\right), 1\right].$$
This implies that
$$ \abs{\zhat^\dagger \mathbf{u}_1^\star}^2 \geq \paren{\sum_{\ell=1}^\dim r_{\ell} \cos\left(\frac{\pi}{K}\right)}^2 = \cos^2\left(\frac{\pi}{K}\right)\norm{\bfu_1^\star}_1^2 \geq \frac{1}{4}\norm{\bfu_1^\star}_1^2, $$
where the final inequality holds for all $K \geq 3$.
In order to lower bound $\norm{\bfu_1^\star}_1$, we notice that
\[
    1 = \norm{\bfu_1^\star}_2^2 \leq \norm{\bfu_1^\star}_1\norm{\bfu_1^\star}_{\infty}
\]
By the assumption that $\norm{\bfu_1^\star}_{\infty}\leq \frac{\mu}{\sqrt{\dim}}$,we have that
\[
    \norm{\bfu_1^\star}_1 \geq \frac{1}{\norm{\bfu_1^\star}_{\infty}} \geq \frac{\sqrt{\dim}}{\mu}
\]
Ultimately, we have shown that 
\begin{equation*}
    \zhat^\dagger \Qst \zhat \geq \frac{1 \dim}{4\mu^2} \lambda^\star_1.
\end{equation*}
Because $\zst$ maximizes $\bfz^\dagger \Qst \bfz$, we know
\begin{equation*}
    \QstZst \geq \zhat^\dagger \Qst \zhat.
\end{equation*}
This allows us to bound $\optst$ as follows:
\begin{equation*}
    \optst = \QstZst \geq \frac{\lambda_1^\star \dim}{4\mu^2}.
\end{equation*}
Substituting this bound back into (\ref{eqn:approxRatWithOptStar}) yields
\begin{align*}
    \frac{\zr^\dagger \Qst \zr}{\optst} &\geq 
    1 - \frac{O \paren{\dim \paren{\lambda_{r+1}^\star + \frac{\lambda_1^\star}{\delta^\star} \norm{\bfH}_2}}}{\lambda^\star_1 \dim}\\
    &\geq 1- O\paren{ \frac{\lambda_{r+1}^\star}{\lambda_1^\star} + \frac{\norm{\bfH}_2}{\delta^\star}}.
\end{align*} 
\end{proof}


\section{Proofs of Sampling-Based Guarantees (Section~\ref{sec:approx})}
\label{app:proofsOfSampling}

\subsection{Definitions}

\begin{definition}[Spectral embedding and pairwise angle]
\label{def:pair_angle}
For $i \in [n]$, let $\bfv_i \in \bbC^r$ denote the $i$-th row of $\bfV$
viewed as a column vector (\emph{same entries}, no conjugation); equivalently,
$\bfV_{i,:} = \bfv_i^\top$. Thus $y_i := \bfV_{i,:}\bfc = \bfv_i^\top \bfc$,
and the rank-1 specialization $\bfV \in \bbC^{n\times 1}$ has
$\bfv_i = V_{i,1} \in \bbC$. We call $\bfv_i$ the \emph{spectral embedding} of vertex $i$.
For $i \neq j$, define
\begin{align}
 \theta_{ij} &:= \arg\bigl(\langle \bfv_j, \bfv_i \rangle\bigr)
 \in (-\pi, \pi],
 \label{eq:def_theta} \\
 \rho_{ij} &:= \frac{\abs{\langle \bfv_j, \bfv_i \rangle}}
 {\norm{\bfv_i}_2 \norm{\bfv_j}_2}
 \in [0, 1].
 \label{eq:def_rho}
\end{align}
We call $\theta_{ij}$ the \emph{complex pairwise angle} and
$\rho_{ij}$ the \emph{magnitude ratio} between vertices $i$ and $j$.
\end{definition}
 
\medskip \noindent \textbf{Convention.}
The argument order $\langle \bfv_j, \bfv_i \rangle$ (not
$\langle \bfv_i, \bfv_j \rangle$) is chosen so that
$\theta_{ij} = \theta_i - \theta_j$ in the rank-$1$
specialization below, matching the convention used throughout the proofs that follow.


\begin{restatable}[Bohnenblust-Sobczyk]{fact}{complexHahnBanach}
\label{fct:complexHahnBanach}
    Let $X$ be a complex vector space and $p: X \to [0,\infty)$ a
    seminorm. For every $x_0 \in X$ with $p(x_0) = M > 0$, there exists a
    $\bbC$-linear functional $\ell: X \to \bbC$ such that
    \[
     |\ell(x)| \leq p(x) \quad \text{for all } x \in X,
     \qquad \ell(x_0) = M.
    \]
\end{restatable}


\subsection{Proof of \Cref{thm:rr_tail}}
\label{app:proof-beta-tail}
\betaTail*

\begin{proof}

Fix $\bfc \in \mathbb{S}^{2r-1}$ and write $y_i := (\bfV\bfc)_i \in \bbC$ for
the $i$-th coordinate of $\bfV\bfc$. Since $z_i = z_i(\bfc)$ is the
nearest $K$-th root of unity to $y_i/|y_i|$, the rounding residual
$\psi_i := \arg z_i - \arg y_i$ satisfies $|\psi_i| \leq \pi/K$, and
\begin{equation}
\label{eq:rounding_pointwise}
 \Re(\overline{y_i} z_i) = |y_i|\cos\psi_i \geq |y_i|\cos(\pi/K).
\end{equation}
Now, by Cauchy--Schwarz,
\begin{align*}
 \|\bfV^\dagger \bfz\|_2
 & = \max_{\|\bfu\|_2 = 1}\bigl|\bfu^\dagger \bfV^\dagger \bfz\bigr|
 \geq \bigl|\bfc^\dagger \bfV^\dagger \bfz\bigr|
 = \bigl|(\bfV\bfc)^\dagger \bfz\bigr| \\
 & \geq \Re\bigl((\bfV\bfc)^\dagger \bfz\bigr)
 = \sum_{i=1}^n \Re(\overline{y_i} z_i)
 \stackrel{\eqref{eq:rounding_pointwise}}{\geq}
 \cos(\pi/K)\sum_{i=1}^n |y_i|
 = \cos(\pi/K) \cdot N(\bfc),
\end{align*}
where we introduce the shorthand $N(\bfc) := \|\bfV\bfc\|_1$.
Squaring,
\begin{equation}
\label{eq:f_pointwise}
 f(\bfc) = \|\bfV^\dagger \bfz\|_2^2 \geq \cos^2(\pi/K) N(\bfc)^2.
\end{equation}
Note that $N$ is a norm on $\bbC^r$, as because $\bfV$ has full column rank, we have
$N(\bfc) = 0 \Leftrightarrow \bfV\bfc = 0 \Leftrightarrow \bfc = 0$.
Positive homogeneity and the triangle inequality are immediate from
the $\ell^1$ norm on $\bbC^n$ post-composed with the linear map $\bfV$. 

Further, $N$ is continuous on the compact set $\{\bfc : \|\bfc\|_2 = 1\}$, so
$M := \sup_{\|\bfc\|_2 = 1} N(\bfc) < \infty$ is attained at some
$\bfc_0 \in \mathbb{S}^{2r-1}$ with $N(\bfc_0) = M$. Applying \Cref{fct:complexHahnBanach} with
$X = \bbC^r$, $p = N$, $x_0 = \bfc_0$, we know that there exists a $\bbC$-linear
functional $\ell(\bfc) = \bfa^\dagger \bfc$ for some $\bfa \in \bbC^r$
satisfying
\begin{equation}
\label{eq:HB_dominance}
 |\bfa^\dagger \bfc| \leq N(\bfc) \quad \text{for all } \bfc \in \bbC^r,
 \qquad \bfa^\dagger \bfc_0 = M.
\end{equation}
For every $\bfc \in \SC{r}$,
\eqref{eq:f_pointwise} gives
\begin{equation}
\label{eq:f-in-terms-of-N}
    f(\bfc) \geq \cos^2(\pi/K) N(\bfc)^2, \qquad N(\bfc) = \|\bfV\bfc\|_1.
\end{equation}
Applying \eqref{eq:HB_dominance}
at $\bfc = \bfc_0$ and Cauchy--Schwarz, we can extract a lower bound on $\|\bfa\|_2$:
\[
 M = |\bfa^\dagger \bfc_0|
 \leq \|\bfa\|_2 \|\bfc_0\|_2 = \|\bfa\|_2,
 \qquad \text{i.e.,}\quad \|\bfa\|_2 \geq M.
\]
Set $\bfu$ to be the unit vector
\[
 \bfu := \bfa / \|\bfa\|_2 \in \mathbb{S}^{2r-1}.
\]
Observe that for {every} unit vector $\bfc \in \SC{r}$,
\begin{equation}
\label{eq:N_pointwise_via_HB}
 N(\bfc)
 \stackrel{\eqref{eq:HB_dominance}}{\geq} |\bfa^\dagger \bfc|
 = \|\bfa\|_2 \cdot |\bfu^\dagger \bfc|
 \geq M \cdot |\bfu^\dagger \bfc|.
\end{equation}
Further, for any $\bfz \in \mathcal{A}_K^n$ (so $|z_i| = 1$),
\begin{align*}
 \|\bfV^\dagger\bfz\|_2
 &= \max_{\|\bfc\|_2 = 1} \bigl|\bfc^\dagger \bfV^\dagger \bfz\bigr|
 = \max_{\|\bfc\|_2 = 1} \bigl|(\bfV\bfc)^\dagger \bfz\bigr| \\
 &\leq \max_{\|\bfc\|_2 = 1} \sum_{i=1}^n |(\bfV\bfc)_i| \cdot |z_i|
 = \max_{\|\bfc\|_2 = 1} \|\bfV\bfc\|_1
 = M.
\end{align*}
Squaring and taking the maximum over $\bfz$, we know
\begin{equation}
\label{eq:OPT_M_bound}
 \mathrm{OPT}_r = \max_{\bfz \in \mathcal{A}_K^n} \|\bfV^\dagger\bfz\|_2^2
 \leq M^2.
\end{equation}
Combining \eqref{eq:f-in-terms-of-N}, \eqref{eq:N_pointwise_via_HB} and \eqref{eq:OPT_M_bound}, we have for a fixed unit vector $\bfu$:
\begin{equation}
\label{eq:f-LB-postHB}
    f(\bfc) \geq \cos^2(\pi / K) \cdot {\abs{\bfu^\dagger \bfc}^2} \cdot \OPTr.
\end{equation}
By \Cref{lem:innerIsBeta}, we know that ${\abs{\bfu^\dagger \bfc}^2} \sim \mathrm{Beta}(1, r-1)$, implying
\begin{align}
    \Pr{{\abs{\bfu^\dagger \bfc}^2} \geq 1 - \varepsilon}
    = \varepsilon^{r-1}.
\end{align}
Consequently, we have
\begin{align}
    \Pr{f(\bfc) \geq \cos^2 (\pi / K) \cdot (1 - \varepsilon) \cdot \OPTr}
    \geq \varepsilon^{r-1}.
\end{align}
Now consider $S$ i.i.d.\ samples $\bfc_1, \ldots, \bfc_S \sim
\mathrm{Unif}(\SC{r})$. Let $A_s := \{f(\ctilde_s) \geq \cos^2(\pi/K)(1-\varepsilon)\mathrm{OPT}_r\}$.
By~\eqref{eq:rr_tail_main}, $\Pr{A_s} \geq \varepsilon^{r-1}$, so the
failure probability (no sample meets the bound) is
\[
 \Pr{\bigcap_{s=1}^S A_s^{\mathrm{c}}}
 = \prod_{s=1}^S \Pr{A_s^{\mathrm{c}}}
 \leq (1 - \varepsilon^{r-1})^S
 \leq e^{-S \varepsilon^{r-1}},
\]
using independence and the standard inequality $1 - t \leq e^{-t}$.
Setting this $\leq \delta$ and solving for $S$, we find
\[
 e^{-S \varepsilon^{r-1}} \leq \delta
 \Longleftrightarrow
 S \varepsilon^{r-1} \geq \log(1/\delta)
 \Longleftrightarrow
 S \geq \frac{\log(1/\delta)}{\varepsilon^{r-1}}.
\]
\end{proof}


\subsection{Proof of \Cref{thm:tail_n_dep}}
\label{app:proof-pz-tail}

We will use the following fact throughout the proof.
\begin{fact}[Paley--Zygmund Inequality]
\label{lem:pz_inequality}
Let $f \geq 0$ be a non-negative random variable with
$\E{f^2} < \infty$ and $\E{f} > 0$. For all $\varepsilon \in (0, 1)$,
\begin{equation}
\label{eq:pz_inequality}
 \Pr{f \geq (1-\varepsilon) \E{f}}
 \geq 
 \varepsilon^2 \cdot \frac{(\E{f})^2}{\E{f^2}}
 = 
 \frac{\varepsilon^2}{1 + V_0},
\end{equation}
where $V_0 := \Var{f} / (\E{f})^2$.
\end{fact}

\pzBasedTail*

\begin{proof}
    
The random variable
$f(\bfc) = \norm{\bfV^\dagger \bfz(\bfc)}_2^2$ is non-negative by construction. It is also bounded: $\norm{\bfz(\bfc)}_2^2 = n$
since $|z_i| = 1$ for each $i$, and
$\norm{\bfV^\dagger \bfz}_2^2 \leq \norm{\bfV^\dagger}_2^2 \norm{\bfz}_2^2
= \lambda_1 \cdot n$. Hence $\E{f^2} \leq (\lambda_1 n)^2 < \infty$. 
We also note $\E{f} > 0$ since $f \geq \cos^2(\pi/K) \OPTr / r > 0$ by \Cref{thm:expected_rr} (and $\OPTr > 0$ for all non-trivial $\bfQ$). Therefore, we can apply \Cref{lem:pz_inequality} (Paley-Zygmund), resulting in:
\begin{equation}
     \Pr[\bfc \sim \mathrm{Unif}(\SC{r})]{f(\bfc) \geq (1-\varepsilon) \E{f(\bfc)}} \geq \frac{\varepsilon^2}{1 + V_0}.
\end{equation}
Recall that $V_0$ is defined as $\Var{f} / (\E{f})^2$. We will bound the numerator and denomitor of $V_0$ separately. First, by \Cref{lem:second_moment}, we have
$\Var{f} \leq c n^2 \Lambda^2$, where $\Lambda = \sum_{i=1}^r \lambda_i$ and $c$ is a constant. We can bound $\E{f}^2$ by combining \Cref{lem:opt_lower,thm:expected_rr}:
\begin{align}
 \E{f} &\geq 
 \frac{\cos^2(\pi/K)}{r} \OPTr
 \geq
 \frac{\cos^2(\pi/K)}{r}
 \cdot
 \frac{\cos^2(\pi/K) \Lambda n}{r^2 \mu^2}
 = 
 \frac{\cos^4(\pi/K) \Lambda n}{r^3 \mu^2},
 \label{eq:thm_tail_Ef_lower}
\end{align}
where the first inequality is due to \Cref{lem:opt_lower} and the second is due to \Cref{thm:expected_rr}.
Squaring both sides, we have:
\begin{equation}
\label{eq:thm_tail_Ef2_lower}
 (\E{f})^2
 \geq 
 \frac{\cos^8(\pi/K) \Lambda^2 n^2}{r^6 \mu^4}.
\end{equation}
Combining the bounds on $\Var{f}$ and $\E{f}^2$, we have \begin{equation}
\label{eq:thm_tail_V0_bound}
 V_0
 = 
 \frac{\Var{f}}{(\E{f})^2}
 \leq
 \frac{c n^2 \Lambda^2}
 {\cos^8(\pi/K) \Lambda^2 n^2 / (r^6 \mu^4)}
 =
 \frac{c\; r^6 \mu^4}{\cos^8(\pi/K)}.
\end{equation}
Substituting this bound into the Paley-Zygmund Inequality, we have:
\begin{equation}
    \Pr[\bfc \sim \mathrm{Unif}(\SC{r})]{f(\bfc) \geq (1-\varepsilon) \E{f(\bfc)}} \geq \frac{\varepsilon^2}{1 + V_0} \geq \frac{\varepsilon^2}{1 + \frac{c \; r^6 \mu^4}{\cos^8 (\pi / K)}}.
\end{equation}
By the monotonicity of probability,
\begin{equation}
\label{eq:thm_tail_step3}
 \Pr{f \geq (1-\varepsilon)\frac{\cos^2(\pi/K)}{r}\OPTr}
 \geq 
 \Pr{f \geq (1-\varepsilon)\E{f}} \geq \frac{\varepsilon^2}{1 + \frac{c\; r^6 \mu^4}{\cos^8 (\pi / K)}}.
\end{equation}
By the same argument given in the proof of the sample complexity of \Cref{thm:rr_tail}, in order to guarantee with probability at least $1-\delta$ that at least one sampled candidate gives a score within $(1-\varepsilon) {\cos^2(\pi / K)} / r$ of $\OPTr$, we must choose
\begin{equation}
    S \geq \left( 1+ \frac{c r^6 \mu^4}{\cos^8(\pi/K)} \right) \cdot \frac{\log(1/\delta)}{\varepsilon^2}.
\end{equation}

\end{proof}
\subsection{Additional Lemmas from \Cref{app:proofsOfSampling}}
\subsubsection{Inner Product Follows Beta distribution}

\begin{restatable}{lemma}{innerIsBeta}
\label{lem:innerIsBeta}
    Let $\bfc \sim \mathrm{Unif}(\SC{r})$. For any fixed unit vector $\bfu \in \bbC^r$,
    \begin{equation}
    \label{eq:sphere_uniform_second_moment}
     \E{|\bfu^\dagger \bfc|^2} = \frac{1}{r},
     \qquad
     |\bfu^\dagger \bfc|^2 \sim \mathrm{Beta}(1, r-1)
     \text{ on } [0,1].
    \end{equation}
\end{restatable}

\begin{proof}
We can sample $\bfc$ as $\bfc = \bfZ / \|\bfZ\|_2$
where $\bfZ \in \bbC^r$ is drawn from a standard complex Gaussian distribution. By unitary invariance of the law of $\bfZ$,
we may rotate so that $\bfu = \mathbf{e}_1$, giving
$|\bfu^\dagger \bfc|^2 = |Z_1|^2 / \|\bfZ\|_2^2$. The numerator
$|Z_1|^2 = (X_1^2 + Y_1^2)/2 \sim \mathrm{Gamma}(1, 1)$ (i.e., Exp$(1)$),
and the orthogonal complement $\sum_{i \geq 2} |Z_i|^2 \sim \mathrm{Gamma}(r-1, 1)$,
independent of $|Z_1|^2$. The ratio
\[
 \frac{\mathrm{Gamma}(1, 1)}{\mathrm{Gamma}(1, 1) + \mathrm{Gamma}(r-1, 1)} \sim \mathrm{Beta}(1, r-1)
\]
is a standard identity (Beta--Gamma representation). Its mean is
$1/(1+(r-1)) = 1/r$. The
PDF, CDF, and tail used later are
\[
 f_X(x) = (r-1)(1-x)^{r-2}, \quad
 F_X(x) = 1 - (1-x)^{r-1}, \quad
 \PP[X \geq 1-\varepsilon] = \varepsilon^{r-1},
 \qquad x \in [0,1].
\]
\end{proof}

\subsubsection{Variance Bound}

\begin{lemma}[Variance of $f$ under uniform sphere]
\label{lem:second_moment}
Under \Cref{asm:incoherence} ($\mu$-incoherence,
$\norm{\bfV_{i,:}}^2 \leq \mu^2 \Lambda/n$),
\begin{equation}
\label{eq:var_bound}
 \Var[\bfc \sim \mathrm{Unif}(\SC{r})]{f(\ctilde)}
 \leq c \mu^4 \Lambda^2 n^2,
\end{equation}
where $c$ is an absolute constant,
and $\Lambda = \sum_{j=1}^r \lambda_j$. Moreover, with $\lambda_1$ the top
eigenvalue of $\bfQ_r$, the \emph{incoherence-free} Frobenius form holds with
absolute constant $1$:
\begin{equation}
\label{eq:var_bound_frob}
 \Var[\bfc \sim \mathrm{Unif}(\SC{r})]{f(\bfc)}
 \leq n^2 \lambda_1 \Lambda \leq n^2 \Lambda^2 .
\end{equation}
\end{lemma}

\begin{proof}
We bound $\Var{f}$ in four steps. The strategy is to center $f$ at its
mean (killing the deterministic diagonal contribution), expand the
variance as a 4-tuple sum, and split by the overlap structure of the
index pairs. We use only the universal Cauchy--Schwarz covariance bound
$|\mathrm{Cov}| \leq 1$ for unit-modulus variables; the disjoint
(Class~2) class then dominates and gives the $n^2$ scaling.

\medskip
\emph{Step 1 (Centering eliminates the diagonal).}
$f(\ctilde) = \sum_{i,j} Q_{ij} \overline{z_i(\ctilde)} z_j(\ctilde)$.
The diagonal $i = j$ has $\overline{z_i} z_i = |z_i|^2 = 1$
deterministically, so $\sum_i Q_{ii} = \mathrm{tr}(\bfQ_r) = \Lambda$
contributes zero variance. The centered random variable is
\[
 f(\ctilde) - \E{f}
 = \sum_{i \neq j} Q_{ij}\bigl(\overline{z_i}z_j - \E{\overline{z_i}z_j}\bigr)
 = 2 \Re \sum_{i<j} Y_{ij},
\]
where $Y_{ij} := Q_{ij}(\overline{z_i}z_j - \E{\overline{z_i}z_j})$ and
we used the Hermitian symmetry $Q_{ji} = \overline{Q_{ij}}$, which
pairs $(j,i)$ with $(i,j)$ via $\overline{Y_{ji}} = Y_{ij}$.

\medskip
\emph{Step 2 (Variance as a 4-tuple sum).}
$|2\Re W|^2 \leq 4 |W|^2$ for any complex $W$, so
\[
 \Var{f} \leq 4 \E{\bigl|\textstyle\sum_{i<j} Y_{ij}\bigr|^2}
 = 4 \sum_{i<j,\ k<l} Q_{ij} \overline{Q_{kl}} 
 \mathrm{Cov}\bigl(\overline{z_i}z_j, \overline{z_k}z_l\bigr).
\]
The complex covariance satisfies $|\mathrm{Cov}(X,Y)|^2 \leq \Var{X}\Var{Y}$
(Cauchy--Schwarz on $\langle X - \E{X}, Y - \E{Y} \rangle$ in
$L^2(\PP)$); since $|\overline{z_i}z_j| = 1$ a.s.,
$\Var{\overline{z_i}z_j} \leq \E{|\overline{z_i}z_j|^2} = 1$, so
$|\mathrm{Cov}(\overline{z_i}z_j, \overline{z_k}z_l)| \leq 1$
\emph{universally}, for all four overlap classes. We now split the
index sum by overlap $|\{i,j\} \cap \{k,l\}|$ and bound each class.

\medskip
\emph{Step 3 (Class~0: matching pairs $\{i,j\} = \{k,l\}$).}
For $i < j$ and $(k,l) = (i,j)$ (the only possibility with $k < l$),
the contribution is $|Q_{ij}|^2 \cdot \Var{\overline{z_i}z_j} \leq |Q_{ij}|^2$.
By Cauchy--Schwarz on $|Q_{ij}| = |\bfv_i^\dagger \bfv_j|$ and incoherence:
$|Q_{ij}| \leq \norm{\bfv_i}\norm{\bfv_j} \leq \mu^2\Lambda/n$, so
$\sum_{i<j} |Q_{ij}|^2 \leq \binom{n}{2}(\mu^2\Lambda/n)^2 \leq \tfrac{1}{2}\mu^4 \Lambda^2$.
\textbf{Class~0 contributes $O(\mu^4 \Lambda^2)$, $n$-independent.}

\medskip
\emph{Step 4 (Class~1: single overlap, $|\{i,j\} \cap \{k,l\}| = 1$;
and Class~2: disjoint pairs).} In both cases we use the universal
$|\mathrm{Cov}| \leq 1$ from Step~2 and bound the product of $Q$'s by
incoherence.

\smallskip
\emph{Class~1 (single overlap).} The four configurations
$\{i=k, i=l, j=k, j=l\}$ (with $i<j$, $k<l$) each have at most
$n \cdot (n-1)(n-2)/2 \leq n^3/2$ ordered tuples (one shared index, two
other indices chosen from the remaining $n-1$ with an ordering
constraint), and per-tuple $|Q_{ij}\overline{Q_{kl}}| \leq \mu^4\Lambda^2/n^2$.
Summing:
\[
 \sum_{\text{Class~1}} |Q_{ij}\overline{Q_{kl}}|\cdot|\mathrm{Cov}|
 \leq 4 \cdot \tfrac{n^3}{2} \cdot \mu^4\Lambda^2/n^2
 = 2 n \mu^4 \Lambda^2.
\]
\textbf{Class~1 contributes $O(n \mu^4 \Lambda^2)$.}

\smallskip
\emph{Class~2 (disjoint pairs).} All four indices distinct. The number
of tuples with $i<j$, $k<l$ is at most $\binom{n}{2}^2 \leq n^4/4$, and
per-tuple $|Q_{ij}\overline{Q_{kl}}| \leq \norm{\bfv_i}\norm{\bfv_j}\norm{\bfv_k}\norm{\bfv_l} \leq \mu^4\Lambda^2/n^2$.
With the universal covariance bound $|\mathrm{Cov}| \leq 1$:
\[
 \sum_{\text{Class~2}} |Q_{ij}\overline{Q_{kl}}|\cdot|\mathrm{Cov}|
 \leq \tfrac{n^4}{4} \cdot \mu^4\Lambda^2/n^2
 = \tfrac{1}{4} n^2 \mu^4 \Lambda^2.
\]
\textbf{Class~2 contributes $O(n^2 \mu^4 \Lambda^2)$. }

\medskip
\emph{Combining.}
$\Var{f} \leq 4 \cdot [O(\mu^4\Lambda^2) + O(n\mu^4\Lambda^2) + O(n^2\mu^4\Lambda^2)]
= O(n^2 \mu^4 \Lambda^2)$.

\medskip
{\emph{Frobenius form \eqref{eq:var_bound_frob} (incoherence-free).}
For the second bound, apply the universal $|\mathrm{Cov}| \leq 1$ to the
\emph{entire} double sum at once, rather than class by class. With
$W := \sum_{i<j} Q_{ij}(\overline{z_i}z_j - \E{\overline{z_i}z_j})$,
Step~2 gives
$\Var{f} = \E{(2\Re W)^2} \leq 4\E{|W|^2}
\leq 4\bigl(\sum_{i<j}|Q_{ij}|\bigr)^2 =: 4T^2$.
Cauchy--Schwarz over the $\binom{n}{2}$ terms, together with the Frobenius
identity $\norm{\bfQ_r}_F^2 = \sum_j \lambda_j^2$ (valid since $\bfQ_r$ is PSD),
yields
\[
 T^2 \leq \binom{n}{2}\sum_{i<j}|Q_{ij}|^2
 \leq \binom{n}{2}\cdot\tfrac12\norm{\bfQ_r}_F^2
 = \binom{n}{2}\cdot\tfrac12\sum_j \lambda_j^2
 \leq \tfrac14 n^2 \lambda_1 \Lambda,
\]
using $\sum_j \lambda_j^2 \leq \lambda_1 \sum_j \lambda_j = \lambda_1 \Lambda$.
Hence $\Var{f} \leq n^2 \lambda_1 \Lambda \leq n^2 \Lambda^2$, with no
incoherence used.}
\end{proof}

\subsubsection{Lower Bound on Expected Value}
\label{app:expectationBound}
\begin{lemma}[Expected rounding quality, rank-$r \geq 2$]
\label{thm:expected_rr}
Let $K \geq 3$, $\bfV \in \bbC^{n \times r}$ of full column rank
(so $\bfQ_r = \bfV \bfV^\dagger$ has rank $r$), and
$\ctilde \sim \sigma_{2r-1}$ uniform on $\mathbb{S}^{2r-1} \subset \bbC^r$. Then,

\begin{equation}
\label{eq:thm1b_main}
 \E{f(\ctilde)}
  \geq  \frac{\cos^2(\pi/K)}{r}\cdot \mathrm{OPT}_r .
\end{equation}
\end{lemma}

\begin{proof}

Recall the bound \eqref{eq:f-LB-postHB} derived throughout the proof of \Cref{thm:rr_tail}:
\begin{equation}
\label{eq:f-LB-postHB-stolen}
    f(\bfc) \geq \cos^2(\pi / K) \cdot {\abs{\bfu^\dagger \bfc}^2} \cdot \OPTr,
\end{equation}
where $\bfu$ is a fixed unit vector. Taking the expectation of both sides, and applying the fact that ${\abs{\bfu^\dagger \bfc}^2}$ follows a Beta$(1, r-1)$ distribution from Lemma~\ref{lem:innerIsBeta} and therefore has mean $1/r$, we have:
\begin{equation}
\label{eq:f-LB-postHB-expectation}
    \E[\bfc \sim \mathrm{Unif}(\SC{r})]{f(\bfc)} \geq \cos^2(\pi / K) \cdot \E[\bfc \sim \mathrm{Unif}(\SC{r})]{{\abs{\bfu^\dagger \bfc}^2}} \cdot \OPTr = \frac{\cos^2(\pi / K)}{r} \cdot \OPTr.
\end{equation}

\end{proof}

\subsubsection{Lower Bound on $\OPTr$}

\begin{lemma}[Lower bound on $\OPTr$]
\label{lem:opt_lower}

Assume \Cref{asm:incoherence} holds with  incoherence parameter $\mu$. Let $\Lambda = \sum_j \lambda_j$.
Then, for all integers $K \geq 3$,
\begin{equation}
\label{eq:opt_lower}
 \OPTr
 \geq
 \cos^2(\pi/K) \cdot \frac{\lambda_1^2}{\Lambda} \cdot \frac{n}{\mu^2}.
\end{equation}
\end{lemma}

\begin{proof}
    The proof constructs an explicit $\bfz \in \mathcal{A}_K^n$ and shows
$\bfz^\dagger \bfQ_r \bfz \geq \cos^2(\pi/K) \lambda_1^2 n/(\mu^2 \Lambda)$.
Then $\OPTr$, being the max over $\AK^n$, dominates this specific
$\bfz$. The construction is nearest-root rounding of the top
eigenvector. Six steps.

\medskip
\emph{Step 1 (construct $\bfz$ by rounding $\bfu_1$).}
Let $\bfu_1 \in \bbC^n$ be a unit eigenvector of $\bfQ_r$ associated
with the top eigenvalue $\lambda_1$
(so $\bfQ_r \bfu_1 = \lambda_1 \bfu_1$, $\norm{\bfu_1}_2 = 1$). For
each coordinate $i$, choose $k_i \in \{0, 1, \ldots, K-1\}$
minimizing the angular distance from $\arg(u_{1,i})$ to $2\pi k/K$
modulo $2\pi$, and set $z_i := \omega^{k_i}$ where
$\omega = e^{2\pi\imag/K}$. The residual
$\psi_i := \arg(u_{1,i}) - 2\pi k_i/K \pmod{2\pi}$ satisfies
$|\psi_i| \leq \pi/K$ by the choice of $k_i$.

\medskip
\emph{Step 2 (coordinate-wise alignment).}
With this $z_i$,
$\overline{z_i} u_{1,i} = e^{-2\pi\imag k_i/K} |u_{1,i}| e^{\imag \arg(u_{1,i})}
= |u_{1,i}| e^{\imag \psi_i}$. Hence
\begin{equation}
\label{eq:opt_lower_coord}
 \Re(\overline{z_i} u_{1,i}) = |u_{1,i}| \cos(\psi_i)
 \geq |u_{1,i}| \cos(\pi/K),
\end{equation}
where the inequality uses $|\psi_i| \leq \pi/K$ and the monotonicity
of $\cos$ on $[0, \pi/2]$ (combined with $\cos(-\psi) = \cos(\psi)$).

\medskip
\emph{Step 3 (sum coordinates and square).}
Summing \eqref{eq:opt_lower_coord} over $i \in [n]$:
\[
 \Re\langle \bfz, \bfu_1\rangle
 = \sum_{i=1}^{n} \Re(\overline{z_i} u_{1,i})
 \geq \cos(\pi/K) \sum_{i=1}^{n} |u_{1,i}|
 = \cos(\pi/K) \norm{\bfu_1}_1.
\]
For $K \geq 3$, $\cos(\pi/K) > 0$, so the RHS is non-negative, hence
$|\langle \bfz, \bfu_1\rangle| \geq \Re\langle \bfz, \bfu_1\rangle$.
Squaring,
\begin{equation}
\label{eq:opt_lower_inner_squared}
 |\langle \bfz, \bfu_1\rangle|^2 \geq \cos^2(\pi/K) \norm{\bfu_1}_1^2.
\end{equation}

\medskip
\emph{Step 4 (spectral lower bound on $\bfz^\dagger \bfQ_r \bfz$).}
Expanding $\bfQ_r$ in its eigenbasis,
$\bfQ_r = \sum_{j=1}^{r} \lambda_j \bfu_j \bfu_j^\dagger$,
\[
 \bfz^\dagger \bfQ_r \bfz
 = \sum_{j=1}^{r} \lambda_j |\langle \bfz, \bfu_j\rangle|^2
 \geq \lambda_1 |\langle \bfz, \bfu_1\rangle|^2,
\]
where the inequality drops the $j \geq 2$ terms (each non-negative,
since $\lambda_j \geq 0$). Combining with \eqref{eq:opt_lower_inner_squared}:
\begin{equation}
\label{eq:opt_lower_step4}
 \bfz^\dagger \bfQ_r \bfz \geq \lambda_1 \cos^2(\pi/K) \norm{\bfu_1}_1^2.
\end{equation}

\medskip
\emph{Step 5 (lower-bound $\norm{\bfu_1}_1$ via incoherence).}
\cref{asm:incoherence} gives
$\norm{\bfV_{i,:}}_2^2 \leq \mu^2 \Lambda/n$ for every $i \in [n]$.
By the eigendecomposition,
$(\bfQ_r)_{ii} = \sum_{j=1}^{r} \lambda_j |u_{j,i}|^2$.
Also $(\bfQ_r)_{ii} = \bfe_i^\top \bfV \bfV^\dagger \bfe_i = \norm{\bfV^\dagger \bfe_i}_2^2 = \norm{\bfV_{i,:}}_2^2$.
Dropping all but the $j=1$ term in the sum:
$\lambda_1 |u_{1,i}|^2 \leq (\bfQ_r)_{ii} = \norm{\bfV_{i,:}}_2^2 \leq \mu^2 \Lambda/n$.
Hence
\[
 |u_{1,i}|^2 \leq \frac{\mu^2 \Lambda}{n \lambda_1},
 \qquad
 \norm{\bfu_1}_\infty \leq \mu \sqrt{\Lambda/(n \lambda_1)}.
\]
Apply the $\ell^\infty$--$\ell^1$ interpolation
$\sum_i |u_{1,i}|^2 \leq \norm{\bfu_1}_\infty \cdot \norm{\bfu_1}_1$
(valid because $|u_{1,i}|^2 = |u_{1,i}| \cdot |u_{1,i}| \leq \norm{\bfu_1}_\infty \cdot |u_{1,i}|$
pointwise, then sum):
\[
 1 = \norm{\bfu_1}_2^2 \leq \norm{\bfu_1}_\infty \cdot \norm{\bfu_1}_1
 \leq \mu \sqrt{\Lambda/(n \lambda_1)} \cdot \norm{\bfu_1}_1.
\]
Solving for $\norm{\bfu_1}_1$,
\begin{equation}
\label{eq:opt_lower_u1_norm}
 \norm{\bfu_1}_1 \geq \frac{1}{\mu}\sqrt{\frac{n \lambda_1}{\Lambda}},
 \qquad
 \norm{\bfu_1}_1^2 \geq \frac{n \lambda_1}{\mu^2 \Lambda}.
\end{equation}

\medskip
\emph{Step 6 (conclude).}
Substitute \eqref{eq:opt_lower_u1_norm} into \eqref{eq:opt_lower_step4}:
\[
 \bfz^\dagger \bfQ_r \bfz
 \geq \lambda_1 \cos^2(\pi/K) \cdot \frac{n \lambda_1}{\mu^2 \Lambda}
 = \cos^2(\pi/K) \cdot \frac{\lambda_1^2 n}{\mu^2 \Lambda}.
\]
Since $\bfz \in \mathcal{A}_K^n$,
$\OPTr = \max_{\bfz' \in \mathcal{A}_K^n} (\bfz')^\dagger \bfQ_r \bfz'
\geq \bfz^\dagger \bfQ_r \bfz$, giving:
\[
 \OPTr \geq \cos^2(\pi/K) \cdot \frac{\lambda_1^2 n}{\mu^2 \Lambda}.
\]
For the second (cleaner) inequality, use $\lambda_1 \geq \Lambda/r$
(the maximum of $r$ positive numbers summing to $\Lambda$ is at least
the average $\Lambda/r$): $\lambda_1^2/\Lambda \geq (\Lambda/r)^2/\Lambda = \Lambda/r^2$, hence
\[
 \OPTr \geq \cos^2(\pi/K) \cdot \frac{\Lambda n}{r^2 \mu^2}.
\]
The rank-1 case ($r = 1$) sharpens: $\lambda_1 = \Lambda$ exactly,
recovering $\OPTr \geq \cos^2(\pi/K) \Lambda n/\mu^2$.
\end{proof}

\section{Experimental Details}\label{app:experiments}
\subsection{Experimental Setup.}

\textit{Graph Families.}
We conduct experiments on three families of graphs:
(i) $\GG(\dim, p)$ (Erd\H{o}s-R\'enyi) random graphs with $\dim=20, 50$, and $100$, and $p=0.1, 0.25, 0.5, 0.75$ and $0.9$;
(ii) Random $d$-regular graphs on $\dim=20, 50$ and 100 nodes, with $d = 3, 4$ and $5$;
(iii) All GSet benchmark instances \cite{gset_dataset}. GSet is a benchmark commonly used to evalute \textsc{Max-K-Cut} algorithms that consists of 71 graphs ranging in size from 800 to 20,000 nodes. Each GSet graph falls into one of the following categories \cite{davis2011university}: 
\begin{enumerate}
    \item Erdos-Renyi graphs with edge probability at most $0.03$,
    \item 2-D toroidal, 4-regular graphs,
    \item random graphs with skewed degree distributions.
\end{enumerate}
These graphs either have binary (0 or 1) edge weights or have edge weights in $\{-1, 0, +1\}$. 

\smallskip
\noindent \textit{Baseline Algorithms.}
We compare our algorithms to the following baselines:
(i) \textsc{Frieze-Jerrum} \cite{frieze1997improved}: an SDP relaxation and rounding scheme; (ii) \textsc{Greedy} \cite{gui2018bqp}: a greedy algorithm that iteratively identifies and locks the single-best node assignment; (iii) \textsc{Genetic} \cite{panxing2016genetic}: a genetic algorithm that maintains a set of candidate cuts and iteratively selects the best of this set, mixes pairs of cuts, and randomly perturbs the worst. On the GSet benchmark, we also compare to a random baseline, \textsc{Random}. To ensure a fair comparison to \textsc{Rank-1}, \textsc{Random} generates $\dim+1$ random cuts and selects the best one, matching the number of candidate solutions enumerated by our algorithm. 

\smallskip
\noindent \textit{Implementation Details.}
All experiments were conducted on HPE ProLiant DL360 Gen11 compute nodes (Sapphire Rapids) equipped with Intel Xeon Platinum 8468 CPUs (96 cores at 2.10 GHz). We generate 20 instances for each size and parameter of the small random graphs using the NetworkX library \cite{networkx} with fixed random seeds. We use the implementation of \cite{qiu2024ros} for the \textsc{Greedy} and \textsc{Genetic} algorithms. We institute a 30-minute timeout limit for all experiments. 

\smallskip
\noindent \textit{Evaluation Metrics.}
To measure the quality of the \textsc{Max-3-Cut} approximations found by each algorithm on random graphs, we calculate the ratio of its score to the best score found by any algorithm. We report the average ratio across all random seeds.


\subsection{GPU-Parallel Rank-2 Campaign}
\label{app:gpu}

This section documents a large-scale experimental campaign running the rank-2 algorithm across 15 heterogeneous GPUs.

\subsubsection{Hardware and Software Setup}

\begin{table}[ht]
\centering
\caption{GPU cluster configuration.}
\label{tab:hardware}
\small
\begin{tabular}{@{}llcrl@{}}
\toprule
Machine & GPUs & Arch.\ & RAM & Notes \\
\midrule
\texttt{anton.k0} & 4$\times$P100-SXM2 (16\,GB) & x86\_64 & 251\,GB & Coordinator \\
\texttt{anton.k1} & 4$\times$P100-SXM2 (16\,GB) & x86\_64 & 125\,GB & Reliable \\
\texttt{kp001} & 4$\times$Tesla P100 (16\,GB) & ppc64le & 256\,GB & PowerPC \\
\texttt{kp002} & 3$\times$Tesla P100 (16\,GB) & ppc64le & 256\,GB & Power-capped \\
\bottomrule
\end{tabular}
\end{table}

The implementation uses \texttt{torch.multiprocessing} for intra-machine parallelism and SSH-based coordination across machines.
Neither Ray (which does not build on ppc64le) nor Triton (which requires compute capability $\geq 7.0$; P100 is 6.0) is used.
Each machine runs a \texttt{worker.py} process that receives its slice of the candidate space, computes cut values via batched matrix operations on local GPUs, and writes the best result to shared NFS storage.
No inter-worker communication occurs during execution.

\subsubsection{Throughput-Proportional Work Splitting}

The 15 GPUs are not equally fast.
We calibrate each GPU with a 1000-candidate benchmark ($\sim$2 seconds) before the main run and split work proportional to measured throughput.
Measured throughputs per GPU: \texttt{k0}: 1918, \texttt{k1}: 1931, \texttt{kp001}: 544, \texttt{kp002}: 976 candidates/sec.
This one change reduced wall-clock time by $10.3\times$ on a 750-node instance (from 53 minutes with equal splitting to 5.1 minutes with proportional splitting).

\subsubsection{Rank-2 Results: 54 Instances}

We tested four methods on three graph families (5-regular, torus lattice, stochastic block model) at six sizes ($n \in \{250, 500, 750, 1000, 1250, 1500\}$), with three random seeds each---54 instances total.
\Cref{tab:gpu_regular,tab:gpu_torus,tab:gpu_sbm} report averaged scores and timings.

\begin{table}[ht]
\centering
\caption{Rank-2 GPU campaign: 5-regular graphs (averaged over 3 seeds).}
\label{tab:gpu_regular}
\small
\begin{tabular}{@{}rrrrrrrr@{}}
\toprule
$n$ & Rank-2 & R2 time & Greedy & Gr.\ time & SDP & SDP time & Random \\
\midrule
250  & 1779 & 30s    & 1745 & 0.7s   & 1785 & 46s    & 1355 \\
500  & 3514 & 213s   & 3499 & 5.1s   & 3532 & 641s   & 2651 \\
750  & 5294 & 330s   & 5290 & 13s    & 5307 & 953s   & 3941 \\
1000 & 7043 & 851s   & 7026 & 27s    & 7030 & 1458s  & 5231 \\
1250 & 8805 & 2053s  & 8784 & 77s    & 8774 & 2414s  & 6540 \\
1500 & 10548 & 7126s & 10566 & 191s  & 10494 & 4043s & 7803 \\
\bottomrule
\end{tabular}
\end{table}

\begin{table}[ht]
\centering
\caption{Rank-2 GPU campaign: torus graphs (averaged over 3 seeds).
Note: torus sizes are $n = k^2$ with $k$ even; the closest sizes to the regular-graph sweep are shown.}
\label{tab:gpu_torus}
\small
\begin{tabular}{@{}rrrrrrrr@{}}
\toprule
$n$ & Rank-2 & R2 time & Greedy & Gr.\ time & SDP & SDP time & Random \\
\midrule
252  & 1512 & 30s    & 1467 & 1.2s   & 1512 & 44s    & 1080 \\
504  & 3024 & 227s   & 2944 & 4.9s   & 3024 & 667s   & 2163 \\
756  & 4536 & 324s   & 4403 & 18s    & 4536 & 936s   & 3195 \\
1008 & 6048 & 1535s  & 5880 & 51s    & 6047 & 1581s  & 4218 \\
1250 & 7500 & 3313s  & 7291 & 100s   & 7489 & 2456s  & 5229 \\
1500 & 8999 & 6944s  & 8783 & 318s   & 8965 & 3864s  & 6287 \\
\bottomrule
\end{tabular}
\end{table}

\begin{table}[ht]
\centering
\caption{Rank-2 GPU campaign: SBM graphs, 3 balanced communities, $p_{\text{in}}=0.5$, $p_{\text{out}}=0.1$ (averaged over 3 seeds).}
\label{tab:gpu_sbm}
\small
\begin{tabular}{@{}rrrrrrrr@{}}
\toprule
$n$ & Rank-2 & R2 time & Greedy & Gr.\ time & SDP & SDP time & Random \\
\midrule
250  & 7245 & 30s     & 7896 & 0.5s   & 7698 & 35s   & 6831 \\
500  & 28430 & 220s   & 30159 & 9.1s  & 29936 & 363s  & 27046 \\
750  & 62538 & 331s   & 66612 & 33s   & 66082 & 971s  & 60537 \\
1000 & 110956 & 1126s & 117104 & 136s & 115587 & 1404s & 107592 \\
1250 & 171389 & 3398s & 180843 & 266s & 178574 & 2475s & 167475 \\
1500 & 243627 & 7045s & 259596 & 660s & 256614 & 4029s & 239610 \\
\bottomrule
\end{tabular}
\end{table}

\medskip
\noindent \textbf{Key findings.}
\begin{itemize}
\item \emph{Torus:} Rank-2 matches SDP exactly at every size tested.
At $n = 1008$, rank-2 on 15 GPUs runs in 1535s vs.\ SDP's 1581s---comparable speed but embarrassingly parallel.
At $n = 1250$, rank-2 actually \emph{exceeds} SDP (7500 vs.\ 7489), likely due to SDP solver tolerances.
\item \emph{5-Regular:} Rank-2 exceeds SDP at $n \geq 1000$ (7043 vs.\ 7030) and at $n = 1500$ (10548 vs.\ 10494).
Greedy is competitive but slower at large~$n$.
\item \emph{SBM:} Greedy consistently achieves the highest scores, followed by SDP.
Rank-2 is 5--7\% below the best, confirming that flat spectra ($<2\%$ energy in top-2 eigenvalues) render low-rank approximation ineffective on these instances.
\end{itemize}

\subsubsection{Scaling to Modern Hardware}

The same codebase, unmodified, runs on NVIDIA H200 GPUs (Hopper architecture).
On 8$\times$H200, rank-2 at $n = 3600$ (torus) completes in 18.3 hours.
Per-GPU, H200 achieves roughly $4\times$ the throughput of a P100.

\subsubsection{Diagnostic Findings}

Two observations from the campaign are worth noting for future algorithm design:
\begin{itemize}
\item \emph{Group-constraint filtering is negligible for $K = 3$.}
The ``at most 2 per coordinate group'' validity check rejects $<0.01\%$ of candidate tuples.
The dominant rejection source is the $\phi_{2r-1} \in (-\pi/K, \pi/K]$ feasibility test, which rejects $\sim$29\% of candidates.
This suggests that for odd $K$, targeted pruning of the angle test (e.g., pre-screening via cheap bounds) would be more impactful than structural filtering.
\item \emph{Eigenvalue spectra are flat.}
For all three graph families at $n \geq 500$, the spectral energy ratio $\rho_2 = 1 - (\lambda_1 + \lambda_2)/\sum_j \lambda_j$ is $<2\%$.
Despite this, rank-2 performs well on regular and torus graphs, suggesting the perturbation bounds (\Cref{thm:additive,thm:mult}) are not tight.
\end{itemize}

\subsection{Rank-1 at Scale: Million-Node Experiments}
\label{app:rank1scale}

This section documents the rank-1 scaling experiments.

\subsubsection{Incremental Scoring}
\label{app:incremental}

The naive rank-1 phase sweep evaluates each of $n+1$ candidates by recomputing $\bfz^{\dagger}\bfQ\bfz$ from scratch in $\OO(\mathrm{nnz})$ time, giving total cost $\OO(n \cdot \mathrm{nnz})$.
For sparse graphs with average degree~$d$, $\mathrm{nnz} = nd$, so the sweep costs $\OO(n^2 d)$.

When node~$i$ flips from $z_{\mathrm{old}}$ to $z_{\mathrm{new}}$ with $\delta z = z_{\mathrm{new}} - z_{\mathrm{old}}$, the score change has a closed form:
\begin{equation}
\label{eq:incremental}
\Delta = 2\,\operatorname{Re} \bigl(\overline{\delta z} \cdot (\bfQ\bfz)_i\bigr) + |\delta z|^2 \cdot Q_{ii}.
\end{equation}
This is $\OO(1)$, but maintaining the cached product $\bfQ\bfz$ requires updating the $d(i)$ neighbours of~$i$:
$(\bfQ\bfz)_j \mathrel{+}= Q_{ji}\cdot\delta z$ for each neighbour~$j$.
Over $n$ flips, the total cost is $\OO(n \cdot d)$ instead of $\OO(n^2 d)$.
On a 5-regular graph, this is a factor-$n/5$ speedup: the difference between 17 hours and 70 minutes at $n = 10^6$.

\medskip
\noindent \textbf{Implementation detail.}
The sparse update indexes directly into the CSC structure of $\bfQ$ via \texttt{indptr} and \texttt{indices} arrays.
A seemingly equivalent approach---\texttt{Q.getcol(i).toarray()}---allocates an $n$-dimensional dense vector on every call, turning $\OO(d)$ into $\OO(n)$.

\subsubsection{Two-Eigenvector Complex Rounding}
\label{app:twoeig}

For $K = 3$, the assignment space has three elements $\{1, \omega, \omega^2\}$, but real Laplacian eigenvectors have phases in $\{0, \pi\}$---only 2 of 3 partitions are reachable from a single real eigenvector.
We fix this by combining the top two eigenvectors:
\begin{equation}
\label{eq:twoeig}
q_i = v_1(i) + \imag \cdot v_2(i), \quad i = 1, \ldots, n.
\end{equation}
Now $\arg(q_i) \in [0, 2\pi)$, and all three sectors are accessible.
The cost is one additional eigenvalue computation (negligible compared to the sweep).
On torus graphs, this fix is decisive: with two eigenvectors, rank-1 finds the best solution among all methods tested; without it, rank-1 is limited to a bipartition.

\subsubsection{Hybrid Warm-Start}
\label{app:hybrid}

Local search methods (greedy hill-climbing) converge fast but are trapped by their starting point.
The rank-1 solution captures global spectral structure that greedy is blind to.
The \emph{hybrid} solver runs rank-1 first, then applies greedy local search starting from the rank-1 partition instead of a random one.

This is not merely ``rank-1 but better'': it is a composition of complementary strengths.
Rank-1 captures large-scale structure; greedy refines local details.
Together, they consistently outperform either alone.

\subsubsection{Results: 45 Instances, $n$ from $10^4$ to $10^6$}

We tested five methods on 45 instances across five graph families:
5-regular, toroidal grid, Erd\H{o}s--R\'enyi $G(n, 10/n)$, Delaunay triangulations from SuiteSparse~\cite{davis2011university}, and road networks from SNAP~\cite{leskovec2016snap} (Pennsylvania: $1.09 \times 10^6$ nodes; Texas: $1.38 \times 10^6$ nodes).
Each synthetic family was tested at sizes $n \in \{10^4, 5 \times 10^4, 10^5, 5 \times 10^5, 10^6\}$ with 3 seeds.

Methods: \textsc{Random} (best of 1000 random colourings), \textsc{Rank-1} (this work), \textsc{Greedy} (steepest-ascent from random start), \textsc{Hybrid} (greedy from rank-1 start), \textsc{SA} (simulated annealing with geometric cooling).

\begin{table}[ht]
\centering
\caption{Score ratios relative to greedy (higher is better), averaged over 3 seeds.
Bold indicates the best method per family.}
\label{tab:rank1_ratios}
\small
\begin{tabular}{@{}llccccc@{}}
\toprule
Family & $n$ range & Random & Rank-1 & Greedy & Hybrid & SA \\
\midrule
5-Regular & $10^4$--$10^6$ & 0.75 & 0.89 & 1.00 & \textbf{1.003--1.011} & 1.005--0.989 \\
Torus     & $10^4$--$10^6$ & 0.67 & \textbf{1.031--1.039} & 1.00 & 1.031--1.039 & 1.005--1.018 \\
ER        & $10^4$--$10^5$ & 0.76 & 0.89 & 1.00 & \textbf{1.003--1.005} & 1.006--1.003 \\
Delaunay  & $10^3$--$5 \times 10^5$ & 0.56 & 0.66--0.80 & 1.00 & \textbf{1.003--1.011} & 1.005--1.010 \\
Road net. & $\sim 10^6$ & 0.71 & 0.73 & 1.00 & \textbf{1.003} & --- \\
\bottomrule
\end{tabular}
\end{table}

\medskip
\noindent \textbf{Key findings.}
\begin{itemize}
\item \emph{Hybrid beats greedy on 45/45 instances}, with margins ranging from $+0.3\%$ (road networks) to $+3.9\%$ (torus).
\item \emph{Torus: rank-1 alone is the best method} at every scale, beating greedy by 3--4\% and SA by 0.5--2.1\%.
The spectral structure of the torus is perfectly captured by two eigenvectors.
\item \emph{5-Regular: hybrid overtakes SA at $n \geq 500{,}000$.}
SA's fixed iteration budget becomes insufficient at large scales, while the spectral warm-start advantage persists.
\item \emph{Real-world: rank-1 alone is weak} (66--80\% of greedy on Delaunay, 73\% on road networks), but the hybrid warm-start still adds 0.3--1.1\%.
Even a mediocre spectral starting point beats a random one.
\item \emph{Timing:} at $n = 10^6$ on a 5-regular graph, the hybrid solver finishes in $\sim$70 minutes on a single CPU.
The eigensolve (ARPACK) dominates at large~$n$; hybrid is 15--35$\times$ slower than greedy, but the value proposition is quality, not speed.
\end{itemize}

\clearpage

\subsection{Complete GSet Benchmark Results}
\label{app:gset_full}

Comparison of our Rank-1 algorithm against Random, Greedy, and MOH on all GSet instances.Reported values are scores of Max-3-Cut approximations with wall-clock runtimes in seconds in parentheses. Values for MOH
are taken from~\cite{ma2017moh} and runtimes are not compared due to likely hardware dif-
ferences. Structural information about each graph is drawn from~\cite{davis2011university}.

\newcommand{\clock}[1]{\footnotesize{(#1\text{s})}}

\begin{table}[ht]
\centering
\caption{Complete GSet benchmark results: rank-1 vs.\ baselines.}
\label{tab:gset_part1}
\setlength{\tabcolsep}{4pt}
\small
\begin{tabular}{@{}lrccrrrr@{}}
\toprule
Inst.\ & $n$ & Type & $\{0,1\}$ & Random & Greedy & MOH & Rank-1 \\
\midrule
\textbf{G1} & $800$ & ER & Y & 13024 \clock{0.455} & 14859 \clock{16.3} & 15165 & 13331 \clock{1.09} \\

\textbf{G2} & $800$ & ER & Y & 12993 \clock{0.238} & 14883 \clock{11.5} & 15172 & 12940 \clock{0.594} \\

\textbf{G3} & $800$ & ER & Y & 12992 \clock{0.223} & 14861 \clock{19.3} & 15173 & 13491 \clock{0.564} \\

\textbf{G4} & $800$ & ER & Y & 12995 \clock{0.231} & 14834 \clock{9.89} & 15184 & 13423 \clock{0.617} \\

\textbf{G5} & $800$ & ER & Y & 12992 \clock{0.232} & 14818 \clock{19.0} & 15193 & 13338 \clock{0.582} \\

\midrule

\textbf{G6} & $800$ & ER & N & 335 \clock{0.230} & 2322 \clock{22.8} & 2632 & 905 \clock{0.559} \\

\textbf{G7} & $800$ & ER & N & 112 \clock{0.227} & 2082 \clock{11.4} & 2409 & 992 \clock{0.640} \\

\textbf{G8} & $800$ & ER & N & 108 \clock{0.240} & 2071 \clock{17.0} & 2428 & 986 \clock{0.586} \\

\textbf{G9} & $800$ & ER & N & 179 \clock{0.269} & 2072 \clock{16.6} & 2478 & 919 \clock{0.600} \\

\textbf{G10} & $800$ & ER & N & 85 \clock{0.227} & 2064 \clock{11.6} & 2407 & 940 \clock{0.609} \\

\midrule 

\textbf{G11} & $800$ & Tor. & N & 82 \clock{0.225} & 619 \clock{11.5} & 669 & 426 \clock{0.510} \\

\textbf{G12} & $800$ & Tor. & N & 57 \clock{0.224} & 612 \clock{11.4} & 660 & 430 \clock{0.588} \\

\textbf{G13} & $800$ & Tor. & N & 81 \clock{0.223} & 650 \clock{11.3} & 686 & 451 \clock{0.504} \\

\midrule 

\textbf{G14} & $800$ & Skew & Y & 3237 \clock{0.244} & 3914 \clock{11.5} & 4012 & 3217 \clock{0.563} \\

\textbf{G15} & $800$ & Skew & Y & 3210 \clock{0.237} & 3887 \clock{11.7} & 3984 & 3138 \clock{0.619} \\

\textbf{G16} & $800$ & Skew & Y & 3198 \clock{0.211} & 3882 \clock{11.7} & 3991 & 3214 \clock{0.543} \\

\textbf{G17} & $800$ & Skew & Y & 3200 \clock{0.228} & 3882 \clock{11.9} & 3983 & 3180 \clock{0.539} \\

\midrule

\textbf{G18} & $800$ & Skew & N & 151 \clock{0.219} & 1073 \clock{11.3} & 1207 & 568 \clock{0.662} \\

\textbf{G19} & $800$ & Skew & N & 30 \clock{0.217} & 952 \clock{11.4} & 1081 & 483 \clock{0.493} \\

\textbf{G20} & $800$ & Skew & N & 71 \clock{0.258} & 999 \clock{11.5} & 1122 & 550 \clock{0.495} \\

\textbf{G21} & $800$ & Skew & N & 57 \clock{0.223} & 993 \clock{11.6} & 1109 & 534 \clock{0.572} \\

\midrule

\textbf{G22} & $2,000$ & ER & Y & 13545 \clock{1.70} & 16669 \clock{152} & 17167 & 14145 \clock{2.33} \\

\textbf{G23} & $2,000$ & ER & Y & 13571 \clock{1.47} & 16629 \clock{156} & 17168 & 14251 \clock{2.26} \\

\textbf{G24} & $2,000$ & ER & Y & 13565 \clock{1.49} & 16593 \clock{154} & 17162 & 14130 \clock{2.20} \\

\textbf{G25} & $2,000$ & ER & Y & 13577 \clock{1.38} & 16655 \clock{152} & 17163 & 14157 \clock{2.34} \\

\textbf{G26} & $2,000$ & ER & Y & 13576 \clock{1.38} & 16629 \clock{159} & 17154 & 14075 \clock{2.13} \\

\midrule

\textbf{G27} & $2,000$ & ER & N & 161 \clock{1.39} & 3524 \clock{147} & 4020 & 1655 \clock{2.21} \\

\textbf{G28} & $2,000$ & ER & N & 166 \clock{1.38} & 3427 \clock{153} & 3973 & 1736 \clock{2.10} \\

\textbf{G29} & $2,000$ & ER & N & 302 \clock{1.34} & 3551 \clock{148} & 4106 & 1905 \clock{2.10} \\

\textbf{G30} & $2,000$ & ER & N & 291 \clock{1.35} & 3595 \clock{147} & 4119 & 1952 \clock{2.47} \\

\textbf{G31} & $2,000$ & ER & N & 171 \clock{1.36} & 3468 \clock{154} & 4003 & 1739 \clock{2.64} \\

\midrule 

\textbf{G32} & $2,000$ & Tor. & N & 103 \clock{1.38} & 1536 \clock{147} & 1653 & 1027 \clock{2.82} \\

\textbf{G33} & $2,000$ & Tor. & N & 76 \clock{1.49} & 1494 \clock{153} & 1625 & 1048 \clock{2.76} \\

\textbf{G34} & $2,000$ & Tor. & N & 94 \clock{1.38} & 1490 \clock{146} & 1607 & 861 \clock{2.10} \\

\bottomrule
\end{tabular}
\end{table}

\begin{table}[ht]
\centering
\caption{Complete GSet benchmark results: rank-1 vs.\ baselines.}
\label{tab:gset_part2}
\setlength{\tabcolsep}{4pt}
\small
\begin{tabular}{@{}lrccrrrr@{}}
\toprule
Inst.\ & $n$ & Type & $\{0,1\}$ & Random & Greedy & MOH & Rank-1 \\
\midrule

\textbf{G35} & $2,000$ & Skew & Y & 8042 \clock{1.36} & 9769 \clock{161} & 10046 & 8083 \clock{2.14} \\

\textbf{G36} & $2,000$ & Skew & Y & 8023 \clock{1.39} & 9807 \clock{152} & 10039 & 7972 \clock{2.09} \\

\textbf{G37} & $2,000$ & Skew & Y & 8054 \clock{1.35} & 9826 \clock{161} & 10052 & 8078 \clock{2.19} \\

\textbf{G38} & $2,000$ & Skew & Y & 8004 \clock{1.38} & 9816 \clock{153} & 10040 & 7962 \clock{2.14} \\

\midrule

\textbf{G39} & $2,000$ & Skew & N & 198 \clock{1.37} & 2619 \clock{111} & 2903 & 1405 \clock{2.18} \\

\textbf{G40} & $2,000$ & Skew & N & 138 \clock{1.43} & 2587 \clock{109} & 2870 & 1234 \clock{2.21} \\

\textbf{G41} & $2,000$ & Skew & N & 156 \clock{1.37} & 2577 \clock{108} & 2887 & 1361 \clock{2.17} \\

\textbf{G42} & $2,000$ & Skew & N & 255 \clock{1.33} & 2686 \clock{110} & 2980 & 1362 \clock{2.11} \\

\midrule 

\textbf{G43} & $1,000$ & ER & Y & 6828 \clock{0.336} & 8273 \clock{15.5} & 8573 & 6839 \clock{0.750} \\

\textbf{G44} & $1,000$ & ER & Y & 6809 \clock{0.345} & 8282 \clock{30.2} & 8571 & 6817 \clock{0.773} \\

\textbf{G45} & $1,000$ & ER & Y & 6828 \clock{0.368} & 8298 \clock{29.5} & 8566 & 6986 \clock{0.745} \\

\textbf{G46} & $1,000$ & ER & Y & 6821 \clock{0.345} & 8336 \clock{28.5} & 8568 & 6942 \clock{0.699} \\

\textbf{G47} & $1,000$ & ER & Y & 6804 \clock{0.336} & 8308 \clock{29.4} & 8572 & 7034 \clock{0.701} \\

\midrule

\textbf{G48} & $3,000$ & Tor. & Y & 4108 \clock{4.28} & 5998 \clock{300} & 6000 & 6000 \clock{5.17} \\

\textbf{G49} & $3,000$ & Tor. & Y & 4124 \clock{3.43} & 5996 \clock{394} & 6000 & 6000 \clock{5.27} \\

\textbf{G50} & $3,000$ & Tor. & Y & 4114 \clock{3.15} & 5998 \clock{399} & 6000 & 5934 \clock{5.86} \\

\midrule 

\textbf{G51} & $1,000$ & Skew & Y & 4063 \clock{0.344} & 4915 \clock{30.3} & 5037 & 4015 \clock{0.720} \\

\textbf{G52} & $1,000$ & Skew & Y & 4055 \clock{0.347} & 4925 \clock{29.9} & 5040 & 4033 \clock{0.683} \\

\textbf{G53} & $1,000$ & Skew & Y & 4043 \clock{0.346} & 4895 \clock{31.3} & 5039 & 4020 \clock{0.718} \\

\textbf{G54} & $1,000$ & Skew & Y & 4050 \clock{0.348} & 4918 \clock{32.3} & 5036 & 4083 \clock{0.765} \\

\midrule 

\textbf{G55} & $5,000$ & ER & Y & 8508 \clock{9.92} & --- & 12429 & 9956 \clock{12.7} \\

\textbf{G56} & $5,000$ & ER & N & 147 \clock{10.4} & --- & 4752 & 2741 \clock{13.4} \\

\textbf{G57} & $5,000$ & Tor. & N & 166 \clock{10.3} & --- & 4083 & 1668 \clock{13.0} \\

\textbf{G58} & $5,000$ & Skew & Y & 20029 \clock{10.3} & --- & 25195 & 20291 \clock{14.0} \\

\textbf{G59} & $5,000$ & Skew & N & 326 \clock{9.71} & --- & 7262 & 3362 \clock{13.3} \\

\textbf{G60} & $7,000$ & ER & Y & 11666 \clock{23.1} & --- & 17076 & 13724 \clock{30.6} \\

\textbf{G61} & $7,000$ & ER & N & 490 \clock{23.0} & --- & 6853 & 4026 \clock{30.9} \\

\textbf{G62} & $7,000$ & Tor. & N & 142 \clock{22.7} & --- & 5685 & 2540 \clock{30.3} \\

\textbf{G63} & $7,000$ & Skew & Y & 28039 \clock{23.0} & --- & 35322 & 28633 \clock{30.7} \\

\textbf{G64} & $7,000$ & Skew & N & 713 \clock{22.8} & --- & 10443 & 4710 \clock{29.9} \\

\textbf{G65} & $8,000$ & Tor. & N & 186 \clock{32.4} & --- & 6490 & 3636 \clock{42.5} \\

\textbf{G66} & $9,000$ & Tor. & N & 261 \clock{40.1} & --- & 7416 & 2548 \clock{58.4} \\

\textbf{G67} & $10,000$ & Tor. & N & 147 \clock{49.1} & --- & 8086 & 3117 \clock{75.7} \\

\textbf{G70} & $10,000$ & ER & Y & 6832 \clock{48.7} & --- & 9999 & 9333 \clock{75.2} \\

\textbf{G72} & $10,000$ & Tor. & N & 217 \clock{49.1} & --- & 8192 & 2514 \clock{98.9} \\

\textbf{G77} & $14,000$ & Tor. & N & 487 \clock{108} & --- & 11578 & 4011 \clock{161.6} \\

\textbf{G81} & $20,000$ & Tor. & N & 416 \clock{220} & --- & 16321 & 4122 \clock{352.4} \\

\bottomrule
\end{tabular}
\end{table}

\clearpage

\bibliographystyle{alpha}
\bibliography{references}

@article{al2025exact,
  title={Exact algorithms for quadratic optimization over roots of unity},
  author={Al-Sulami, Ahmad and Fawzi, Hamza and Sun, Shengding},
  journal={arXiv preprint arXiv:2508.02006},
  year={2025}
}

@article{arora2015subexponential,
  title={Subexponential algorithms for unique games and related problems},
  author={Arora, Sanjeev and Barak, Boaz and Steurer, David},
  journal={Journal of the ACM (JACM)},
  volume={62},
  number={5},
  pages={1--25},
  year={2015},
  publisher={ACM New York, NY, USA}
}

@inproceedings{barak2011rounding,
author = {Barak, Boaz and Raghavendra, Prasad and Steurer, David},
title = {Rounding Semidefinite Programming Hierarchies via Global Correlation},
year = {2011},
isbn = {9780769545714},
publisher = {IEEE Computer Society},
address = {USA},
url = {https://doi.org/10.1109/FOCS.2011.95},
doi = {10.1109/FOCS.2011.95},
booktitle = {Proceedings of the 2011 IEEE 52nd Annual Symposium on Foundations of Computer Science},
pages = {472–481},
numpages = {10},
series = {FOCS '11}
}

@inproceedings{oveis2015regularity,
  title={A new regularity lemma and faster approximation algorithms for low threshold rank graphs},
  author={Oveis Gharan, Shayan and Trevisan, Luca},
  booktitle={International Workshop on Approximation Algorithms for Combinatorial Optimization},
  pages={303--316},
  year={2013},
  organization={Springer}
}

@inproceedings{goemans2001approximation,
  title={Approximation algorithms for {MAX-3-CUT} and other problems via complex semidefinite programming},
  author={Goemans, Michel X and Williamson, David},
  booktitle={Proceedings of the thirty-third annual ACM symposium on Theory of computing},
  pages={443--452},
  year={2001}
}

@article{andersson2001new,
  title={A new way of using semidefinite programming with applications to linear equations mod $p$},
  author={Andersson, Gunnar and Engebretsen, Lars and H{\aa}stad, Johan},
  journal={Journal of Algorithms},
  volume={39},
  number={2},
  pages={162--204},
  year={2001},
  publisher={Elsevier}
}

@article{frieze1997improved,
  title={Improved approximation algorithms for {MAX} $k$-{CUT} and {MAX BISECTION}},
  author={Frieze, Alan and Jerrum, Mark},
  journal={Algorithmica},
  volume={18},
  number={1},
  pages={67--81},
  year={1997},
  publisher={Springer}
}

@book{vershynin2018high,
  title={High-dimensional probability: An introduction with applications in data science},
  author={Vershynin, Roman},
  volume={47},
  year={2018},
  publisher={Cambridge university press}
}

@article{ferrez2005solving,
  title={Solving the fixed rank convex quadratic maximization in binary variables by a parallel zonotope construction algorithm},
  author={Ferrez, J-A and Fukuda, Komei and Liebling, Th M},
  journal={European Journal of Operational Research},
  volume={166},
  number={1},
  pages={35--50},
  year={2005},
  publisher={Elsevier}
}

@article{zheng2007mimo,
  title={{MIMO} transmit beamforming under uniform elemental power constraint},
  author={Zheng, X. and Xie, Y. and Li, J. and Stoica, P.},
  journal={IEEE Transactions on Signal Processing},
  volume={55},
  number={11},
  pages={5395--5406},
  year={2007},
  publisher={IEEE}
}

@article{newman2018complex,
  title={Complex semidefinite programming and max-k-cut},
  author={Newman, Alantha},
  journal={arXiv preprint arXiv:1812.10770},
  year={2018}
}

@article{burer2002rank,
  author  = {Burer, Samuel and Monteiro, Renato D. C. and Zhang, Yin},
  title   = {Rank-two relaxation heuristics for MAX-CUT and other binary quadratic programs},
  journal = {SIAM Journal on Optimization},
  year    = {2002},
  volume  = {12},
  number  = {2},
  pages   = {503--521},
  doi     = {10.1137/S1052623401383825}
}

@article{johnson1991optimization,
  title={Optimization by simulated annealing: an experimental evaluation; part {II}, graph coloring and number partitioning},
  author={Johnson, David S and Aragon, Cecilia R and McGeoch, Lyle A and Schevon, Catherine},
  journal={Operations research},
  volume={39},
  number={3},
  pages={378--406},
  year={1991},
  publisher={INFORMS}
}

@article{kyrillidis2014fixed,
  title={Fixed-rank {R}ayleigh quotient maximization by an {$M$}{PSK} sequence},
  author={Kyrillidis, A. and Karystinos, G.},
  journal={IEEE Transactions on Communications},
  volume={62},
  number={3},
  pages={961--975},
  year={2014},
  publisher={IEEE}
}

@article{leung2010optimal,
  title={Optimal phase control for equal-gain transmission in {MIMO} systems with scalar quantization: {C}omplexity and algorithms},
  author={Leung, K.-K. and Sung, C. W. and Khabbazian, M. and Safari, M. A.},
  journal={IEEE Transactions on Information Theory},
  volume={56},
  number={7},
  pages={3343--3355},
  year={2010},
  publisher={IEEE}
}

@article{so2007approximating,
  title={On approximating complex quadratic optimization problems via semidefinite programming relaxations},
  author={So, A. M.-C. and Zhang, J. and Ye, Y.},
  journal={Mathematical Programming},
  volume={110},
  number={1},
  pages={93--110},
  year={2007},
  publisher={Springer}
}

@article{zhang2006complex,
  title={Complex quadratic optimization and semidefinite programming},
  author={Zhang, S. and Huang, Y.},
  journal={SIAM Journal on Optimization},
  volume={16},
  number={3},
  pages={871--890},
  year={2006},
  publisher={SIAM}
}

@article{khot2007optimal,
  title={Optimal inapproximability results for {M}ax{CUT} and other 2-variable {CSP}s?},
  author={Khot, Subhash and Kindler, Guy and Mossel, Elchanan and O’Donnell, Ryan},
  journal={SIAM Journal on Computing},
  volume={37},
  number={1},
  pages={319--357},
  year={2007},
  publisher={SIAM}
}

@article{wedin1972perturbation,
  title={Perturbation bounds in connection with singular value decomposition},
  author={Wedin, P.},
  journal={BIT Numerical Mathematics},
  volume={12},
  number={1},
  pages={99--111},
  year={1972},
  publisher={Springer}
}

@article{chen2018asymmetry,
  title={Asymmetry helps: Eigenvalue and eigenvector analyses of asymmetrically perturbed low-rank matrices},
  author={Chen, Y. and Cheng, C. and Fan, J.},
  journal={arXiv preprint arXiv:1811.12804},
  year={2018}
}

@article{alizadeh1997complementarity,
  title={Complementarity and nondegeneracy in semidefinite programming},
  author={Alizadeh, Farid and Haeberly, Jean-Pierre A and Overton, Michael L},
  journal={Mathematical programming},
  volume={77},
  number={1},
  pages={111--128},
  year={1997},
  publisher={Springer}
}

@article{pataki1998rank,
  title={On the rank of extreme matrices in semidefinite programs and the multiplicity of optimal eigenvalues},
  author={Pataki, G{\'a}bor},
  journal={Mathematics of operations research},
  volume={23},
  number={2},
  pages={339--358},
  year={1998},
  publisher={INFORMS}
}

@article{goemans1995improved,
  title={Improved approximation algorithms for maximum cut and satisfiability problems using semidefinite programming},
  author={Goemans, M. and Williamson, D.},
  journal={Journal of the ACM (JACM)},
  volume={42},
  number={6},
  pages={1115--1145},
  year={1995},
  publisher={ACM}
}

@article{leskovec2016snap,
  title={{SNAP}: A General-Purpose Network Analysis and Graph-Mining Library},
  author={Leskovec, J. and Sosi{\v{c}}, R.},
  journal={ACM Transactions on Intelligent Systems and Technology (TIST)},
  volume={8},
  number={1},
  pages={1},
  year={2016},
  publisher={ACM}
}

@article{singer2011angular,
  title={Angular synchronization by eigenvectors and semidefinite programming},
  author={Singer, A.},
  journal={Applied and computational harmonic analysis},
  volume={30},
  number={1},
  pages={20--36},
  year={2011},
  publisher={Elsevier}
}

@article{burer2003nonlinear,
  title={A nonlinear programming algorithm for solving semidefinite programs via low-rank factorization},
  author={Burer, S. and Monteiro, R.},
  journal={Mathematical Programming},
  volume={95},
  number={2},
  pages={329--357},
  year={2003},
  publisher={Springer}
}

@inproceedings{park2017non,
  title={Non-square matrix sensing without spurious local minima via the {Burer-Monteiro} approach},
  author={Park, D. and Kyrillidis, A. and Carmanis, C. and Sanghavi, S.},
  booktitle={Artificial Intelligence and Statistics},
  pages={65--74},
  year={2017}
}

@inproceedings{kyrillidis2011rank,
  title={Rank-deficient quadratic-form maximization over $M$-phase alphabet: {P}olynomial-complexity solvability and algorithmic developments},
  author={Kyrillidis, A. and Karystinos, G.},
  booktitle={2011 IEEE International Conference on Acoustics, Speech and Signal Processing (ICASSP)},
  pages={3856--3859},
  year={2011},
  organization={IEEE}
}

@article{barvinok1995problems,
  title={Problems of distance geometry and convex properties of quadratic maps},
  author={Barvinok, Alexander I.},
  journal={Discrete \& Computational Geometry},
  volume={13},
  number={2},
  pages={189--202},
  year={1995},
  publisher={Springer}
}

@misc{benhamou2019operatornormupperbound,
      title={Operator norm upper bound for sub-Gaussian tailed random matrices}, 
      author={Eric Benhamou and Jamal Atif and Rida Laraki},
      year={2019},
      eprint={1812.09618},
      archivePrefix={arXiv},
      primaryClass={math.PR},
      url={https://arxiv.org/abs/1812.09618}, 
}

@InProceedings{networkx,
  author =       {Aric A. Hagberg and Daniel A. Schult and Pieter J. Swart},
  title =        {Exploring Network Structure, Dynamics, and Function using NetworkX},
  booktitle =   {Proceedings of the 7th Python in Science Conference},
  pages =     {11 - 15},
  address = {Pasadena, CA USA},
  year =      {2008},
  editor =    {Ga\"el Varoquaux and Travis Vaught and Jarrod Millman},
}

@misc{gset_dataset,
  author =       {Ye, Yinyu},
  title =        {The gset dataset},
  url =   {https://web.stanford.edu/ yyye/yyye/Gset/.},
  year =      {2003},
}

@article{davis2011university,
  title={The University of Florida sparse matrix collection},
  author={Davis, Timothy A and Hu, Yifan},
  journal={ACM Transactions on Mathematical Software (TOMS)},
  volume={38},
  number={1},
  pages={1--25},
  year={2011},
  publisher={ACM New York, NY, USA}
}

@article{panxing2016genetic,
  title={PCI planning method based on genetic algorithm in LTE network},
  author={Panxing, LI and Jing, WANG},
  journal={Telecommunications Science},
  volume={32},
  number={3},
  pages={2016082},
  year={2016}
}

@article{gui2018bqp,
  title={PCI planning based on binary quadratic programming in LTE/LTE-a networks},
  author={Gui, Jihong and Jiang, Zhipeng and Gao, Suixiang},
  journal={IEEE Access},
  volume={7},
  pages={203--214},
  year={2018},
  publisher={IEEE}
}

@article{qiu2024ros,
  title={ROS: A GNN-based Relax-Optimize-and-Sample Framework for Max-k-Cut Problems},
  author={Qiu, Yeqing and Xue, Ye and Wang, Akang and Wang, Yiheng and Shi, Qingjiang and Luo, Zhi-Quan},
  journal={arXiv preprint arXiv:2412.05146},
  year={2024}
}

@article{ma2017moh,
  title={A multiple search operator heuristic for the max-k-cut problem},
  author={Ma, Fuda and Hao, Jin-Kao},
  journal={Annals of Operations Research},
  volume={248},
  number={1},
  pages={365--403},
  year={2017},
  publisher={Springer}
}

@article{barrett2022learning,
  title={Learning to solve combinatorial graph partitioning problems via efficient exploration},
  author={Barrett, Thomas D and Parsonson, Christopher WF and Laterre, Alexandre},
  journal={arXiv preprint arXiv:2205.14105},
  year={2022}
}

@article{tonshoff2022one,
  title={One model, any CSP: graph neural networks as fast global search heuristics for constraint satisfaction},
  author={T{\"o}nshoff, Jan and Kisin, Berke and Lindner, Jakob and Grohe, Martin},
  journal={arXiv preprint arXiv:2208.10227},
  year={2022}
}

@article{schuetz2022combinatorial,
  title={Combinatorial optimization with physics-inspired graph neural networks},
  author={Schuetz, Martin JA and Brubaker, J Kyle and Katzgraber, Helmut G},
  journal={Nature Machine Intelligence},
  volume={4},
  number={4},
  pages={367--377},
  year={2022},
  publisher={Nature Publishing Group UK London}
}

@article{abe2019solving,
  title={Solving np-hard problems on graphs with extended alphago zero},
  author={Abe, Kenshin and Xu, Zijian and Sato, Issei and Sugiyama, Masashi},
  journal={arXiv preprint arXiv:1905.11623},
  year={2019}
}

@article{punnen2022quadratic,
  title={The quadratic unconstrained binary optimization problem},
  author={Punnen, Abraham P},
  journal={Springer International Publishing},
  volume={10},
  pages={978--3},
  year={2022},
  publisher={Springer}
}

@article{forrester2008quadratic,
  title={Quadratic binary programming models in computational biology},
  author={Forrester, Richard John and Greenberg, Harvey J},
  journal={Algorithmic Operations Research},
  volume={3},
  number={2},
  year={2008}
}

@article{date2021qubo,
  title={QUBO formulations for training machine learning models},
  author={Date, Prasanna and Arthur, Davis and Pusey-Nazzaro, Lauren},
  journal={Scientific reports},
  volume={11},
  number={1},
  pages={10029},
  year={2021},
  publisher={Nature Publishing Group UK London}
}

@article{papailiopoulos2010maximum,
  title={Maximum-likelihood noncoherent OSTBC detection with polynomial complexity},
  author={Papailiopoulos, Dimitris S and Karystinos, George N},
  journal={IEEE transactions on wireless communications},
  volume={9},
  number={6},
  pages={1935--1945},
  year={2010},
  publisher={IEEE}
}

\end{document}